\documentclass[aps,reprint,amsfonts, amssymb, amsmath,  showkeys,pra, superscriptaddress, twocolumn,longbibliography,nofootinbib]{revtex4-1}

\usepackage{xcolor}

\usepackage{float}

\usepackage[shortlabels]{enumitem}

\usepackage{braket}
\usepackage{amsthm}
\usepackage{mathtools}
\usepackage{url}
\usepackage{physics}
\usepackage{graphicx}
\usepackage[left=16mm,right=16mm,top=35mm,columnsep=15pt]{geometry} 

\usepackage[T1]{fontenc}

\usepackage{bm}

\usepackage{silence}
\newcommand*{\eh}{\mathrm{End\, }(\mathcal{H})}

\def\HC{\mathcal{H}}

\def\LC{\mathcal{L}}

\def\ad{^{\dagger}}

\def\a{\alpha}
\def\w{\omega}

\newcommand{\fsnull}[1]{}
\newcommand{\old}[1]{}

\addtocontents{toc}{\protect\setcounter{tocdepth}{1}}

\usepackage{hyperref}
\usepackage[toc,page,header]{appendix}

\addtocontents{toc}{\protect\setcounter{tocdepth}{1}}

\usepackage[makeroom]{cancel}
\definecolor{C1}{RGB}{52, 89, 149}
\definecolor{C2}{RGB}{251, 77, 61}
\definecolor{C3}{RGB}{3, 206, 164}
\definecolor{C4}{RGB}{202, 21, 81}
\definecolor{C5}{RGB}{202, 21, 81}
\hypersetup{colorlinks=true, linkcolor=C5, citecolor=C5, urlcolor=C5}

\usepackage{tikz}
\tikzset{every picture/.style=remember picture}

\usepackage[utf8]{inputenc}
\usepackage{graphicx}
\usepackage{xcolor}
\usepackage{amsmath}
\usepackage{amsthm}
\usepackage{bm}
\usepackage{bbm}
\usepackage{comment}
\usepackage{mathdots}
\usepackage{lipsum}
\usepackage{verbatim}
\usepackage{natbib}
\usepackage{nccmath}
\usepackage{amsfonts}
\usepackage{thm-restate}
\usepackage{thmtools}
\usepackage{ytableau}

\newtheorem{fact}{Fact}

\usepackage{amssymb}
\usepackage{dsfont}

\newcommand{\poly}{\operatorname{poly}}

\newcommand{\DC}{\mathcal{D}}

\newcommand{\MC}{\mathcal{M}}

\newcommand{\OC}{\mathcal{O}}
\newcommand{\PC}{\mathcal{P}}

\newcommand{\VC}{\mathcal{V}}
\newcommand{\WC}{\mathcal{W}}

\newcommand{\Var}{{\rm Var}}

\renewcommand{\geq}{\geqslant}
\renewcommand{\leq}{\leqslant}

\renewcommand{\Im}{\text{Im}}

\newcommand{\spn}{{\rm span}}
\def\endh{\mathrm{End\ }\mch}

\newcommand{\ot}{\otimes}

\newcommand{\bs}{\textsf{BS}}

\newcommand{\lm}{\lambda }

\newcommand{\sg}{\sigma }

\newcommand{\Om}{\Omega }

\def\C{\mathbb{C}}
\newcommand{\mcl}{\mathcal{L}}
\newcommand{\mcu}{\mathcal{U}}
\newcommand{\mca}{\mathcal{A}}
\newcommand{\mcw}{\mathcal{W}}
\newcommand{\mcf}{\mathcal{F}}
\newcommand{\mco}{\mathcal{O}}

\newcommand{\mch}{\mathcal{H}}
\newcommand{\mcm}{\mathcal{M}}

\newcommand{\mcp}{\mathcal{P}}
\newcommand{\mcd}{\mathcal{D}}
\newcommand{\mce}{\mathcal{E}}
\newcommand{\mct}{\mathcal{T}}
\newcommand{\mcv}{\mathcal{V}}

\newcommand{\mbc}{\mathbb{C}}
\newcommand{\mbr}{\mathbb{R}}
\newcommand{\mbs}{\mathbb{S}}
\newcommand{\mbz}{\mathbb{Z}}
\newcommand{\mbe}{\mathbb{E}}

\def\VC{\mathcal{V}}
\def\PC{\mathcal{P}}
\def\MC{\mathcal{M}}

\newcommand{\mfh}{\mathfrak{h}}
\newcommand{\mfg}{\mathfrak{g}}
\newcommand{\mfe}{\mathfrak{e}}
\newcommand{\mff}{\mathfrak{f}}
\newcommand{\mfb}{\mathfrak{B}}
\newcommand{\mfu}{\mathfrak{u}}
\newcommand{\mfsu}{\mathfrak{su}}
\newcommand{\mfsp}{\mathfrak{sp}}
\newcommand{\mfso}{\mathfrak{so}}
\newcommand{\mfsl}{\mathfrak{sl}}
\renewcommand{\Im}{{\rm Im}}

\def\be{\begin{equation}}
\def\ee{\end{equation}}
\def\bs{\begin{split}}
\def\e{\end{split}}
\def\ba{\begin{eqnarray}}
\def\bea{\begin{eqnarray}}

\def\tea{\end{eqnarray}}
\def\ea{\end{eqnarray}}
\def\eea{\end{eqnarray}}
\def\w{\omega}

\def\SU{\text{SU}}
\def\SO{\text{SO}}

\def\U{\mathrm{U}}
\def\a{\alpha}
\def\b{\beta}

\def\a{\alpha}

\def\b{\beta}

\def\g{\mathfrak{g}}

\def\tn{^\otimes n}
\def\tk{^\otimes k}

\def\SU{\text{SU}}

\def\U{\mathrm{U}}
\def\a{\alpha}
\def\b{\beta}

\newtheorem{lemma}{Lemma}

\def\tn{^{\otimes n}}

\def\tk{^{\otimes k}}

\newcommand{\id}{\mathds{1}}

\renewcommand{\a}{\alpha}
\renewcommand{\b}{\beta}

\newcommand{\arr}{\xrightarrow[]{}}

\newcommand{\sdket}[1]{| #1 \rangle\!\rangle}
\newcommand{\sdbra}[1]{ \langle \!\langle #1|}

\newcommand{\sbraket}[2]{ \langle#1 | #2 \rangle}

\addtocontents{toc}{\protect\setcounter{tocdepth}{0}}

\def\SU{\textsf{SU}}
\def\SO{\textsf{SO}}
\def\O{\textsf{O}}
\def\SP{\textsf{SP}}

\def\U{\textsf{U}}
\def\SU{\textsf{SU}}

\def\triv{{\rm triv}}

\def\sg{\sigma}

\def\g{\gamma}
\def\SU{\textsf{SU}}

\def\U{\textsf{U}}
\def\triv{{\rm triv}}

\def\poly{{\rm poly}}

\newcommand{\mc}[1]{\mathcal{#1}}

\newcommand\mbb[1]{\mathbb{#1}}

\DeclareMathOperator*{\expect}{\mathbb{E}}

\def\be{\begin{equation}}
\def\te{\end{equation}}
\def\ee{\end{equation}}
\def\ba{\begin{eqnarray}}
\def\bea{\begin{eqnarray}}

\def\tea{\end{eqnarray}}
\def\ea{\end{eqnarray}}
\def\eea{\end{eqnarray}}

\newcommand{\floor}[1]{\left \lfloor #1 \right \rfloor }

\begin{document}

\makeatletter
\newif\ifinappendixtoc  

\newcommand{\tableofcontentsappendixonly}{%
  \begingroup
    \inappendixtocfalse

    \@ifundefined{l@section}{}{%
      \let\ao@l@section\l@section
      \def\l@section##1##2{\ifinappendixtoc \ao@l@section{##1}{##2}\fi}%
    }
    \@ifundefined{l@subsection}{}{%
      \let\ao@l@subsection\l@subsection
      \def\l@subsection##1##2{\ifinappendixtoc \ao@l@subsection{##1}{##2}\fi}%
    }
    \@ifundefined{l@subsubsection}{}{%
      \let\ao@l@subsubsection\l@subsubsection
      \def\l@subsubsection##1##2{\ifinappendixtoc \ao@l@subsubsection{##1}{##2}\fi}%
    }

    \def\AppendixTOCMark{\global\inappendixtoctrue}%

    \tableofcontents
  \endgroup
}
\makeatother

\title{
   Classical shadows with arbitrary group representations
}

\author{Maxwell West}
\affiliation{Theoretical Division, Los Alamos National Laboratory, Los Alamos, New Mexico 87545, USA}
\affiliation{School of Physics, University of Melbourne, Parkville, VIC 3010, Australia}

\author{Fr\'{e}d\'{e}ric Sauvage}
\affiliation{Quantinuum, Partnership House, Carlisle Place, London SW1P 1BX, United Kingdom}

\author{Aniruddha Sen}
\affiliation{Department of Computer Science, University of Texas at Austin, Austin, TX, 78712, USA}
\affiliation{Theoretical Division, Los Alamos National Laboratory, Los Alamos, New Mexico 87545, USA}

\author{Roy Forestano}
\affiliation{Physics Department, Institute for Fundamental Theory, University of Florida, Gainesville, FL, 32611, USA}

\author{David Wierichs}
\affiliation{Xanadu, Toronto, ON, M5G 2C8, Canada}

\author{\\Nathan Killoran}
\affiliation{Xanadu, Toronto, ON, M5G 2C8, Canada}

\author{Dmitry Grinko}
\affiliation{QuSoft, Amsterdam, The Netherlands}
\affiliation{Institute for Logic, Language and Computation, University of Amsterdam, Amsterdam, The Netherlands}
\affiliation{Korteweg-de Vries Institute for Mathematics, University of Amsterdam, Amsterdam, The Netherlands}

\author{M. Cerezo}
\affiliation{Information Sciences, Los Alamos National Laboratory, Los Alamos, New Mexico 87545, USA}
\affiliation{Quantum Science Center, Oak Ridge, TN 37931, USA}

\author{Mart\'{i}n Larocca}
\affiliation{Theoretical Division, Los Alamos National Laboratory, Los Alamos, New Mexico 87545, USA}
\affiliation{Quantum Science Center, Oak Ridge, TN 37931, USA}

\begin{abstract}
Classical shadows (CS) 
has recently emerged as an  important framework to efficiently predict  properties of an unknown quantum state. A common strategy in CS protocols is to parametrize the basis in which one measures the state  by a random group action; many examples of this have been proposed and  studied on a case-by-case basis. 
In this work, we present a unified theory that allows us to simultaneously understand  CS protocols based on sampling from general group representations, extending previous approaches 
that worked in simplified (multiplicity-free) settings. 
We identify a class of measurement bases which we call ``centralizing bases''  
that allows us to analytically characterize and invert the measurement channel, minimizing classical post-processing costs.
We complement this analysis by deriving general bounds on the sample-complexity necessary to obtain estimates of a given precision. 
Beyond its unification of previous CS protocols, our method
allows us to readily generate new protocols based on other groups, or different representations of previously considered ones. For example, we characterize novel shadow protocols based on sampling from  the spin and tensor representations of $\SU(2)$,   symmetric and orthogonal groups, and the exceptional Lie group $G_2$.

\end{abstract}

\maketitle

\section{Introduction}

Learning properties of unknown quantum states is a key component of quantum computation and information, but is significantly hindered by the exponential growth in their complexity as a function of system size. Indeed, this typically renders complete characterization intractable for even the relatively modestly sized systems amenable to experimental control today.
To combat this, recent protocols~\cite{huang2020predicting,aaronson2018shadow,cotler2020quantum,sugiyama2013precision,elben2020mixed,rath2021quantum,zhao2021fermionic,low2022classical,wan2023matchgate,paini2021estimating,van2022hardware,ippoliti2024classical,zhao2024group,king2024triply,jerbi2023shadows,koh2022classical,chen2021robust,hearth2024efficient,bertoni2024shallow,chan2022algorithmic,sauvage2024classical,sack2022avoiding,Helsen_2023,kunjummen2023shadow,denzler2023learningfermioniccorrelationsevolving,Wu_2024,west2024random,grier2024sample,brandao2020fast,vitale2024estimation,somma2024shadow,west2025real,bringewatt2025classical} have focused on efficiently learning a small subset of favoured observables, for example few-body~\cite{huang2020predicting} or symmetry-respecting~\cite{sauvage2024classical} operators. In particular, the framework developed in Ref.~\cite{huang2020predicting} allows for the prediction of many expectation values from ``classical shadows'' (CS) of a target state $\rho$, constructed via measurements made after a random action of some unitary ensemble. 
Given a desired set of properties, it is preferable to work with a protocol that minimizes the variance of the CS estimator; conversely, given a shadows protocol, a key part of the analysis is the determination of the set of observables that lead to small variances, and thus can be  efficiently estimated.

\begin{figure}[t]
  \includegraphics[width=.5\textwidth]{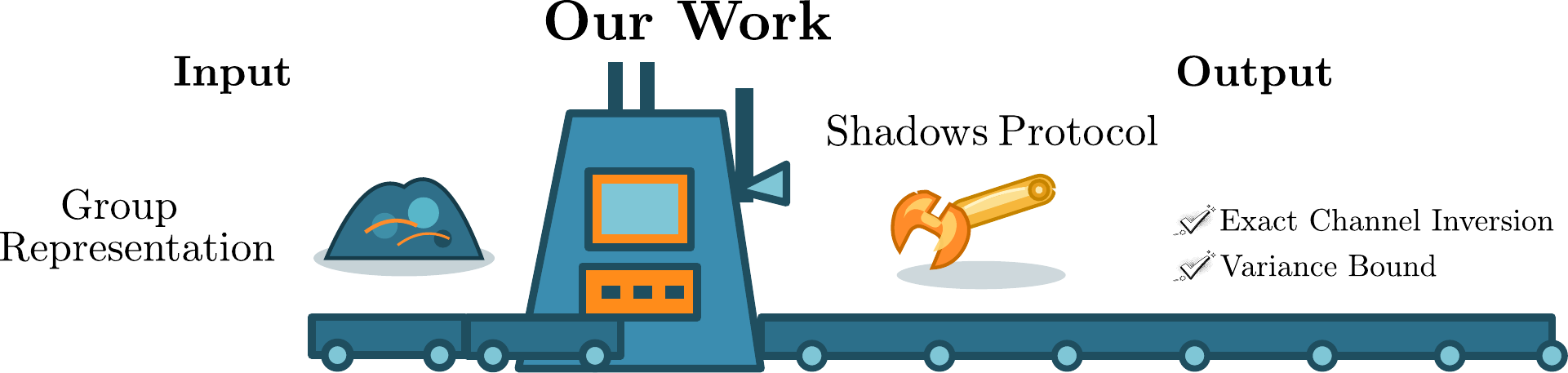}
  \caption{
Schematic overview of our framework. From any choice of a group representation
$R:G\to U(\mathcal H)$ (and a choice of commuting subgroup $H\subseteq G$
with a non-degenerate simultaneous eigenbasis) we obtain a classical-shadows
protocol such that a) the resulting measurement channel is
constant on the isotypic components of operator space,
$\MC=\sum_{\lambda} a_\lambda^{H} \PC_\lambda^{\VC}$ (see Theorem~\ref{thm:mc}), and therefore channel inversion
$\MC^{-1}=\sum_{\lambda}(a_\lambda^{H})^{-1} \PC_\lambda^{\VC}$ is exact, and b) the
sample complexity is controlled by the visibility factors $a_\lambda^{H}$ via
general variance bounds (see Theorem~\ref{thm:var}). In this sense, our framework turns group
representations into off-the-shelf shadows protocols.}
\label{fig:schematic}
\end{figure}
 
The use of group actions to generate random measurement bases lies at the heart of CS protocols, providing a powerful framework for quantum state estimation. Traditionally, groups such as the local and global Clifford groups~\cite{huang2020predicting} -- or more recently, matchgate circuits~\cite{zhao2021fermionic,wan2023matchgate} -- have been employed to construct these measurement ensembles, each analyzed within its own specialized context. While this case-by-case approach has proven effective, it complicates exploration and experimentation with different groups, as each requires a custom inversion strategy and corresponding variance bounds. An important step in a unifying direction was taken in Ref.~\cite{chen2021robust}, where 
it was shown that if a certain action of the group is multiplicity-free, then one can derive a general expression for the corresponding shadows measurement channel and invert it (a necessary component of the classical postprocessing) analytically, with no significant computational overhead.  
However, many groups of practical interest do not satisfy this assumption, considerably limiting the applicability of the existing analysis.

The primary contribution of this work is to show that it is possible to transcend this multiplicity constraint and derive analytically invertible channels for a large class of non-multiplicity-free representations of arbitrary compact groups $G$. Key to this is measuring in a (non-degenerate) simultaneous eigenbasis of an abelian subgroup $H$ of $G$. We find that different choices of $H$ lead naturally to shadows protocols that variously favor different sets of target observables, all for the same choice of representation of the parent group $G$.

We also derive a sequence of bounds on the variance of the resulting estimators that depend purely on the relation of the target observable $O$ to the $G$-module structures of $\mch$ and  $\eh$. 
Roughly speaking, we will see that observables that have an appreciable overlap with an exponentially large $G$-module will be prohibitively expensive to be learned accurately;
a second key result of ours is to quantify this phenomenon. Our work, schematically depicted in Figure~\ref{fig:schematic}, provides out-of-the-box access to a broad family of CS protocols, parametrized by a group, an abelian subgroup, and a representation.

Our framework encompasses and (when $G$ {and $H$} are chosen accordingly) reproduces  known results for schemes based on global and local Cliffords~\cite{huang2020predicting,ippoliti2024classical}, fermionic Cliffords~\cite{zhao2021fermionic,low2022classical,wan2023matchgate}, orthogonal~\cite{west2025real} and symplectic unitaries~\cite{west2024random}. 
These direct reductions occur as the aforementioned protocols involve (if implicitly) making measurements in a non-degenerate weight basis of a group whose adjoint action decomposes multiplicity-freely. For all of these examples, our formalism allows the measurement channels to be obtained with significantly less effort than in the original derivations. In the cases of $\SU(2)$, and the orthogonal group,   we furthermore obtain shadow protocols which differ from those obtained from those groups in previous literature~\cite{west2025real,liang2024real,sauvage2024classical,huang2020predicting}. In the former instance, we  derive a sample- and computationally-efficient method for estimating $n$-qubit permutation-invariant operators $O$, with a variance scaling as 
${\rm poly}\hspace{0.5mm}(n, \|O\|_\infty)$,
allowing us to efficiently learn  expectation values of \textit{any} permutation-invariant operator with spectral norm polynomial in the system size.

Beyond these examples, we develop novel shadow protocols based on the symmetric group and, more exotically, the exceptional Lie group $G_2$. As we shall see, our formalism makes such extensions fairly mechanical and straightforward for groups whose representation theory is well-understood.

\section{Main results}
Let us begin by establishing some notation. 
Given access to an unknown state $\rho$ acting on a Hilbert space $\mch$, a CS protocol applies a random unitary $U$ sampled from an ensemble $\mce\subset \U(\mch)$, followed by a measurement in a basis $\mcw=\{ \ket{w}\}_w$ of $\mch$ with (say) outcome $w$, and records a ``snapshot'', given by the pair $(w,U)$. We will assume that $\mce$ is a 3-design~\cite{mele2024introduction} over a group $G$, and consider the effect of various choices of groups.  To produce an estimate for the expectation value $o=\Tr[\rho O]$ of some observable $O$, the shadows protocol collects several pairs, and averages the single-shot estimates $\hat{o}_{w,U} = \Tr[  \MC^{-1}( U^\dagger \ketbra{w} U)\, O]$, where 
\begin{equation}\label{eq:mcinit}
\mcm  = \expect_{U\sim \mce} {\rm Ad}_{U^\dagger}\circ \mca_\mcw \circ {\rm Ad}_{U}
\end{equation}
denotes the so-called  \textit{measurement channel}\footnote{Such details of CS are reviewed in Appendix~\ref{sec:cs}.}. Here $\mca_\mcw(-)=\sum_w \bra{w}(-)\ket{w}\ketbra{w}$ is a dephasing channel with respect to $\mcw$ and ${\rm Ad}_U(-) = U(-)U\ad$ the adjoint action of $U$.  That is, $\MC$ is simply the average of the action of the unitary ensemble on $\mca_\mcw$. Note that the expression for the estimates implies that being able to invert $\MC$ is fundamental for a shadows protocol to work. On occasion, when the emphasis drawn to the basis dependence of $\mcm$ seems worth the additional notational clutter, we will write $\mcm_\mcw$ for the measurement channel associated to the basis $\mcw$.

Now, suppose $G$ is a compact group with a unitary action on $\HC$, that is, with a group homomorphism\footnote{We can relax this assumption and assume that $R$ is a projective representation of the group $G$.} $R:G\arr \U(\HC)$. The $G$-module $\HC$ decomposes under this action as
\begin{equation}
\mch\cong \bigoplus_{\eta \in \widehat{\mch} }\bigoplus_{i=1}^{m_{\eta}^\HC} \,\,\HC^{\eta,i} = \bigoplus_{\eta \in \widehat{\mch}} \mbb{C}^{m_{\eta}^\HC} \ot V^{\eta}_G\,,
\end{equation}
where we use $\HC^{\eta,i} \subseteq \HC$ to denote $G$-submodules isomorphic to $V^\eta_G$, the irreducible $G$-module labelled by $\eta$, appearing with multiplicity $m_{\eta}^\HC$, and $\widehat{\mch}$ is the set of all irreps appearing in $\mch$. For a $G$-module $\HC$, we will write $\HC_\eta\coloneqq\bigoplus_{i=1}^{m_{\eta}^\HC} \,\,\HC^{\eta,i}\subset \HC$ to denote the so-called \textit{maximal $\eta$-isotypic component} or simply $\eta$-\textit{isotypic}, the submodule containing all irreducibles of the same type $\eta$ in the module $\HC$.

Next, let $\LC:=\eh$ denote the linear operators on $\HC$, which likewise constitute a $G$-module and thus decompose under $G$\footnote{We will typically use $\eta,\zeta$ for $G$-irreps appearing in $\mch$, and $\lm,\mu$ for irreps appearing in $\mcl$. It is traditional to write $\widehat{G}$ to denote the set of all inequivalent $G$-irreps.} into irreps,
\begin{equation}
    \mcl \cong \bigoplus_{\lm \in \widehat{\LC}} \bigoplus_{i=1}^{m_\lm^\LC} \LC^{\lm,i}= \bigoplus_{\lm \in \widehat{\LC}} \mbb{C}^{m_\lm^\LC}\ot V^\lm_G,\label{eq:ldecomp}
\end{equation}
where $\widehat{\LC}$ is the set of all irreps appearing in $\LC$.
Recalling that $\mce$ is taken to be a 3-design over $G$, the explicit twirl in Eq.~\eqref{eq:mcinit} makes it evident that $\MC$ is a $G$-equivariant linear operator on $\LC$, i.e., $\mcm(R(g)\rho R(g)\ad)=R(g)\mcm(\rho) R(g)\ad$ for all $g\in G$; this implies, for example, that  both its image and kernel constitute $G$-modules. They are especially relevant in our context since they correspond, respectively, to the subsets of operators that can and cannot be estimated unbiasedly by shadows. We will use $\LC^\VC \coloneqq {\rm Im}(\MC_\WC)$ to denote the former, the so-called \textit{visible space}, which is equivalently~\cite{van2022hardware} defined as:
\begin{equation}\label{eq:visspace}
\LC^\VC = \spn \Big\{ R(g) \ketbra{w}{w}  R(g)\ad\Big\}_{g,w}  \,.
\end{equation}
Note that
\begin{equation}
    \LC^\VC \cong \bigoplus_{\lm \in \widehat{\VC}} \mbb{C}^{m_\lm^{\VC}}\ot V^\lm_G,
\end{equation}
where $\widehat{\VC}$ is the set of irreps appearing in $\LC^\VC$ with multiplicity $m_\lm^{\VC} \leq m_\lm^{\LC}$ (since $\LC^\VC$ is a $G$-submodule inside $\LC$).
As $\LC^\VC$ is patently a $G$-module, let us write $\LC^\VC_\lm \subseteq \LC^\VC$ for its isotypic components, and $\PC^\VC_\lm$ for the corresponding projectors. That is, $\mcp_\lm^\mcv$ projects onto the isotypic $\lm$ within the space $\LC^\VC$, which may be strictly contained within the full $\lm$ isotypic of $\mcl$.

We make the following important definition:

\begin{restatable}{defn}{deffb}\label{def:fb}
We call $\WC$ a Fourier basis (FB) if it is
adapted to the irrep decomposition of $\HC$, that is, if it is of the form $\WC=\bigcup_{\eta,i} \WC^{\eta,i}$ with $\WC^{\eta,i}$ a basis for each irrep $\HC^{\eta,i}\subseteq \HC$.
\end{restatable}
It is not hard to see that when measuring in a FB, Eq.~\eqref{eq:visspace} forces the visible space $\LC^\VC$ to be constrained to the \textit{diagonal} operator submodule $\LC^\DC \coloneqq \bigoplus_{\eta,i} {\rm End}(\HC^{\eta,i})$. However, we will soon see this does not necessarily mean all operators are visible, as the  inclusion $\LC^\VC \subseteq \LC^\DC$ may be strict. Also note that $\LC^\DC \subseteq \LC$; in particular, if $\HC$ is irreducible then $\mcl^\DC=\mcl$.

Now, by virtue of the $G$-equivariance of $\MC$, one can (using Schur's lemma~\cite{fulton1991representation}) always write 
\begin{equation}
    \mcm \cong \bigoplus_{\lm \in \widehat{\VC}} \MC_\lm \ot \id_{d_\lm}\,,
\end{equation}
with $\MC_\lm \in {\rm End}(\mbb{C}^{m_\lm^\VC})$, viewed as embedded into ${\rm End}(\mbb{C}^{m_\lm^\LC})$. That is, the measurement channel acts non-trivially only on the multiplicity spaces of the decomposition Eq.~\eqref{eq:ldecomp}. Evidently, the extent to which this observation actually simplifies things for a given group depends on the magnitude of the corresponding $m_\lm^\VC$. For a multiplicity-free decomposition (i.e., with all $m_\lm^\VC=1$) the channel becomes amenable to analytic characterization; in particular its (pseudo-)inversion is trivial if one can  calculate the  $\MC_\lm$ (which in this case would be scalars). There are natural cases of interest, however, for which the decomposition is not multiplicity-free. For example, in the case of the diagonal action of $\SU(2)$ on a system of $n$ qubits one finds (see Appendix~\ref{sec:appdetails}) that the multiplicities can be of order $\exp (n)$. Generically  this will result in the channel being composed of exponentially large dense matrices, the inversion of which would be intractable\footnote{Recall that calculating the (pseudo-)inverse of $\mcm$ is a crucial part of the CS procedure (see Appendix~\ref{sec:cs} for more discussion on this point).}. Indeed, difficulties related to inverting large (if polynomial-sized) matrices have previously been encountered in the CS literature~\cite{sauvage2024classical}.

Certainly, however, these complications would vanish if each $\MC_\lm$ were proportional to the identity $\id_{m_\lm^\mcv}$,
corresponding to a $\MC$ of the form $\MC=\sum_{\lm \in \widehat{\VC}} a_\lm \PC_\lm^\VC$ for some coefficients $a_\lm$.
Remarkably, as we now discuss, it is often possible to engineer such a feature by a suitable choice of $\WC$. To that end, 
we begin with the following definition. 

\begin{restatable} 
{defn}{defnlmgood}\label{def_lmgood}
A basis $\mcw$ of $\HC$ is called ``$\lm$-centralizing'' if the corresponding channel $\MC$ acts as a scalar on the isotypic $\LC^\VC_\lm$,
\[
\MC \big|_{\LC^\VC_\lm} = a_\lm \PC_\lm^\VC
\]
for some $a_\lm \in \mbb{R}$.
\end{restatable}
\noindent
The value of $\lm$-centralizing bases is that, if one is interested in some $O\in \LC^\VC_\lm$, then one can trivially invert the action of the measurement channel on $O$; indeed $\mcm^{-1}(O)=a_\lm^{-1}O$.
As a first example, we have
\begin{restatable} 
{lemma}{lemfgood}\label{lem:fgood}
Any FB is trivial-centralizing. Furthermore, $a_\triv=1$.
\end{restatable}
In fact, one can show (see Appendix~\ref{sec:gshaddetails}) that all such \textit{invariants} (that is, operators in the isotypic corresponding to the trivial action of $G$) which live in the visible space of the protocol commute, so that one does not really need shadows to simultaneously estimate them; the content of Lemma~\ref{lem:fgood} is then that their simultaneous eigenbasis is given by (any) FB. 
Next we define:

\begin{restatable}{defn}{defngood}\label{def_good}
A basis $\mcw$ of $\HC$ is called ``centralizing'' if it is $\lm$-centralizing for each $\lm \in \widehat{\VC}$, i.e., if the channel $\MC$ has the form
\begin{equation}\label{eq:gb}
    \MC = \sum_{\lm \in \widehat{\VC}} a_\lm \PC_\lm^\VC\,.
\end{equation}
\end{restatable}
The reason why we call such $\WC$ ``centralizing'' is that the corresponding $\MC$ are central in the algebra of  $G$-invariant linear operators from $\mcl^\mcv$ to $\mcl^\mcv$.
The value of centralizing bases is evident: we can trivially invert $\MC$ over its image: $\MC^{-1}=\sum_{\lm \in \widehat{\VC}} a_\lm^{-1}  \PC^\VC_\lm$.
We are able to construct centralizing bases for a broad family of group representations, by making use of the following notion:
\begin{restatable}{defn}{defcse}\label{def:cse}
We call $\WC$ a ``commuting subgroup eigenbasis'' (CSE) if both  a) it is a FB as per Definition~\ref{def:fb}, and b) each $\WC^{\eta,i}$ is a simultaneous eigenbasis for an abelian subgroup $H\subseteq G$.
We say a CSE is ``non-degenerate'' (ND), and use NDCSE for short, if for every $\eta \in \widehat{\HC}$ the restriction of each $G$-module $\HC^{\eta,i}$ to the corresponding abelian subgroup $H$ is multiplicity-free.
\end{restatable}
We remark that one could, in principle, choose different subgroups $H^{\eta,i}$ for different irreps $\HC^{\eta,i}$, in which case all of our results would generalize smoothly. However, to keep the notation under control (and because none of our examples require such fine-grained control) we will henceforth assume for simplicity a global choice of $H$. We will, however, occasionally consider shadow protocols obtained from different global choices of \(H\) for the same parent group \(G\). 
In this context, we can prove the following result:

\begin{restatable}{thm}{thmmc}\label{thm:mc} 
Any NDCSE is centralizing, and thus leads to $\MC=\sum_{\lm \in \widehat{\VC}} a_\lm^H \PC^\VC_\lm$. 
Moreover, for every $\lm \in \widehat{\VC}$ each $a_\lm^H$ can be analytically described: $a_\lm^H = d^H_\lm/d_\lm$ with $d_\lm=\dim(V^\lm_G)$ and $d^H_\lm=\dim( (V^\lm_G)^H)$, the dimension of the $H$-invariant subspace in $V^\lm_G$,
with $H\subseteq G$ the commuting subgroup defining the NDCSE.
\end{restatable}

The situation is depicted in Fig.~\ref{fig:1}.
Now, having analytical access to the scalars $a_\lm^H$, we can straightforwardly invert $\MC$ over its image: 
\begin{equation}
    \mcm^{-1} =  \sum_{\lm \in \widehat{\VC}} \frac{1}{a_\lm^H} \mcp_\lm^{\VC}\,.
\end{equation}
Previously, we mentioned that if $\mcw$ is a FB, then the visible space is contained in $\LC^\DC$. We can now refine our identification of the visible operators in a non-degenerate CSE protocol: They are those diagonal isotypics $\LC^\DC_\lm$ with $a_\lm^H > 0$
\begin{equation}
\LC^{\VC} = \bigoplus_{\lm \,:\, a_\lm^H > 0 } \LC^\DC_\lm \subseteq \LC^\DC \,.
\end{equation}
For example, an irreducible $\HC$ combined with  measurements with respect to an NDCSE may lead to invisible operators if certain $\lm$ appearing in $\LC$ have no $H$-invariants, i.e., if $d^H_\lm=a_\lm^H=0$. A particularly striking example of this occurs when taking $G$ to be the set of $n$-qubit Pauli strings (with phases in $\{\pm 1,\,\pm i\}$), acting in the usual way on $\mch=(\mbc^2)\tn$ (this turns out to be an irreducible action). Note that the computational basis is an NDCSE for $H=\{ Z^\a\}_{\a\in\{0,1\}^n}$. In this case $\LC^\DC=\LC$ decomposes into $4^n$ one-dimensional irreps, each spanned by a single Pauli string $P$; these irreps have $a_{P}= 0$ if $P$ is not composed of only $I$ and $Z$, and $a_{P}= 1$ otherwise (see Appendix~\ref{sec:appdetails}). Thus, we have an example where the image of $\mcm$ is a \textit{strict} subspace of $\mcl^\DC$. 

More generally, operator isotypics with small $a_\lm^H$ (those corresponding to irreps where the $H$-invariant subspace is $d_\lm^H\ll d_\lm$) will be suppressed by $\mcm$, suggesting observables therein will be more difficult to reconstruct, being ``less visible'' in a sense that we now make precise. Recalling the direct correlation between the sample-complexity of a CS protocol and the variance of its estimators~\cite{huang2020predicting}, the preceding intuition is characterized by the following theorem.

\begin{restatable}{thm}{thmvar}\label{thm:var}
In the context of an NDCSE protocol, let $O\in\LC$ with components $O^\lm $ in each visible $\LC^\DC_\lm$. Given a target $o=\Tr[\rho O]$, the variance of the single-shot estimator $\hat{o}=\Tr[\hat{\rho} O]$ satisfies
\begin{equation}
\Var[\hat{o}] \leq \sum_{\lm} \frac{\|O^\lm\|^2_2}{a_\lm^H} \leq \|O\|^2_2 \max_{\lm\hspace{0.5mm}:\hspace{0.5mm}O^\lm\neq 0} \frac{1}{a_\lm^H},
\label{eq:2nb}         
\end{equation}
and additionally
\begin{equation}
\Var[\hat{o}] \leq  \left\|   \sum_{\lm} \frac{O^{\lm} }{a_\lm^H}  \right\|_\infty^2. \label{eq:inb}
\end{equation}
\end{restatable}
Theorem~\ref{thm:var} guarantees that an operator $O$ which lives in a subspace with $1/a_\lm^H \in \mco(\poly (n))$ can be estimated with a polynomial number of samples if additionally $\|O\|_\infty\in \mco(\poly\, n)$. Conversely, to accurately estimate expectation values of operators with $1/a_\lm^H\in \Om(\exp(n))$ may be intractable (note that this is not guaranteed, as the bounds of Theorem~\ref{thm:var} can be loose). As we discuss in Appendix~\ref{sec:gshaddetails}, these bounds can be tightened under various additional assumptions on the operators being measured. The proofs of the above theorems may also be found in Appendix~\ref{sec:gshaddetails}.

Now, since upon restriction, degeneracy can only increase, when looking for NDCSEs it is best to pick $H$ ``maximally\footnote{That is, not properly contained in an abelian subgroup of $G$.}''. {In general, however, a group may have many different maximal abelian subgroups, leading to shadow protocols with the same group $G$ but different measurement bases. Sometimes we have a natural candidate: For example,} if $G$ is a Lie group, there is an obvious choice of maximal commuting $H$: the \textit{Cartan subgroup} (see Appendix~\ref{sec:wc}). In this context, the CSE is the corresponding \textit{weight-basis}~\cite{fulton1991representation} and the shadows protocol is promised to be centralizing when weights on each $G$-irrep $V^\eta$ are non-degenerate.
More generally, one could consider commuting \textit{subalgebra} eigenbases, which are defined as in Definition~\ref{def:cse} except with respect to a commuting subalgebra of the group algebra $\mbc[G]$~\cite{fulton1991representation}. The well-known Gelfand-Tsetlin bases
~\cite{vershik2005new,molev2006gelfand} furnish examples of non-degenerate commuting subalgebra eigenbases. In contrast to Theorem~\ref{thm:mc}, on the other hand, we find
\begin{restatable}{lemma}{lemaba}\label{lem:aba}
 Non-degenerate commutative subalgebra eigenbases are not necessarily centralizing.
\end{restatable}
For example, every irrep of the symmetric group $S_n$ possesses a Gelfand-Tsetlin basis~\cite{vershik2005new}, but we show   in Appendix~\ref{sec:snshad} that there exist  $S_n$ irreps for which the Gelfand-Tsetlin basis is not  centralizing.

\begin{table}
\centering
\begin{tabular}[b]{l ccccc}
\hline
Group $G$& $\mch$ & $H$ & $\dim(\LC^\VC)$  & $\#\lambda$ &\hspace{-2mm} Mult. free    \\
\hline
${\rm Cl}(n)$~\cite{huang2020predicting} & $(\mbb{C}^2)\tn$  & $\mbz_2^n$ & $4^n$& $2$ & Yes \\
${\rm Cl}(1)^{\times n}$~\cite{huang2020predicting} & $(\mbb{C}^2)\tn$  & $\mbz_2^n$ & $4^n$& $2^n$ & Yes \\
$\SU(d)$ [New] &  ${\rm Sym}^t(\mbb{C}^d)$  & $  \U(1)^{\times (d-1)}$ & $\binom{d+t-1}{t} ^2 $& $t+1$ & Yes \\
$\SU(2)$ [New] & $\mbb{C}^{2S+1}$ & $\U(1)$  & $(2S+1)^2$ & $2S+1$ & Yes\\
$\SU(2)$ [New] & $(\mbb{C}^2)\tn$ & $\U(1) $  & $\mco(\exp(n))$ & $n+1$  & No\\
$\SO(2n)$~\cite{zhao2021fermionic}  & $(\mbb{C}^2)\tn$ &$ \mbz_2^{n}$&$2^{2n-1}$ &$2n+2$ &No   \\
$\O(d)$~\cite{west2025real} & $\mbb{C}^d$ & $\mbz_2^d$ & $d(d+1)/2 $& 3&Yes \\
$\O(d)$ [New] & $\mbb{C}^{d}$ & $\U(1)^{\times d/2}$ & $d^2$& 3&Yes \\
$\SP(d)$~\cite{west2024random}  & $\mbb{C}^d$& $\U(1)^{\times d}$ &$d^2$ & 3&Yes \\ 
$S_n$  [New]& $\mbb{C}^n$ & $\mbz_n$&\hspace{-3mm}  $n^2-3n+3$& 7 & No \\
$G_2$  [New]& $\mbb{C}^7$ & $\U(1)^{\times 2}$& 49& 4& Yes \\
\hline
\end{tabular}
\caption{Summary of different instantiations of the results of Theorems~\ref{thm:mc} and~\ref{thm:var}. For each of the protocols addressed we indicate the reference, if already treated in the literature, or if it is new ($1$\textsuperscript{st} col). We further provide the Hilbert space on which it acts ($2$\textsuperscript{nd} col),  the NDCSE $H$ ($3$\textsuperscript{rd} col), the dimension $\dim(\LC^\VC)$ of the visible space ($4$\textsuperscript{th} col), and the number $\#\lambda$ of irreps appearing in the decomposition of $\LC$ ($5$\textsuperscript{th} col). For the latter, we further report if the decomposition is multiplicity-free  ($6$\textsuperscript{th} col.). }
\label{table}
\end{table}

\begin{figure}[t]
  \includegraphics[width=0.43\textwidth]{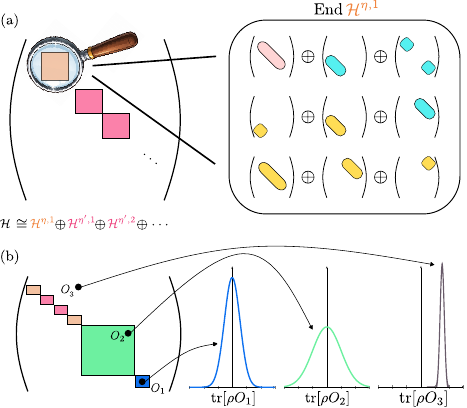}
  \caption{
(a) The subspace $\mcl^D$ of \textit{$G$-diagonal operators} are those which map the irreducible $G$-modules $\HC^{\eta,i} \subseteq \mch$ to themselves; these operator spaces  are generally reducible and decompose into irreps.  
  (b) An NDCSE shadows protocol produces unbiased estimates for the expectation values of the operators within the visible space   $\mcl^\mcv\subseteq\mcl^D$. For a given operator, the variance of the corresponding estimator (and therefore sample-complexity of the protocol) is a function of the dimension of the irrep to which the operator belongs (see Theorem~\ref{thm:var}). }
\label{fig:1}
\end{figure}

\section{Applications}

In this section we briefly discuss the application of the above theorems to several  examples. A summary of our results is presented in Table~
\ref{table}. Further details, and further examples, may be found in Appendix~\ref{sec:appdetails}.

\textit{Global Cliffords.}
Perhaps the canonical example of a CS protocol is given by sampling uniformly from the $n$-qubit global Clifford group\footnote{Equivalently, by virtue of the Clifford group being a unitary 3-design, the full unitary group $\U(d)$.} ${\rm Cl}(n)$, and measuring in the computational basis~\cite{huang2020predicting} of $\mch=(\mbc^2)\tn$. We begin by noting that the computational basis is indeed an NDCSE (with respect, for example, to the subgroup $H_{{\rm Cl}(n)}=\{I,Z\}\tn\subset{\rm Cl}(n)$). Now, as we discuss in Appendix~\ref{sec:appdetails}, the Cliffords act irreducibly on $\mch$, from which we conclude $\mcl=\mcl^\DC$; that is, in this case \textit{all} operators are  $G$-diagonal, and therefore potentially\footnote{At this point, it remains to be seen that all of the $a_\lm^H$ of Theorem~\ref{thm:mc} are non-zero; this will however turn out to be true.} visible. Next, we would like to understand the structure of $\mcl$ as a $G$-module. For global Cliffords, one obtains~\cite{zhu2016clifford,helsen2018representations} the decomposition $\mcl\cong\mcl^0\oplus\mcl^{\rm ad}$ into
a single one-dimensional (trivial) irrep $\mcl^0$ (spanned by the identity) and a second, $(d^2-1)$-dimensional irrep, the \textit{adjoint representation} of $\mfsu(d)$. In the notation of Theorem~\ref{thm:mc}, $a_0=1,$ $a_{\rm ad}=(d-1)/(d^2-1)=1/(d+1)$; to see the latter case note that the adjoint action of $H_{{\rm Cl}(n)}$ on $\mcl^{\rm ad}$ is trivial exactly on the $(d-1)$-dimensional subspace spanned by $\{I,Z\}\tn\setminus\{I\tn\}$. From Theorem~\ref{thm:mc} we then see $
\mcm(\rho) =
(a_{0}\mcp_0 + a_{\rm ad}\mcp_{\rm ad})(\rho)= (d+1)^{-1}( \rho+\id\Tr[\rho]).$
This reproduces the result of  Ref.~\cite{huang2020predicting}, without necessitating the use of the Weingarten calculus~\cite{mele2024introduction} employed in that work. 

\textit{Local Cliffords.} In the case ${\rm Cl}(1)^{\times n}$ of local Cliffords (which again act irreducibly on $\mch=(\mbc^2)\tn$), the irreps of $\mcl=\mcl^\DC$ are in bijection with the choices of qubit subsets $S\subseteq [n]$; each such irrep $\mcl^S$ is the span of the Pauli strings that act non-trivially exactly on the qubits in $S$. We can take $H_{\rm comp}=\{I,Z\}\tn$  to conclude that the computational basis is again centralizing. As each $\mcl^S$ has a single diagonal Pauli string (acting as $Z$ on each of the selected qubits, and the identity elsewhere),  we can immediately see {$a_S^{H_{\rm Comp}}=3^{-|S|}$,} so that 
$\mcm  = \sum_{S\subset [n]}3^{-|S|} \mcp_S,$
with $\mcp_S$ the projector onto $\mcl^S$. We conclude that the local Pauli measurement channel suppresses operators exponentially in their support. {This reproduces another result of Ref.~\cite{huang2020predicting}. 

More interestingly, let us consider (taking $n$ to be even, say) the (patently abelian) subgroup $H_{\rm Bell}$ generated by the set $\{X_{2i-1}X_{2i},\, Z_{2i-1}Z_{2i}\}_{i=1}^{n/2}$; an NDCSE with respect to this choice of subgroup is given by measuring in the ``local Bell basis'' $\{(\ket{00}\pm \ket{11})/\sqrt{2},\, (\ket{01}\pm \ket{10})/\sqrt{2}\}\tn$. The irreps of ${\rm Cl}(1)^{\times n}$ in $\mcl$ remain (of course) given by the $\mcl^S$, but the coefficients  $a^{H_{\rm Bell}}_S$ differ significantly from their counterparts in the previous example. For example, consider $S=\{1,2\}$. We have a non-trivial invariant subspace of $\mcl_{\{1,2\}}$ spanned by $\{X_1X_2,Y_1Y_2,Z_1Z_2\}$, so that $a^{H_{\rm Bell}}_{\{1,2\}}=3/9$ (c.f. $a^{H_{\rm Comp}}_{\{1,2\}}=1/9$). More generally, we see that
\begin{equation}
    a^{H_{\rm Bell}}_S = \begin{cases}
        3^{-t}, & S \text{\ a union of $t$ paired sites}\\
        0,& \text{otherwise}
    \end{cases}\ .
\end{equation}
Note that in the case where our target observable is supported exactly on the union of $t$ of the sites paired by Bell states we have improved significantly from $a_S^{H_{\rm comp}}=3^{-2t}$. The price we pay for this, however, is that observables which have support only on half of one of the Bell states have become \textit{invisible}. This is not the last we will see of the trade-off between improving sample complexity for one set of observables, while sacrificing it for another. This example reproduces the results of Ref.~\cite{ippoliti2024classical}. 
}

\textit{Matchgate circuits.}
We next consider the case where $G$ is the group of \textit{fermionic Gaussian unitaries} (FGU), which (in the so-called matchgate circuit representation) has received attention in the CS literature as an appropriate ensemble for measuring low-degree fermionic observables~\cite{zhao2021fermionic,low2022classical,wan2023matchgate,diaz2023showcasing,braccia2025optimal}. We show in Appendix~\ref{sec:appdetails} that in the matchgate circuit representation
the computational basis (which is used in the existing fermionic shadow schemes~\cite{zhao2021fermionic}) is a centralizing basis. As in the Clifford case, employing Theorem~\ref{thm:mc} substantially reduces the technical difficulty of obtaining the channel. The FGU ensemble gives us our first example of a \textit{reducible} action on $\mch$; specifically we have $ \mch = \mch^+ \oplus \mch^-,$
where $\mch^{\pm}$ are respectively given by the span of the computational basis vectors with even or odd Hamming weight (or parity). It follows that $\mcl^\DC\subsetneq\mcl$ is a strict subspace of $\mcl$, and that the protocol is not tomographically complete.
The fact that certain operators lie outside the visible space has previously been physically interpreted as corresponding to a  restriction imposed by parity superselection~\cite{StreaterWightman2001}; 
we show in Appendix~\ref{sec:appdetails} that these
unphysical operators are exactly those that map $\mch^\pm\mapsto\mch^\mp$, and are therefore not $G$-diagonal. As is well-known, under the matchgate group $\mcl$ decomposes into modules given by the span of \textit{Majorana monomials} of a given degree~\cite{diaz2023showcasing}; denoting the span of the $k$\textsuperscript{th} order monomials by $\mcl_k$, we show in Appendix~\ref{sec:appdetails} that for an $n$ qubit system 
$a_{2k}=\binom{n}{k}/\binom{2n}{2k} $, whence
Theorem~\ref{thm:mc} immediately characterizes the shadow protocol. \\

\textit{$\SU(2)$: spin $J$ irrep.}
As a simple first example of acting with $\SU(2)$, let us consider the task of learning properties of a spin-$J$ particle{; that is, our  Hilbert space is the $(2J+1)$-dimensional irrep of $G=\SU(2)$. Let us take as our measurement basis the eigenstates of definite spin in the $Z$-direction; that is, the NDCSE with respect to the maximal torus \(H_Z = \{\exp(-i\theta J_Z): \theta \in \mathbb{R}\}\cong \U(1)\)}. 
By the self-duality of $\SU(2)$ irreps~\cite{fulton1991representation} and the usual rules for the addition of angular momentum we have $\mcl^\DC\cong\mcl = \bigoplus_{j=0}^{2J}\mcl^j$.
The irrep $\mcl^j$ (of spin $j$) has  dimension $2j+1$ and a 1-dimensional {$H_Z$-invariant subspace corresponding to the weight-zero line; for this irrep we therefore have $a_j=1/(2j+1)$.
More generally,  it is in fact not hard to see that the irreducible modules $\HC= {\rm Sym}^t(\mbb{C}^d)$ of $G=\SU(d)$, the natural generalization of the spin irreps of $\SU(2)$, are also multiplicity-free\footnote{For any composition $\mu$ of $t$, the dimension of the weight $\mu$ subspace in ${\rm Sym}^t(\mbb{C}^d)$ is given by the \textit{Kostka number} $K^{\mu}_{[t]}$,  the number of trivial $S_t$ irreps in the permutation module induced from the Young subgroup $S_\mu\subseteq S_t$. This number is $1$ for all $\mu$.} and thus their weight-bases are centralizing bases.

\textit{$\SU$(2): tensor rep.}
As our next example we consider the natural   $\SU(2)$-representation $Q$ on $(\mbc^2)^{\otimes n}$ given by
$Q(U)=U^{\otimes n}. $
This representation is far from being multiplicity-free (see Appendix~\ref{sec:appdetails}), and so failure to measure in a centralizing basis can lead to the need for costly numerical matrix inversion~\cite{sauvage2024classical}. We show in Appendix~\ref{sec:appdetails}, however, that the \textit{Schur basis}~\cite{bacon2006efficient,bacon2007quantum,kirby2017practical,krovi2019efficient,burchardt2025high} is centralizing. The construction of a bespoke variance bound (which holds in addition to the generic variance bounds of Theorem~\ref{thm:var}, see Appendix~\ref{sec:su2}) {and a characterization of the visible space} leads to
\begin{restatable}{thm}{sutwovar} \label{thm:su2var}
{For any permutation-invariant operator $O$, $\SU (2)$-shadows produces unbiased estimators with variance bounded as}
\begin{equation}
\Var[\hat{o}]\leq \frac 43 (n+1)^4\|O\|_\infty^2    
\end{equation}
when $\SU(2)$ acts via the $n$-qubit tensor representation.
\end{restatable}
That is, we are able to sample- and computationally-efficiently estimate any permutation-invariant operators whose spectral norm does not increase superpolynomially with system size (some numerical results attesting to this are presented in Appendix~\ref{sec:su2}), an important case with widespread applications throughout both quantum information theory and physics more generally~\cite{schatzki2022theoretical,sauvage2024classical,toth2010permutationally,anschuetz2022efficient}.

\textit{The symmetric group.} As our final example in the main text (more examples may be found in Appendix~\ref{sec:appdetails}) let us consider the permutation representation $\HC=\mbb{C}^n$ of the symmetric group $S_n$. By taking $H=\langle (12\cdots n)\rangle \cong \mbz_n$, the subgroup generated by an $n$-cycle, we can readily see that the corresponding CSE, {the Fourier basis} $\WC=\{ \ket{\chi_k} \}_{k}$ with $\ket{\chi_k} \coloneqq \sum_j e^{2\pi ijk/n} \ket{j}$, is non-degenerate. Hence, we can obtain an $S_n$-shadow protocol inheriting the properties guaranteed by Theorems~\ref{thm:mc} and~\ref{thm:var}. As is well-known, $\mch$ decomposes into a single copy of the trivial representation $[n]$, and the standard representation $[n-1,1]$, which leads to a decomposition $\mcl^\DC \cong 2[n]\oplus [n-1,1]\oplus [n-2,2]\oplus[n-2,1,1].$
Direct calculations yield $a_{[n]}=1,\ a_{[n-1,1]}=0$ and $ a_{[n-2,2]},a_{[n-2,1,1]}\in\mco(1/n)$ (see  Appendix~\ref{sec:appdetails} for details, explicit constants and proofs). 

Although we have just seen that the permutation representation admits an NDCSE, this need not be the case for arbitrary $S_n$ irreps. Indeed, we can leverage well-known properties of the representation theory of the symmetric group to show that the following result holds:
\begin{restatable}{lemma}{snndcse}\label{lem:snndcse}
Not all $S_n$ irreps have NDCSEs.
\end{restatable}
The proof, given in Appendix~\ref{sec:appdetails}, is in fact a relatively simple counting argument: the number of inequivalent irreps of an abelian group is equal to its order, and there are (in the limit of large $n$, anyway) simply no abelian subgroups of $S_n$ big enough to cover the largest $S_n$ irreps in a multiplicity-free fashion.

\section{Discussion}
Although many distinct instantiations of the  CS theory that work well under various circumstances~\cite{huang2020predicting,aaronson2018shadow,cotler2020quantum,sugiyama2013precision,elben2020mixed,rath2021quantum,zhao2021fermionic,paini2021estimating,low2022classical,van2022hardware,wan2023matchgate,ippoliti2024classical,zhao2024group,king2024triply,jerbi2023shadows,koh2022classical,chen2021robust,hearth2024efficient,bertoni2024shallow,chan2022algorithmic,sauvage2024classical,sack2022avoiding,Helsen_2023,kunjummen2023shadow,denzler2023learningfermioniccorrelationsevolving,Wu_2024,west2024random,grier2024sample,brandao2020fast,vitale2024estimation,somma2024shadow,west2025real,bringewatt2025classical} have been developed, a more general theory that explains the performance of individual protocols in a simple and unified manner has remained lacking. Some hints to the essence of our results, however, have appeared in previous literature. For example, the connection between the sample efficiency of a protocol and the dimension of the subspace of operators to be learnt has previously been noticed in the CS  context~\cite{zhao2021fermionic, king2024triply, west2025real}. More specifically, the importance of the \textit{ratio} of that dimension to the dimension of a maximally commuting subspace of operators has been recognised in the case of free fermions~\cite{zhao2021fermionic}. Intuitively, measuring in the basis in which such a set of operators is diagonal allows one to unbiasedly estimate their expectation values; repeating the process a number of times comparable to the aforementioned ratio then accounts for the other operators of interest. A key contribution of our framework is to generalize this from the free-fermionic case to an arbitrary compact group $G$, where the notion of maximally commuting sets of operators is naturally encapsulated by weight-zero operator subspaces, and the aforementioned ratio by the $a_\lm^H$ of Theorems~\ref{thm:mc} and~\ref{thm:var}. Indeed, as we have found in our main theorems, the results of Ref.~\cite{zhao2021fermionic} generalize smoothly to this setting, modulo the complication that, at this level of generality, the  variance bounds may be weaker than those attainable when exploiting specific properties of a given $G$. 

Although the results presented here constitute a significant step towards a unified understanding of group-based CS, the story is not yet complete. Perhaps the main remaining open question is
\begin{center}
    \textit{Given a group acting in some representation, does a centralizing basis necessarily exist?}
\end{center}
We have shown that the existence of an NDCSE is a sufficient condition for the existence of a centralizing basis -- however, it is not necessary. Indeed, consider global Clifford shadows. On the one hand, as the Cliffords form a unitary 3-design, the channel will take its central form regardless of which basis we measure in; on the other, a generic basis is certainly not going to be a non-degenerate simultaneous eigenbasis for some abelian subgroup of the Clifford group. In particular, the nonexistence of NDCSEs for certain $S_n$ irreps guaranteed by Lemma~\ref{lem:snndcse} therefore does not preclude the existence of a centralizing basis for those irreps, and more generally the above question remains open.

On a different note, in Appendix~\ref{sec:tf}, we generalize {a technical result}  of Ref.~\cite{zhao2021fermionic} -- {obtained in their (free-fermionic) context via} the invocation of the machinery of 
\textit{tight frames}~\cite{vale2004tight,cotfas2010finite} -- to arbitrary compact groups. We will find that this method is highly compatible with our formalism, with the generalization very natural. Moreover, this result   allows us, in some instances, to improve upon the variance bounds generically guaranteed by Theorem~\ref{thm:var}. 

Beyond its theoretical appeal as a unifying framework, our work also offers a concrete suggestion for the development of new protocols: {pick an abelian subgroup $H$ of a group $G$ and measure in an NDCSE such that the $a_\lm^H$ of the observables of interest are as large as possible}. 
An example of the difficulties one can face {when failing to measure in a centralizing basis is} afforded by Ref.~\cite{sauvage2024classical}, wherein a sample efficient protocol based on sampling from $\SU(2)$ and acting with its $n$-qubit tensor representation is developed, but suffers from significant classical post-processing costs that can be thought of as arising due to measuring in a ``bad'' basis. As we have seen, a key property of measuring in a centralizing basis is the ease with which one can invert the measurement channel, and therefore the absence of such costs. In the case of the orthogonal group, for example, the suggestion to measure in a (non-degenerate) weight basis leads to the discovery of a shadow protocol distinct from those already considered~\cite{west2025real,liang2024real} (see Appendix~\ref{sec:appdetails}). In Appendix~\ref{sec:appdetails} we further see the benefits of the concreteness of our procedure on display in the somewhat exotic case of the exceptional Lie group $G_2$, where the recipe afforded by Theorem~\ref{thm:mc} easily and mechanically characterizes a novel shadow protocol. More generally, this framework adds to an ever-increasing list~\cite{heinrich2022randomized,wilkens2024benchmarking,arienzo2024bosonic,helsen2019new} of  examples of the ability of representation theory to facilitate and clarify quantum information-theoretic protocols.

Next, we remark that we have considered here only the setting of independent, non-adaptive measurements. In  more general settings one can imagine either making entangling measurements in multiple copies of the input state~\cite{huang2021information,grier2024sample,chen2021exponential,king2024exponential}, or measuring only one copy of the state at a time, but in a fashion where the choice of latter measurements is allowed to depend on the results of previous measurements~\cite{chen2023does}. There are scenarios in which single-copy measurements provably fail to be sample-efficient, but for which many-copy measurements can provide sample efficient protocols~\cite{king2024triply,chen2021exponential}; likewise there are scenarios in which adaptive measurements are provably sample-efficient, but non-adaptive measurements are not~\cite{chen2024adaptivity}. The extension of the representation-theoretic ideas of this work to these regimes constitutes an exciting avenue for future investigation.

{A further interesting avenue for generalization is to consider \textit{non-abelian} analogues of our NDCSE construction. From a representation-theoretic point of view, the simplicity of Theorem~\ref{thm:mc} can in some sense be traced to the fact that when $\mcw$ is an NDCSE the dephasing channel $\mca_\mcw$ of Eq.~\eqref{eq:mcinit} can be shown to be equivalent to the twirl of the group action over the corresponding $H$ (see Appendix~\ref{sec:gshaddetails}). In the case of a non-abelian subgroup $K$, then, one could naturally consider replacing the dephasing channel with the twirl over $K$, which will \textit{not} be a rank-one dephasing channel, as any non-abelian group possesses at least one irrep with dimension greater than one~\cite{fulton1991representation}. Nonetheless, the proof of Theorem~\ref{thm:mc} in Appendix~\ref{sec:gshaddetails} (relying as it does largely on Schur's lemma) reveals that an analogous result will hold in this more general case.  }

Finally, we note that while the main focus of our work is a mathematical characterization of CS protocols, we leave a thorough exploration of the implications of our main theorems for future work. For instance, recent works have uncovered a universal behavior across multiple quantum resource theories where Lie groups  constitute free operations~\cite{chitambar2019quantum,diaz2025unified}. Namely,  low resource states tend to live in subspaces with $1/a_\lambda^H\in\OC(\poly(n))$, whereas highly resourceful ones have support in subspaces with $1/a_\lambda^H\in\Om(\exp(n))$~\cite{bermejo2025characterizing,coffman2026group}. A direct application of Theorem~\ref{thm:var} implies that extracting information via CS from highly resourceful states could be significantly more expensive than from low resource states, thus making an explicit connection between worst-case shot-cost and the hardness of learning resourceful states. We expect that additional deep connections will be unearthed as our CS results are applied to other areas of quantum information.

\section{Acknowledgments}

ML, MW and MC acknowledge support by the Laboratory Directed Research and Development (LDRD) program of LANL under project number 20230049DR and 20260043DR, and by the LANL's ASC Beyond Moore’s Law project. MW, AS and RTF were supported by the U.S. Department of Energy through a quantum computing program sponsored by the Los Alamos National Laboratory Information Science \& Technology Institute. MW acknowledges the support of the Australian Government Research Training Program Scholarship, the Dr. Albert Shimmins Fund,  and the IBM Quantum Hub  at the University of Melbourne. This work was also supported by the Quantum Science Center (QSC), a National Quantum Information Science Research Center of the U.S. Department of Energy (DOE). DG acknowledges support by NWO grant NGF.1623.23.025 (“Qudits in theory and experiment”). FS thanks Colin Krawchuk and Jed Burkat for comments on the manuscript.

\bibliography{short_main/refs,short_main/quantum}

\newpage

\onecolumngrid
\appendix

\addtocontents{toc}{\protect\AppendixTOCMark}

\tableofcontentsappendixonly

\section{Classical shadows}\label{sec:cs}
We begin by briefly reviewing the formalism of CS, as introduced in~\cite{huang2020predicting}. CS have become a major area of research within quantum computation and information, and more details may be found in a rapidly growing list of sources~\cite{huang2020predicting,aaronson2018shadow,cotler2020quantum,sugiyama2013precision,elben2020mixed,rath2021quantum,zhao2021fermionic,paini2021estimating,low2022classical,van2022hardware,wan2023matchgate,ippoliti2024classical,zhao2024group,king2024triply,jerbi2023shadows,koh2022classical,chen2021robust,rath2021importance,hearth2024efficient,bertoni2024shallow,chan2022algorithmic,sauvage2024classical,sack2022avoiding,Helsen_2023,kunjummen2023shadow,denzler2023learningfermioniccorrelationsevolving,Wu_2024,west2024random,grier2024sample,brandao2020fast,vitale2024estimation,barthe2024gate,somma2024shadow,diaz2023showcasing,west2025real}. \\

Given access to copies of an (unknown) $d$-dimensional quantum state $\rho$ and a set of target observables $\{O_i\}_{i=1}^M$, a CS protocol produces estimates of the corresponding expectation values, $ {o}_i=\Tr [\rho O_i]$. The procedure begins by evolving $\rho$ via a randomly sampled unitary $U\sim \mcu$, where $\mcu$ is a specified ensemble of unitaries determined by the protocol, and then measuring the resulting state in a chosen basis $\mc{W}=\{\ket{w_i}\}_{i=1}^d$, obtaining some result $\ket{w}$. The combined choice of unitary distribution and measurement basis characterizes the shadow protocol. One considers the output of the process to be the state $U^\dagger\ketbra{w}U\equiv U^\dagger\Pi_{w}U$, and classically stores the pair $(U,w)$, also known as a ``snapshot''. The result of this is to effect the  channel
\begin{equation}
\mcm(\rho)= \expect_{U\sim \mcu}  \sum_w  \Tr\left[\rho U^\dagger\Pi_wU\right] U^\dagger\Pi_w U =: \mathbb{E}_{U,w|\rho}\left[ U^\dagger\Pi_w U\right].
\end{equation}
The next step is to (pseudo-)invert $\mcm$; the result of the channel inversion on a given instance,
\begin{equation}\label{eq:shad}
\hat{\rho} = \mcm^{-1}\left(U^\dagger\Pi_wU\right)
\end{equation}
is termed a \textit{classical shadow} of $\rho$. When the protocol is  tomographically complete,  the shadows constitute an unbiased estimator of  $\rho$, as one has in expectation
\begin{align*}
\mathbb{E}_{U,w|\rho}\left[\hat{\rho}\right]&= \expect_{U\sim \mcu}  \sum_w p(U,w|\rho) \hat{\rho} \\
&=\expect_{U\sim \mcu}  \sum_w  \Tr\left[\rho U^\dagger\Pi_wU\right] \mcm^{-1}\left(U^\dagger\Pi_wU\right)\\
&=\mcm^{-1}\left(\expect_{U\sim \mcu}  \sum_w  \Tr\left[\rho U^\dagger\Pi_wU\right] U^\dagger\Pi_wU\right)\\
&=\mcm^{-1}\left( \mcm\left(\rho \right) \right)\\
&=\rho\,.
\end{align*}
More generally, $\mcm$ is invertible only on its image $\Im(\mcm)$, whence (with the inverse understood to denote the Moore-Penrose pseudoinverse)
\begin{equation}
\mcm^{-1}\circ \mcm = \Pi_{\Im(\mcm)} \neq {\rm id}_{\mcl}\,.
\end{equation}
One may of course still estimate expectation values with respect to $\rho$ via $\hat{o}_i=\Tr [\hat{\rho} O_i]$, which
will be unbiased when the measured operator $O$ lives in the \textit{visible space} of the protocol~\cite{van2022hardware}, 
\begin{align}\label{eq:vis}
\mathsf{VisibleSpace}\left(\mcu, \mc{W}\right) &= {\rm span }\left\{ U^\dagger \Pi_{w} U \right\}_{U\sim\hspace{0.3mm}\mcu, \ket{w} \in\mc{W}}\\
&=\Im(\mcm)\,.
\end{align}
In the non-tomographically complete case the visible space will be a strict subspace of $\eh$, necessitating assessing whether the target observables belong to the visible space. When this is indeed the case, the selection of a tailored ensemble can lead to drastic improvements in sample-complexity over a generic  tomographically complete choice~\cite{sauvage2024classical}. \\

Having verified that one's estimates are unbiased (or having accepted that they are not~\cite{van2022hardware}), a key remaining question is the number of samples of $\rho$ needed to guarantee that the empirical average of estimators  obtained from  multiple independent CS is close to the true expectation value of the underlying state.
This number depends critically on the variance of the single-shot estimator,
\begin{align}
\mathrm{Var }\left[ \hat{o} \right] &= \mbe_{U,w|\rho}[\hat{o}^2]-\left( \mbe_{U,w|\rho}[\hat{o}]\right)^2\nonumber\\
&= \mbe_{U,w|\rho}\left[\Tr \left[  \mcm^{-1}\left(U^\dagger\Pi_wU\right) O \right]^2 \right]-\left( \mbe_{U,w|\rho}[\hat{o}]\right)^2\nonumber\\
&\leq \max_\rho \mbe_{U,w|\rho}\left[\Tr \left[  \mcm^{-1}\left(U^\dagger\Pi_wU\right) O \right]^2 \right]\label{eq:shadvar}\nonumber\\
&=:\|O\|_{\rm shadow}^2 \, .
\end{align}
Specifically, one is guaranteed  (with probability at least $1-\delta$) to be able to estimate each of $\{\Tr [\rho O_i]\}_{i=1}^M$ to within a precision $\epsilon$ using~\cite{huang2020predicting} 
\begin{equation}\label{eq:shadsamp}
N_{\rm shad} = \mc{O} \left(\frac{\log(M/\delta)}{\epsilon^2} \max_{1\leq i\leq M} \|O_i\|_{\rm shadow}^2\right)
\end{equation}
copies of $\rho$. For a successful CS protocol, bounding the shadow norms for an interesting class of observables is vital.  Indeed, doing so for both general and specific choices of $G$ and $O$ will form an important part of the forthcoming analysis. 

\section{Technical preliminaries}\label{sec:wc}
Let us next briefly review some further relevant background details, starting with some elementary elements of twirling over groups~\cite{weingarten1978asymptotic,mele2024introduction,fulton1991representation}.  

\subsection{Group twirling}
\noindent
An important component of analyzing a CS protocol is the evaluation  of the \textit{$k$\textsuperscript{th}-moment operator}
\begin{equation}\label{eq:momop}
\mce^{(k)}_G (A) =\expect_{U\sim G} \ U^{\dagger \otimes k} AU^{ \otimes k} ,
\end{equation}
where $A\in{\rm End}\,(\mathcal{H}^{\otimes k} )$ is a linear operator on $\mathcal{H}^{\otimes k} $, and $k=2,3$. For example, the $k=2$ case appears in the definition of the measurement channel:
\begin{align}
\mc{M}(\rho) &=\sum_w \expect_{U\sim G} \ \Tr[ \rho U\ad \Pi_w U] U\ad \Pi_w U= \sum_w\Tr_1 \left[ (\rho\otimes\id) \expect_{U\sim G} \   U^{\dagger\otimes 2} \Pi_w^{\otimes 2} U^{\otimes 2}\right], \label{eq:k2}
\end{align}
and the $k=3$ case when calculating the estimator variance:
\begin{align}
\mathrm{Var}[\hat{o}]&=\mbe[\hat{o}^2]-\mbe[\hat{o}]^2\\
&\leq \mbe[\hat{o}^2]\\
&=\sum_w\expect_{U\sim G}\  \Tr[\rho U^\dagger \Pi_{w}  U]  \Tr[  O\mcm^{-1}\left(  U^\dagger \Pi_{w} U\right)] ^{ 2}\\
&=\sum_w\expect_{U\sim G}\  \Tr[\rho U^\dagger \Pi_{w}  U]  \Tr[  \mcm^{-1}\left(O\right)  U^\dagger \Pi_{w} U] ^{ 2}\\
&=\sum_w\Tr[(\rho\otimes\mcm^{-1}\left(O\right)\otimes\mcm^{-1}\left(O\right)) \expect_{U\sim G}\   U^{\dagger \otimes 3} \Pi_{w}^{\otimes 3}  U^{\otimes 3}]  \label{eq:k3}\,.
\end{align}
Above, we have used the self-adjointness of $\mcm^{-1}$ with respect to the Hilbert-Schmidt inner product. 
Now, as is well-known~\cite{mele2024introduction}, the $k$\textsuperscript{th}-moment operator $\mce^{(k)}_G$ orthogonally (with respect to the Hilbert–Schmidt) projects onto the $G$-invariant subspace $({\rm End}\,(\mathcal{H}^{\otimes k} ))^G$ of ${\rm End}\,(\HC^{\otimes k})\cong (\eh)\tk=\mcl^{\otimes k}$, namely
\begin{equation}
({\rm End}\,(\mathcal{H}^{\otimes k} ))^G= \big\{ B\in \mcl\left(\mathcal{H}^{\otimes k} \right) \ \vert \ [B,U^{\otimes k}]=0\ \forall U\in G \big\} \subseteq {\rm End}\,(\mathcal{H}^{\otimes k} ).
\end{equation}
So, for an operator $B \in {\rm End}\,(\mathcal{H}^{\otimes k} )$ to belong to $({\rm End}\,(\mathcal{H}^{\otimes k} ))^G$ one must have that $\forall U\in G,\ BU^{\otimes k}=U^{\otimes k}B$, and therefore $U^{\otimes k}BU^{\dagger\otimes k}=B$, i.e., the (adjoint) action of $G$ on $B$ is trivial. From this we see that $\mce^{(k)}_G $ projects onto the trivial irreps of $\mcl^{\otimes k}$ (with respect to the adjoint action of the $k$-fold tensor product of $G$), which we are therefore interested in characterizing. For large $k$ this is in general an intricate problem, and one of the main advantages of our approach over the previously existing case-by-case analysis in the literature is that we can derive the exact shadow channel and (bounds on the) variance using knowledge only of the trivial irreps at $k=1$.

\subsection{Weight-modules}\label{sec:weights}
In this section we take a brief break from CS to define the minimal properties of \textit{weight-modules} needed to understand the constructions in the main text. All of the material of this appendix is well-known and stated without justification; a general reference in which much more detail may be found is given by, for example, Ref.~\cite{fulton1991representation}. We begin by recalling the fact that any finite-dimensional complex semisimple Lie algebra $\mfg$ admits a \textit{root space decomposition}
\begin{equation}\label{eq:rsd}
\mfg=\mfh\oplus\bigoplus_{\alpha\in\Phi}\mfg_\alpha.
\end{equation}
Here $\mfh$ is a so-called Cartan subalgebra, i.e., a maximal abelian subalgebra for which the adjoint actions ${\rm ad}_h$ are simultaneously diagonalisable as $h$ ranges over $\mfh$. In the case of interest to us (finite-dimensional complex semisimple Lie algebras) the Cartan subalgebra is unique up to an automorphism of $\mfg$. For each linear functional $\alpha\in\mfh^*$ one defines the set $\mfg_\alpha=\{x\in\mfg : {\rm ad}_h(x)=\alpha(h)x\ \forall h\in\mfh \}$; the set $\Phi$ of Eq.~\eqref{eq:rsd} is the set of functionals $\alpha\in\mfh^*$ for which $\mfg_\alpha$ is non-zero; such functionals are called \textit{roots}. One can show  that if $\alpha\in\Phi$ then $\mfg_\alpha$ is one-dimensional, and moreover that $-\alpha\in\Phi$. We also introduce a partition of the roots into ``positive'' and ``negative'' roots, $\Phi=\Phi^+\cup\Phi^-$, such that if $\alpha\in\Phi^+$ then $-\alpha\in\Phi^-$. The exact choice of partition is somewhat arbitrary, but does not turn out to affect things materially. A positive (negative) root that cannot be written as the sum of two positive (negative) roots is said to be \textit{simple}. The positive simple roots $\{\alpha_i\}_{1\leq i\leq r}$ can be shown to span $\mfh^*$, and by the duality induced by the so-called \textit{Killing form} yield a basis $\{H_{\alpha_i}\}_{1\leq i\leq r}$ of $\mfh$. 
\\

Now, suppose we have a representation $\pi$ on a vector space (indeed $\mfg$-module) $V$ of a (finite-dimensional complex semisimple) Lie algebra $\mfg$, with some choice of Cartan subalgebra $\mfh$. 
One defines a weight vector $v\in V$ of weight $\lm\in\mfh^*$ to be any simultaneous eigenvector of $\pi(\mfh)$ with eigenvalues given by
\begin{equation}\label{eq:wv}
    \pi(h)v=\lm(h)v\ \ \forall h\in\mfh;
\end{equation}
the span of all such $v$ is said to be the weight-space of weight $\lm$ of the $\mfg$-module $V$.
If $V$ has a basis consisting of weight vectors, it is said to be a weight-module. The representation theory of complex semisimple Lie algebras is greatly facilitated by the following two results:

\begin{fact}
If $\mfg$ is a  finite-dimensional complex semisimple Lie algebra, then  every finite-dimensional $\mfg$-module $V$ is completely reducible.
\end{fact}
\begin{fact}\label{fact:f2}
If $\mfg$ is a  finite-dimensional complex semisimple Lie algebra, then  every finite-dimensional irreducible $\mfg$-module $V$ is a weight-module.
\end{fact}
In light of these facts, given some irrep $\mcl^\DC_\lm$ (as encountered in the main text) upon which the Lie algebra $\mfg$ of the Lie group $G$ that constitutes the CS ensemble acts in the adjoint representation (i.e., via the correspondence $U\in G : O \mapsto UOU^\dagger = {\rm Ad}_U (O) = \exp ({\rm ad}_X)(O)$ where $X\in\mfg$ satisfies $\exp X = U$) one can define the weight-zero subspace with respect to some chosen Cartan subalgebra $\mfh$. In the notation of Eq.~\eqref{eq:wv} we have $\lm=0$ and $\pi={\rm ad}$, i.e., the weight-zero subspace is the set of operators $O$ satisfying $[X,O]=0\ \forall X\in\mfh$. The uniqueness (up to automorphism) of Cartan subalgebras guarantees that the dimension $d^H_\lm$ of the weight-zero subspace is independent of the choice of $\mfh$.\\ 

Indeed, we note that Fact~\ref{fact:f2} can be strengthened further. We need two more definitions. First, define a \textit{highest weight vector} of an irreducible $\mfg$-module to be a weight vector that is annihilated by $\mfg_\alpha$ for all $\alpha\in\Phi^+$. Call the weight of the highest weight vector  the \textit{highest weight}. Second, introduce the dual basis $\{\omega_i\}_{1\leq i\leq r}$ to the basis $\{H_{\alpha_i}\}_{1\leq i\leq r}$ of $\mfh$ induced by the simple roots, and call $\{\omega_i\}_{1\leq i\leq r}$ the \textit{fundamental weights} of the representation. We then have
\begin{fact}\label{fact:f3}
If $\mfg$ is a  finite-dimensional complex semisimple Lie algebra, then every finite-dimensional irreducible $\mfg$-module $V$ is a weight-module which possesses a highest weight $\w=\sum_{i=1}^r n_i\w_i$ where $n_i\in\mbz_{\geq 0}\ \forall\ i$. The coefficients $n_i$ are called the Dynkin labels of the weight. The highest weight uniquely characterizes the representation.
\end{fact}
Fact~\ref{fact:f3} will play an important role in the forthcoming appendices, as we try to understand the irreps appearing in the decomposition of $\mcl^\DC$ for various examples. It follows from what we have seen that an irreducible representation is  exactly specified by the vector $(n_1, n_2,\ldots,n_r)$ of the Dynkin labels of its highest weight; we will occasionally use this notation.\\

For Lie groups, our treatment of CS will be systematic: we determine the structure of $\mch$ as a $G$-module, and find a weight basis. The weights of $\mch$ as a $G$-representation will in all our examples be degeneracy-free, Theorem~\ref{thm:mc} will imply that that weight basis is centralizing. We then identify the  space of visible operators $\mcl^\DC$, and determine its structure  as a $G$-module (including finding the dimension of the weight-zero subspaces of the irreps, which corresponds to the $d_\lm^H$ of Theorem~\ref{thm:mc}), thereby obtaining the measurement channel.

\section{Proofs of main results}\label{sec:gshaddetails}

In this appendix we give the proofs of the main theorems from the text.  Let us begin by recalling some definitions:
\deffb*
\defnlmgood*
\defngood*
\defcse*
\noindent
Let us begin by proving:
\lemfgood*
\begin{proof}
It is sufficient to show that the corresponding channel is the identity on the trivial irreps in $\mcl^\DC$ (we will moreover show that it annihilates everything not in $\mcl^\DC$). Indeed, as by Schur's lemma an arbitrary trivial irrep in $\mcl$ is spanned by $O_{i,j}=\sum_\a\ketbra{\eta,i,\a}{\eta,j,\a}$ for some pair $1\le i,j \le m_\eta^\mch$ we can simply directly calculate:
\begin{align}
    \mcm(O_{i,j})&=\expect_{g\sim G} g\ad\cdot\mca_\mcw\cdot g \cdot \left(O_{i,j}\right)\\
    &=\expect_{g\sim G} g\ad\cdot\mca_\mcw\cdot  \left(O_{i,j}\right)\\
    &=\expect_{g\sim G} g\ad\cdot\sum_{\zeta,k,\b}\bra{\zeta,k,\b}\left(O_{i,j}\right)\ket{\zeta,k,\b}\ketbra{\zeta,k,\b}\\
    &=\expect_{g\sim G} g\ad\cdot\sum_{\zeta,k,\b}\bra{\zeta,k,\b}\left(\sum_\a\ketbra{\eta,i,\a}{\eta,j,\a}\right)\ket{\zeta,k,\b}\ketbra{\zeta,k,\b}\\
    &=\expect_{g\sim G} g\ad\cdot  \sum_\a \delta_{i,j}\ketbra{\eta,i,\a}{\eta,j,\a}\\
    &=\delta_{i,j}\expect_{g\sim G} g\ad\cdot O_{i,j}\\
    &=\delta_{i,j}  O_{i,j}\,,
\end{align}
where we have used that by definition the action $g\cdot O_{i,j}=O_{i,j}$ is trivial. As $O_{i,j}\in\mcl^\DC$ iff $i=j$ we see that the measurement channel indeed acts as the identity on the visible portion of the trivial isotypic, and is therefore trivial centralizing. 

\end{proof}

\noindent
Before coming to the proof of Theorem~\ref{thm:mc}, we have two preliminary lemmas:

\begin{lemma}\label{lem:2}
If $\mcw=\{\ket{\eta,i,\a}\}_{\eta,i,\a}$ is an NDCSE with respect to an abelian subgroup $H$, then the dephasing map $\mca_\mcw = \sum_{w\in W} \braket{w|(-)}{w}\ketbra{w}$ satisfies 
\begin{equation}
    \mca_\mcw = \mct_H\circ \mcp_D,\label{eq:hdeph}
\end{equation}
where $\mct_H = \expect_{h\sim H} h$ and $\mcp_D$ is the projector onto the space $\mcl^\DC$ of $G$-diagonal operators. 
\end{lemma}

\begin{proof}
First, note that because $\mcw$ is a CSE, every basis vector $\ket{\eta,i,\alpha}\in \mcw_{\eta,i}$ is an $H$-eigenvector, i.e. $h\,|{\eta,i,\alpha}\rangle = \chi_{\eta,i,\alpha}(h)\ket{{\eta,i,\alpha}}$
for some character $\chi_{\eta,i,\alpha}\in \widehat{H}$. Consider a matrix unit $E_{\eta,i;\alpha\beta}:=\ket{{\eta,i,\alpha}}\bra{ \eta,i,\beta}$ inside a single block
${\rm End}\,(\mathcal H_{\eta,i})$; we then have
\begin{equation}
  \mathcal{T}_H(E_{\eta,i;\alpha\beta})
  = \mathbb{E}_{h\sim H}\, \chi_{\eta,i,\alpha}(h)\,\overline{\chi_{\eta,i,\beta}(h)}\, E_{\eta,i;\alpha\beta}.
\end{equation}
If $\alpha=\beta$, the coefficient is $1$, and if
 $\alpha\neq\beta$, then by the multiplicity-freeness of   the characters guaranteed by the NDCSE condition, $\chi_{\eta,i,\alpha}$ and $\chi_{\eta,i,\beta}$ are distinct, whence
$\mathbb{E}_{h\sim H}\chi_{\eta,i,\alpha}(h)\overline{\chi_{\eta,i,\beta}(h)}=0$ by character orthogonality.
Therefore, within each block ${\rm End}\,(\mathcal H_{\eta,i})$ the $H$-twirl acts exactly as dephasing in the basis
$W_{\eta,i}$. Combining this with the immediate observation that both the left and right hand sides of Eq.~\eqref{eq:hdeph} annihilate the orthogonal complement of $\mcl^\DC$ concludes the proof.
\end{proof}

\noindent
Next we have another lemma:

\begin{lemma}\label{lem:conjtwirl}
For $h\in H\subseteq G$, the ``conjugacy-class twirl''
\begin{equation}\label{eq:cctwirl}
  \mathcal{C}_h(X) := \mathbb{E}_{g\sim G}\, (g^\dagger h g)\cdot X.
\end{equation}
satisfies
\begin{equation}\label{eq:ch}
\mathcal{C}_h  =\sum_{\lm} \frac{\chi_\lambda(h)}{d_\lambda} \id_{V^\lambda},
\end{equation}
where we sum over all irreps which appear in $\mcl$. 
\end{lemma}

\begin{proof}
To begin, note that the manifest $G$-equivariance of  $C_h$ implies that it acts non-trivially at most on the multiplicity spaces of the decomposition of $\mcl$ into $G$-irreps. Importantly, more is true: It is in fact \textit{central} in the commutant of $G$; indeed, $C_h$ evidently commutes with any map that commutes with the action of  $G$. It then follows from Schur's lemma that $C_h$ restricts to a scalar multiple $c_h(\lm)$ of the identity on each isotypic. Taking the trace of the equation $\mathbb{E}_{g\sim G}\, (g^\dagger h g)\rvert_{\lm} = c_h(\lm)\id_\lm$ then fixes the scalar coefficient and immediately yields the result.
\end{proof}

\noindent
We are now ready to prove:

\thmmc*

\begin{proof}
The combination of Lemmas~\ref{lem:2} and~\ref{lem:conjtwirl} allow us to proceed by a direct calculation. First, recalling the definition of the measurement channel and using the result of Lemma~\ref{lem:2} gives
\begin{equation}
    \mcm =\expect_{g\sim G} g\ad\cdot\mca_\mcw\cdot g = \expect_{g\sim G} g\ad\cdot(\mct_H\circ \mcp_D)\cdot g=\left(\expect_{g\sim G} g\ad\cdot \mct_H \cdot g \right) \circ \mcp_D=\left(\expect_{g\sim G} \expect_{h\sim H}g\ad\cdot h \cdot g \right) \circ \mcp_D
\end{equation}
where we have used that the projector $\mcp_D$ onto the space of $G$-diagonal operators manifestly commutes with the action of any $g\in G$. Now, by Lemma~\ref{lem:conjtwirl}, we have
\begin{equation}
    \expect_{g\sim G} \expect_{h\sim H}g\ad\cdot h \cdot g =  \expect_{h\sim H}C_  h   = \sum_\lm \frac{\expect_{h\sim H} \chi_\lm (h)}{d_\lm}\id_{V^\lm}  = \sum_\lm \frac {d_\lm^ H}{d_\lm}\id_{V^\lm}  = \sum_\lm  a_\lm^ H \id_{V^\lm}  
\end{equation}
where we have used that $\expect_{h\sim H} \chi_\lm (h) $ is exactly the dimension of the $H$-invariant subspace of the $G$-irrep $\lm$~\cite{fulton1991representation}. 
Recalling that we write $\PC^\VC_\lm$ for the projector onto the copies of the irrep $\lm$ in the image of $\mcm$ (which are exactly the irreps which occur in $\mcl^\DC$ with a non-zero value of $a_\lm^H$), we are in a position to conclude that $\mcm=\sum_\lm a_\lm^H \PC^\VC_\lm$, as desired. 
\end{proof}

Now, as small values of $a_\lm^H$ result in the corresponding module being almost entirely annihilated by $\mcm$, one would expect that observables living there are sample-wise expensive to estimate.
Recalling the connection between the variance of the estimator, the shadow norm and the sample-complexity of CS (see Eqs.~\eqref{eq:shadvar} and~\eqref{eq:shadsamp}), the previous intuition is
made rigorous by the following theorem:

\thmvar*

\begin{proof}
We begin with some general observations. First, as
the measurement channel $\mathcal{M}$ is self-adjoint it admits a set of eigenvectors $\set{E_k}_{k}$ that form an {orthonormal} basis of $\mcl$. From Theorem~\ref{thm:mc} we know that the elements of any basis $\{E^{\lm,i}\}_i$ of operators for the isotypic $\mcl^\VC_\lm$ each have $\mcm$ eigenvalue $a_\lm^H$. Taking the union over $\lm$ yields a basis for the subspace $\mcl^\VC$, which we extend to a basis for $\mcl$ (with all the new eigenvectors being annihilated by $\mcm$). By the orthonormality of the (eigen-)basis elements we have:
\begin{align}
\begin{split}
\label{eq:app_lambda_i_alt}
    a_\lambda \delta_{\lm,\nu} &= \Tr [\mathcal{M}(E_\lm) E_\nu] = \Tr \left[\left( \sum_w \expect_{U\sim G} \  \Tr[U^{\dagger} \Pi_w U E_\lm] U^{\dagger} \Pi_w U \right) E_\nu \right] \\ 
    &=\sum_w \expect_{U\sim G} \  \Tr [U^{\dagger} \Pi_w U E_\lm] \Tr[U^{\dagger} \Pi_w U E_\nu]\,.
\end{split}
\end{align}
We also note that, by H\"older's inequality and the unitary invariance of the 1-norm, 
\begin{equation}~\label{eq:holder}
\lvert\Tr[U^{\dagger} \Pi_w U A]\rvert\leq \| U^{\dagger} \Pi_w U \|_1 \| A\|_{\infty}=\|  \Pi_w  \|_1 \| A\|_{\infty}=\| A\|_{\infty},
\end{equation}
and further recall that the pseudo-inverse $\mathcal{M}^{-1}$ is only defined on the image $\mcl^\mcv$ of $\mathcal{M}$, $\mathcal{M}^{-1}(E^{\lm,i}) = (1/a_\lm^H) E^{\lm,i}  $ and annihilates the operators outside of $\mcl^\mcv$. 
Writing $O=\left( \sum_\lm O^\lm \right) + O'=\left( \sum_{\lm,i} o^{\lm,i} E^{\lm,i} \right) + O'$,
with $O'$ the component of $O$ in the orthogonal complement of $\mcl^\mcv$, we therefore have $\mcm^{-1}(O)=\sum_{\lm\in\widehat{\mcv}} (a_\lm^H)^{-1}O^\lm .$
We now turn to the proofs of the quoted bounds.
\\

\textit{First bound.}
From the definition of the variance we have
\begin{align}\label{eq:app_var_ei}
\Var[\hat{o}] &=\mbe[\hat{o}^2]-\mbe[\hat{o}]^2\\
&\leq \mbe[\hat{o}^2]\\
&=\sum_w \expect_{U\sim G} \  \Tr[U^{\dagger} \Pi_w U\rho]\Tr[U^{\dagger}\Pi_w U \mathcal{M}^{-1}(O)] ^2 \\
&\leq  \sum_w \expect_{U\sim G} \  \Tr[U^\dagger\Pi_w U \mathcal{M}^{-1}(O)] ^2 \label{eq:rhoprob} \\
&= \sum_{\lm,\lm',i,i'}\frac{o^{\lm,i} o^{\lm',i'}}{a_{\lm}a_{\lm'}} \sum_w \expect_{U\sim G} \  \Tr[U^{\dagger}\Pi_w U E^{\lm,i}]\Tr[U^{\dagger}\Pi_w U E^{\lm',i'}]\\
&= \sum_{\lm,\lm',i,i'}\frac{o^{\lm,i} o^{\lm',i'}}{a_{\lm}a_{\lm'}} a_\lm^H\delta_{\lm,\lm'}\delta_{i,i'}\label{eq:ortho}\\
&= \sum_{\lm}\frac{\|O^\lm\|_2^2}{a_{\lm}} \,.
\end{align}
Eq.~\eqref{eq:rhoprob} is obtained by using   $0\leq \Tr[U^{\dagger} \Pi_w U\rho] \leq 1$, as it corresponds to a probability. In Eq.~\eqref{eq:ortho}, we have recalled  Eq.~\eqref{eq:app_lambda_i_alt}.\\

\textit{Second bound.} We have:
\begin{align}
\mathrm{Var}[\hat{o}]&=\mbe[\hat{o}^2]-\mbe[\hat{o}]^2\\
&\leq \mbe[\hat{o}^2]\\
&=\sum_w\expect_{U\sim G} \  \Tr[\rho U^\dagger \ketbra{w}  U]  \Tr[  \mcm^{-1}\left(O\right)  U^\dagger \ketbra{w} U] ^{ 2}\\
&\leq\sum_w\expect_{U\sim G} \  \Tr[\rho U^\dagger \ketbra{w}  U]\label{eq:hi}   \| \mcm^{-1}\left(O\right)\|_\infty^2  \\
&=  \| \mcm^{-1}\left(O\right)\|_\infty^2\expect_{U\sim G} \  \Tr[\rho U^\dagger \sum_w\ketbra{w}  U]   \\
&=  \| \mcm^{-1}\left(O\right)\|_\infty^2 \label{eq:idres}\\
&=  \left\|   \sum_{\lm} \frac{O^{\lm} }{a_\lm^H}  \right\|_\infty^2\,,
\end{align}
where in Eq.~\eqref{eq:hi} we have used Eq.~\eqref{eq:holder} with $A=\mcm^{-1}\left(O\right)$, and in Eq.~\eqref{eq:idres} that $\sum_w\ketbra{w} =\id$. 

\end{proof}

In the examples we study below we find that typically $d^H_\lm\ll d_\lm$, and so 
observables living in large modules are  generally difficult to estimate, one of the main lessons of this work. Conversely, observables living in a module polynomial in the system size can, if their spectral norm is likewise polynomial in the system size, always be efficiently estimated with such shadows. Which of the bounds Eqs.~\eqref{eq:2nb} and~\eqref{eq:inb} are tighter will depend on the specific observable; we always have $\|O^\lm\|_\infty\leq\|O^\lm\|_2$ and $a_\lm^H\leq 1$, with the better bound therefore being determined on a case by case basis. We note that the existence of an analogous result to the final bound of Eq.~\eqref{eq:2nb} in the infinity norm case (i.e., $\|O\|_\infty^2 \max_{\lm}{(a_\lm^H)^{-2}}$) is less clear, as we have $\sum_\lm \|O^\lm\|_2^2=\|O\|_2^2$, but in general   $\sum_\lm \|O^\lm\|_\infty^2\neq \|O\|_\infty^2$. \\

When certain conditions are met we can improve the bounds of Theorem~\ref{thm:var}.
For example, when restricting to a specific group of interest one can potentially use the additional structure present to obtain tighter bounds than the general case (we will carry out such analysis for two of the examples below, when $G=\SU(2),\ \O(2^n) $). More generally, we additionally have:
\begin{restatable}{thm}{thmvarf}\label{thm:varf}
Let $O\in(\mcl^\DC)^\lm$. If one of the following conditions 
\begin{enumerate}[(i)]
  \item $O$ is positive or negative semi-definite
  \item $O$ is, for some $\eta,\lm,i,j$, an element of a uniform spectral norm self-adjoint basis of $({\rm End\,}\mch^{\eta,i})^{\lm,j}$ whose normaliser contains $G$, and whose members are either diagonal or entirely off diagonal with respect to the measurement basis
\end{enumerate}
is met, then we have 
\begin{align}
\Var[\hat{o}] &\leq  \frac{\|O\|^2_\infty}{a_\lm^H}\label{eq:inbf}.
\end{align}
\end{restatable}
The first condition occurs, for example, when the observable one is measuring is a projector with support on a single isotypic, or the fidelity with respect to another state.
An interesting example of the second condition being met is when $G$ is the set of matchgate Cliffords, and the observables Majorana fermions of fixed even degree (upon which $G$ simply acts as a permutation), the setting considered in Ref.~\cite{zhao2021fermionic}. The proof of this theorem may be found in Appendix~\ref{sec:tf}, where we discuss and generalize the formalism (based on \textit{tight frames}) of Ref.~\cite{zhao2021fermionic}.

\section{Details of applications}\label{sec:appdetails}
We now give a series of examples of shadows in an NDCSE basis (some of which briefly appeared in the main text). We begin by comparing the protocols obtained by taking various representations of the unitary, Clifford,  orthogonal and symplectic groups to the existing classical shadow protocols in the literature that employ those representations~\cite{huang2020predicting,zhao2021fermionic,low2022classical,wan2023matchgate,west2025real,liang2024real,west2024random}. In all cases   we find that the derived shadows protocol is in agreement with the previous protocols -- i.e., that they involve (if implicitly) measuring in a centralizing basis. Of particular interest is the case of global orthogonal shadows, where we are naturally led to discover a genuinely new protocol. Following this, we move on to the novel settings of the spin-$S$ and tensor representations of $\SU(2)$, as well as, more exotically, a CS protocol in which one samples unitaries from the exceptional Lie group $G_2$.

\subsection{Global unitaries and global Cliffords}\label{sec:gu}
We begin with the important case of $G$ being the full unitary group $\U(d)$ on $n$ qubits,
where $d=2^n$.
Our first step will always be to ask how $G$ acts on $\mch$; the answer in this instance is \textit{irreducibly},  from which we conclude that \textit{all} operators $O\in\eh$ are $\U(d)$-diagonal, in agreement with the known fact that global unitaries yield a tomographically complete protocol~\cite{huang2020predicting}. Next, we need to know how $\mcl^\DC=\mcl$ decomposes into $\U(d)$-irreps. Here  we have a single one-dimensional (trivial) irrep $\mcl^0$ (spanned by the identity) and a second, $(d^2-1)$-dimensional irrep. This second irrep is the \textit{adjoint representation} $\mcl^{\rm ad}$ of $\mfsu(d)$, although for our purposes this is not particularly important. In any event, we have
\begin{equation}~\label{eq:gud}
\mcl=\mcl^\DC \cong \mcl^0 \oplus \mcl^{\rm ad}.
\end{equation} 
Next we pick a Cartan subalgebra of $\mfu(d)$. First, we recall that the Lie algebra of the unitary group is given by (with $M_{d,d}(\mbc)$ the set of all $d\times d$ complex matrices)
\begin{equation}
    \mfu(d) = \big\{ X\in M_{d,d}(\mbc)\hspace{0.7mm} :\hspace{0.7mm} X^\dagger = -X\big\}.
\end{equation}
Let us take (with respect to, say, the computational basis) the Cartan subalgebra $\mfh$ to be given by the diagonal elements of $\mfu(d)$; it follows immediately that the computational basis\footnote{Which of course consists of simultaneous eigenvectors of matrices that are diagonal in the computational basis.} is, for this choice of $\mfh$, a (non-degenerate) weight basis, and therefore an NDCSE with respect to the maximal torus $T=\exp(\mfh)\subset\mathsf{U}(d)$. Evidently, $T$ consists of the diagonal unitaries, and plays the role of our abelian subgroup $H$. 
So, for this choice of $G$ and $\mfh$, shadows  reduces to the original proposal of Huang et al~\cite{huang2020predicting}: evolve with global unitaries  (alternately, as we momentarily discuss, global Cliffords, which produce the equivalent effect to global unitaries by virtue of their being a unitary 3-design) and measure in the computational basis. Let us explicitly verify that our expression for the measurement channel (Theorem~\ref{thm:mc}) agrees with the result derived in Ref~\cite{huang2020predicting}: 
\begin{equation}\label{eq:rh}
    \mcm(\rho) = \frac{\rho+\Tr[\rho]\id}{d+1}.
\end{equation}
To this end, we need to find the weight-zero subspaces of the representations  of Eq.~\eqref{eq:gud}. The case of $\mcl^0$ is trivial; we have (in the notation of Theorem~\ref{thm:mc}) $a_0=1$. The $T$-invariant subspace of $\mcl^{\rm ad}$, on the other hand, is spanned by the (non-identity) operators that commute with $T$ (i.e., all of the diagonal traceless operators); it is therefore  a $d-1$ dimensional subspace spanned by all tensor products of $I$ and $Z$ (except for $I^{\otimes n}$), and we have $a_{\rm ad}=(d-1)/(d^2-1)=1/(d+1)$. We can now readily verify that, in agreement with Eq.~\eqref{eq:rh},
\begin{align}
\mcm(\rho)&=
\left(a_{0}\mcp_0 + a_{\rm ad}\mcp_{\rm ad}\right)\rho\\
&=\left(\mcp_0 + \frac{\mcp_{\rm ad}}{d+1}\right)\rho\label{eq:gcc}\\
&=\frac{\Tr[\rho]\id}{d} + \frac{\rho-\Tr[\rho]\id/d}{d+1}\\
&= \frac{\rho+\Tr[\rho]\id}{d+1}.
\end{align}
Although the two protocols are  identical, the variance bounds derived through Theorem~\ref{thm:var} are, in this instance, weaker than obtained in Ref.~\cite{huang2020predicting}. In particular, by Theorem~\ref{thm:var}, we get $\mathrm{Var\ }\hat{o}\leq (d+1)\|O\|_2^2$ for any $O\in\mcl$, to be compared to $\mathrm{Var\ }\hat{o}\lesssim \|O\|_2^2$ from Ref.~\cite{huang2020predicting}. This discrepancy may be understood as a consequence of Theorem~\ref{thm:var} applying to \textit{any} group $G$; when specializing to a particular group improvements may be possible (indeed, we will see some examples of this throughout this section). \\

By virtue of their forming a unitary 3-design~\cite{mele2024introduction}, much of the preceding analysis goes through identically in the case of taking 
$G$ to be the Clifford group on $n$ qubits, i.e., the normaliser of the set $P_n$ of $n$-qubit Pauli strings, 
\begin{equation*}
\mathrm{Cl}(n)=\{U\in\U(d) : UP_nU^\dagger = P_n\},    
\end{equation*}
where $d=2^n$.   Nonetheless, let us see how things play out within our  framework. First of all, the Clifford's status as a unitary 3-design ensures that their action on $\mch$ is irreducible and that the  decomposition 
\begin{align}
\mcl^\DC = \LC &\cong \mcl^0 \oplus \mcl^{\rm ad}\label{eq:gcld}
\end{align}
of $\mcl^\DC$ into $\mathrm{Cl}(n)$-irreps is the same\footnote{In fact, this requires only the weaker notion of a unitary 2-design.} as that of the full unitary group (Eq.~\eqref{eq:gud}). As we saw in the main text, the computational basis is centralizing, and as in the case of the full unitary group, the variance bounds afforded by directly applying the generic Theorem~\ref{thm:var} are loose by a factor of $d$. As an example of the tightening of the bounds possible when more structure is imposed, consider the case where we are interested in a Pauli string $P$. Here we can employ Theorem~\ref{thm:varf}, as condition (ii) applies (because the Cliffords normalize the Paulis), and  we  conclude  
\begin{equation}
\mathrm{Var\ }\hat{p}\leq a_{\rm ad}^{-1}\|P\|_\infty^2=d+1,
\end{equation} 
recovering in this instance the scaling of Ref.~\cite{huang2020predicting}. 

\subsection{Local Unitaries}
Sticking with the case of an $n$-qubit system, we next consider $G=\U(2)^{\times n}$, i.e., the unknown state is evolved with (independent) random single qubit unitaries\footnote{To keep the amount of repetition tolerable we will not explicitly go through the very similar case of local Cliffords.}.  First,
as for the global unitaries,  we note that $\U(2)^{\times n}$ acts irreducibly on $\mch$, this time as the weight-module with highest weight $(1/2, 1/2, \ldots, 1/2)$ (weights are reviewed in Appendix~\ref{sec:weights}). So we again have $\mcl^\DC=\mcl$ and are in the tomographically complete case. Differences appear, however, when we turn to the structure of $\mcl$ as a $\U(2)^{\times n}$-module. Under the action of $\U(2)^{\times n}$, $\mcl$ decomposes into $d$ irreps, conveniently indexed by choices of subsets of qubits~\cite{zhao2024group}. That is, we have
\begin{equation}
\mcl^\DC \cong \bigoplus_{S\subseteq [n]} \mcl_S\,,
\end{equation}
where 
\begin{equation}
\mcl_S \coloneqq {\rm span}_{\hspace{0.5mm}\mbc}\left\{P\in P_n \ \vert\ {\rm supp}(P)=S\right\}
\end{equation}
is the $3^{|S|}$-dimensional space of operators that act non-trivially exactly on the qubits belonging to $S\subseteq [n]$. 
In the depiction of Fig.~\ref{fig:1}(b), therefore, global and local unitaries take the same form on the left hand side (one single block) but differ on the right hand side, decomposing respectively into two and $d$ blocks. \\

The Lie algebra of the local unitary ensemble is the ($4n$-dimensional) $\mfg\cong\mfu(2)^{\oplus n}$; for the Cartan subalgebra let us take the subspace spanned by the generators of $Z$ rotations on each qubit. The computational basis is then again a degeneracy-free weight basis, and therefore centralizing. This time, however, we only have a one-dimensional invariant-subspace of a given irrep $(\mcl^\DC)_S$ that commutes with these elements, namely that spanned by the Pauli string which acts as $Z$ on each of the $|S|$ qubits associated to the irrep. As the dimension of the irrep is $3^{|S|}$ we immediately obtain $a_S=3^{-|S|}$ and a measurement channel given by
\begin{equation}
\mcm = \sum_{S\subset [n]}3^{-|S|} \mcp_S,
\end{equation}
with $\mcp_S$ the projector onto the space $\mcl_S$ of operators that act non-trivially exactly on the qubits in $S$. By means of comparison with the result obtained for the measurement channel corresponding to local unitaries obtained in Ref.~\cite{huang2020predicting} we obtain the identity
\begin{equation}
    \mcd_{2/3}^{\otimes n} = \sum_{S\subset [n]}3^{-|S|} \mcp_S,
\end{equation}
where 
\begin{equation}\label{eq:depol}
\mcd_{p}(A) = \frac{p\Tr[A]\id}{2}+(1-p)A
\end{equation}
is a single-qubit depolarizing channel. So a Pauli string of weight $k$, for example, is by the measurement channel suppressed exponentially in its weight, by a factor of $3^{-k}$. This can be understood intuitively by means of the equivalence between local unitary and local Clifford shadows~\cite{huang2020predicting}: there are $3^k$ Pauli strings in the orbit of a  weight-$k$ Pauli string under the local Cliffords, with the probability that it is diagonalised by a random element of that group being $3^{-k}$~\cite{king2024triply}. Finally, from Theorem~\ref{thm:var}
we obtain, for an operator $O\in\mcl_S$ (say with $|S|=k$) variance bounds of $\Var\ \hat{o} \leq 3^{k}\|O\|_2^2$ and $\Var\ \hat{o} \leq 3^{2k}\|O\|_\infty^2$. Analogously to the previous examples, these bounds are quite loose  compared to those  derivable by performing a bespoke calculation tailored to the local unitary group explicitly, namely $\Var\ \hat{o}\leq 4^k \|O\|_\infty^2$~\cite{huang2020predicting}. 

\subsection{Pauli strings}\label{sec:paulis}
As our next example we take $G$ to be the group $P_n$ of $n$-qubit Pauli strings (including phase factors of $\pm 1,\ \pm i$), with measurements in the computational basis (which we see to be an NDCSE with respect to, say,  $H=\{I,Z\}\tn$). As we shall see, the induced CS protocol is in a sense trivial, possessing a visible space of mutually commuting operators. Nonetheless, from our perspective it is somewhat notable, constituting an extreme example of the framework.
First of all, we note that the action of $G$ on  $\mch$ is irreducible, corresponding to the well-known fact that the Pauli strings form a 1-design~\cite{mele2024introduction}. Their action on $\mcl=\mcl^\DC$, on the other hand, is maximally reducible: every Pauli string spans a one-dimensional representation of the adjoint action of $G$. Indeed, for all Pauli strings $P$ and $Q$, we have $P^\dagger QP\propto Q$. The action of $H$ on a given $P$ is trivial iff $P\in H$ (up to phases) so that the coefficients of Eq.~\eqref{eq:gb} are either zero or one, and  $\mcm  = \mcp_{{\rm span}\hspace{0.5mm} \{I,Z\}^{\otimes n}}$ (i.e., a fully dephasing channel).
In fact, we see that the additional assumptions of Theorem~\ref{thm:varf} apply, and so the variance of the estimator of any Pauli string are upper bounded by one. Of course, and as alluded to above, we see that this protocol is nonetheless  useless from a practical point of view, as its visible space is simply given by ${\rm span}\hspace{0.5mm} \{I,Z\}^{\otimes n}$; indeed if the operators in which one is interested are contained within  this space, one could simply measure the unknown state in the computational basis.

\subsection{Free fermions}\label{sec:ff}
We next consider the case where $G$ is the group of \textit{fermionic Gaussian unitaries} (FGU), which has received attention in the CS literature as an appropriate ensemble for measuring low degree fermionic observables~\cite{zhao2021fermionic,low2022classical,wan2023matchgate}. FGU contains unitaries of the form~\cite{zhao2021fermionic}
\begin{equation}
U_{\rm FGU} = \exp\left(-\frac{1}{4}\sum_{\mu,\nu=1}^{2n} A_{\mu,\nu}\gamma_\mu\gamma_\nu\right)\label{eq:fgu}\,,
\end{equation}
where $A=-A^\mathsf{T}$ is an anti-symmetric real matrix, and $\{\gamma_\mu\}_{1\leq\mu\leq 2n}$ is a set of \textit{Majorana fermions} satisfying $\gamma_\mu=\gamma_\mu^\dagger$ and the canonical anti-commutation relations $\{\gamma_\mu,\gamma_\nu\}=2\delta_{\mu,\nu}$. For example, one can take~\cite{diaz2023showcasing}
\begin{align}
\gamma_1 = XI\ldots I,\ \gamma_3 = ZXI\ldots I&,\ \gamma_{2n-1} = Z\ldots ZX \nonumber\\
\gamma_2 = YI\ldots I,\ \gamma_4 = ZYI\ldots &I,\ \gamma_{2n} = Z\ldots ZY\label{eq:majoranas} \,.
\end{align}
We will sometimes refer to the set of unitaries of the form Eq.~\eqref{eq:fgu} as \textit{matchgate circuits}. 
The FGU ensemble gives us our first example of a reducible action on $\mch$; specifically we have
\begin{equation}
\mch = \mch^+ \oplus \mch^-,
\end{equation}
where $\mch^{\pm}$ are respectively given by the span of the computational basis vectors with even or odd Hamming weight. Indeed, we can see that these subspaces are invariant under the dynamics of Eq.~\eqref{eq:fgu}; expanding the exponential gives a sum of terms which are a product of an even number of Majoranas, which from Eq.~\eqref{eq:majoranas} leads to a sum of Pauli strings which act as either $X$ or $Y$ on an even number of qubits (and $I$ or $Z$ on the remainder),  therefore preserving the parity of the Hamming weight of any computational basis state. $\mch^{\pm}$ are moreover irreducible, corresponding respectively to the $\mathfrak{so}(2n)$ modules\footnote{Strictly speaking, the fermionic Gaussian unitaries furnish the spinor representation of
$\mathrm{Spin}(2n)$, and their adjoint action on the Majorana operators $\{\gamma_\mu\}_{\mu=1}^{2n}$
factors through the standard $\mathrm{SO}(2n)$ action via the double cover $\mathrm{Spin}(2n)\to\mathrm{SO}(2n)$.} with highest weights $(1/2,1/2,\ldots,\pm1/2)$.
Let us take the Cartan subalgebra to be $\mfh=\spn_{\hspace{0.15mm}\mbr} \{\g_{2\mu-1}\g_{2\mu}\}_{1\leq\mu\leq n}$, i.e., the generators diagonal in the computational basis, $\{Z_i\}_{i=1}^n$.  
It follows from this choice and the description of the irreps of $\mch$ that the computational basis constitutes an NDCSE for $\mch$ {with respect to the corresponding maximal torus $H$}, and is therefore once again centralizing. \\

\noindent
The FGU measurement channel has previously been analytically diagonalised~\cite{zhao2021fermionic}; for a product
\begin{equation}
\Gamma_{\boldsymbol{\mu}} = (-i)^k \gamma_{\mu_1}\gamma_{\mu_2}\cdots\gamma_{\mu_{2k}}    
\end{equation}
of an even  number $2k$ of Majoranas one has
\begin{equation}\label{eq:fguc}
    \mcm(\Gamma_{\boldsymbol{\mu}}) = a_{\boldsymbol{\mu}}^H \Gamma_{\boldsymbol{\mu}}
\end{equation}
with
\begin{equation}\label{eq:fa}
    a_{\boldsymbol{\mu}}^H =\binom{n}{k} \bigg/ \binom{2n}{2k},
\end{equation}
while products of odd degree are annihilated. As polynomials in the Majoranas span the entire operator space $\mcl$~\cite{diaz2023showcasing},  this completely characterizes the action of the channel. The fact that the odd degree products lie outside the visible space corresponds to physical fermionic operators being of even degree due to parity superselection~\cite{StreaterWightman2001};  the FGU protocol is therefore tomographically complete with respect to the physically realisable operators. This is also readily interpreted within viewpoint of Theorem~\ref{thm:mc}: we note that, given a product of an odd number of Majoranas, there is necessarily an odd number of qubits which are acted upon by either an $X$ or a $Y$ gate, therefore flipping the parity of the Hamming weight of any computational basis state. Such a product is therefore not in ${\rm End}\,\,\mch^+$ or ${\rm End}\,\,\mch^-$, mapping as it does $\mch^\pm\mapsto\mch^\mp$. \\

Turning to the case of $\Gamma_{\boldsymbol{\mu}}$ a product of an even number $2k$ of Majoranas, we wish to verify that the action Eq.~\eqref{eq:fguc} of the FGU channel is consistent with our expression Eq.~\eqref{eq:gb}. As in the previous cases, for this we need to understand the decomposition of $\mcl^\DC$ into irreducible $G$ modules. In general, it is known~\cite{diaz2023showcasing} that we have
\begin{equation}
    \mcl = \bigoplus_{\kappa=0}^{2n} \mcl_\kappa\,,
\end{equation}
where $\mcl_\kappa = \spn\{ \g^\a \}_{\a \subseteq \binom{[2n]}{\kappa}}$ is the $\binom{2n}{\kappa}$-dimensional subspace spanned by the set of homogeneous polynomials of degree $\kappa$ in the Majoranas. These subspaces are irreducible with the exception of $\mcl_n$, which decomposes into irreps as $\mcl_n=\mcl_n^+\oplus\mcl_n^-$ (both of dimension $\frac{1}{2}\binom{2n}{n}$), corresponding to the $\pm 1$ eigenspaces of the parity operator ($Z^{\otimes n}$ in the representation specified by our choice Eq.~\eqref{eq:majoranas} of Majoranas). 
\begin{figure}
    \includegraphics[width=0.45\linewidth]{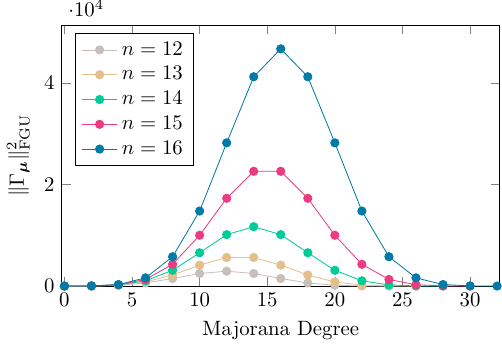}
    \caption{The (squared) shadow norm for Majorana monomials of fixed (even) degree for various values of $n$, when sampling uniformly from the distribution of fermionic Gaussian unitaries. As can be seen by applying Stirling's inequality to Eqs.~\eqref{eq:fa} and~\eqref{eq:fgusn}, monomials of degree  $\sim n$ will possess an intractably large shadow norm of order $\mco(2^n)$. Monomials of degree of order $\mco(1)$ (mod $2n$), on the other hand, can be efficiently learnt. }
    \label{fig:maja}
\end{figure}
From the previous discussion we then have the decomposition
\begin{equation}\label{eq:ffldirreps}
    \mcl^\DC = \begin{cases}
        \bigoplus_{\kappa=0}^{n} \mcl_{2\kappa}, \quad &n {\rm\ odd}\\
        \left(\bigoplus_{\kappa=0,\kappa\neq n/2}^{n} \mcl_{2\kappa}\right)\oplus \mcl_n^+\oplus\mcl_n^-, \quad &n {\rm\ even}
    \end{cases} 
\end{equation}
of $\mcl^\DC$ into $G$-irreps. Each irreducible subspace $\mcl_{2\kappa}$ (with $2\kappa\neq n$) contains an $\binom{n}{\kappa}$-dimensional weight-zero subspace consisting of the (simultaneously diagonalisable) products of $\kappa$ distinct occupation number operators, leading to (in the notation of Eq.~\eqref{eq:gb})
\begin{equation}
    a_{\mcl_{2\kappa}}^H = \frac{d^H_{\mcl_{2\kappa}}}{d_{\mcl_{2\kappa}}} =\binom{n}{\kappa} \bigg/ \binom{2n}{2\kappa},
\end{equation}
exactly as in Eq.~\eqref{eq:fa}. Similarly, in the case of $\mcl_n^\pm$ we have a weight-zero subspace consisting of the $\frac{1}{2}\binom{n}{n/2}$-dimensional subspace given by operators of the form $P\pm Z^{\otimes n}P$, where $P$ is a product of $n/2$ distinct occupation number operators, yielding
\begin{equation*}
    a_{\mcl_{n}^\pm}^H = \frac{d^H_{\mcl_{n}^\pm}}{d_{\mcl_{n}^\pm}} =\frac{1/2}{1/2}\binom{n}{n/2} \bigg / \binom{2n}{n}=\binom{n}{n/2} \bigg / \binom{2n}{n}.
\end{equation*}
In the case of a Majorana monomial $\Gamma_{\boldsymbol{\mu}}$ of (even) degree $n$ we therefore also recover the expected behaviour, 
\begin{align*}
\mcm(\Gamma_{\boldsymbol{\mu}})&=\mcm\left(\frac{1}{2}(\Gamma_{\boldsymbol{\mu}}+Z^{\otimes n}\Gamma_{\boldsymbol{\mu}})+\frac{1}{2}(\Gamma_{\boldsymbol{\mu}}-Z^{\otimes n}\Gamma_{\boldsymbol{\mu}})\right)\\
&=\frac{a_{\mcl_{n}^+}^H}{2}(\Gamma_{\boldsymbol{\mu}}+Z^{\otimes n}\Gamma_{\boldsymbol{\mu}})+\frac{a_{\mcl_{n}^-}^H}{2}(\Gamma_{\boldsymbol{\mu}}-Z^{\otimes n}\Gamma_{\boldsymbol{\mu}})\\
&=\binom{n}{n/2} \bigg / \binom{2n}{n} \Gamma_{\boldsymbol{\mu}}\,.
\end{align*}
So our approach to calculating the measurement channel agrees with the FGU channel as previously derived on all Majorana monomials, and therefore, by linearity, on all operators.\\

\noindent
Finally, it has been shown~\cite{zhao2021fermionic} that under the FGU channel the shadow norm is (see Fig.~\ref{fig:maja})
\begin{equation}\label{eq:fgusn}
    \|\Gamma_{\boldsymbol{\mu}}\|^2_{\rm FGU} =(a_{\boldsymbol{\mu}}^H) ^{-1}\,,
\end{equation}
with $a_{\boldsymbol{\mu}}^H$ as in Eq.~\eqref{eq:fa}. This scaling is not captured by Theorem~\ref{thm:var}, which gives a bound of $(a_{\boldsymbol{\mu}}^H) ^{-2}$ (note that $\|\Gamma_{\boldsymbol{\mu}}\|^2_{\infty}=1 $). Happily, however, the action of FGUs on Majoranas is equivalent to the action of \textit{Clifford} FGUs~\cite{wan2023matchgate}, which satisfies (for $\kappa\neq n/2$) condition (\textit{ii}) of Theorem~\ref{thm:varf} (as the Majoranas are Pauli strings that span their respective modules) allowing us to recover in this instance the tighter bound of $(a_{\boldsymbol{\mu}}^H) ^{-1}$. This is discussed in Appendix~\ref{sec:tf}.  In the case of even $n$ and a product $\Gamma_{\boldsymbol{\mu}}$ of Majoranas in the $\kappa=n/2$ component we recover (up to a constant factor) the same scaling; defining $\Gamma_{\boldsymbol{\mu}}^\pm = \Gamma_{\boldsymbol{\mu}}\pm Z^{\otimes n}\Gamma_{\boldsymbol{\mu}}$ we have:
\begin{align*}
\|\Gamma_{\boldsymbol{\mu}}\|_{{\rm shad}}^2&=\left\|\frac{\Gamma_{\boldsymbol{\mu}}^++\Gamma_{\boldsymbol{\mu}}^-}{2}\right\|_{{\rm shad}}^2\\
&\leq\frac{1}{4}\left(\|\Gamma_{\boldsymbol{\mu}}^+\|_{{\rm shad}}+\|\Gamma_{\boldsymbol{\mu}}^-\|_{{\rm shad}}\right)^2\\
&=\frac{1}{4}\binom{2n}{n} \Big / \binom{n}{n/2}\left(\|\Gamma_{\boldsymbol{\mu}}^+\|_{\infty}+\|\Gamma_{\boldsymbol{\mu}}^-\|_{\infty}\right)^2\\
&\leq 4\binom{2n}{n}\Big / \binom{n}{n/2}\,,
\end{align*}
where we have applied the triangle inequality followed by Theorem~\ref{thm:varf} to the $G$-irreps $\mcl_n^\pm$.  Finally, by the submultiplicativity of the spectral norm and again by the triangle inequality, we have used $\|\Gamma_{\boldsymbol{\mu}}^\pm\|_{\infty}\leq\|\Gamma_{\boldsymbol{\mu}}\|_{\infty}+\|Z^{\otimes n}\|_{\infty}\|\Gamma_{\boldsymbol{\mu}}\|_{\infty}= 2$. So in this case our bound on the variance is only weaker by a factor of four compared to the bespoke calculation of Ref.~\cite{zhao2021fermionic}.


\subsection{Symplectic shadows}\label{sec:symp}

\noindent 
For our next example we take the ensemble of unitaries to be the uniform measure over the unitary symplectic group 
\begin{equation}\label{eq:symp_def}
\SP(d/2) = \{ U\in\U(d) : U^{\mathsf{T}} \Omega U = \Omega\},
\end{equation}
where $\Omega$ is a skew-symmetric bilinear form that we take to be
\begin{equation} \label{eq:omega}
	\Omega=iY\otimes I^{\otimes(n-1)}=\begin{pmatrix} 0& \id_{d/2} \\ - \id_{d/2} & 0\end{pmatrix}\,,
\end{equation}
with $\id_{d/2}$ the $d/2 \times d/2$ identity matrix. The symplectic shadows protocol (with measurement in the computational basis\footnote{By the transitivity~\cite{albertini2001notions,oszmaniec2017universal} of the action of $\SP(d/2)$ on $\mch$ and the invariance of the Haar measure, the choice of basis is immaterial.})  has been investigated in Ref.~\cite{west2024random} and found to be equal to that obtained from global Cliffords; as in the previous sections, our goal is to derive this from the NDCSE perspective. First we note that by the transitivity~\cite{albertini2001notions,oszmaniec2017universal} of the action of $\SP(d/2)$ on $\mch$, $\mch$ is certainly irreducible as an $\SP(d/2)$-module, so that we are back in the tomographically complete case $\mcl^\DC=\mcl$. As we shall presently see, the decomposition of $\mcl$ into $\SP(d/2)$-irreps takes the form 
\begin{equation}\label{eq:symp_decompm}
\mcl^\DC=\mcl = {\rm Sym}^2\hspace{0.5mm}\mch\hspace{0.7mm} \oplus \hspace{0.7mm}W\hspace{0.7mm}\oplus\hspace{0.7mm} \mbc.
\end{equation}
where $W\oplus\mbc=\Lambda^2\hspace{0.5mm}\mch$, and  ${\rm Sym}^2\hspace{0.5mm}\mch$ and $\Lambda^2\hspace{0.5mm}\mch$ are simply\footnote{Strictly, $\mcl\cong\mch\otimes\mch^*$, whereas ${\rm Sym}^2\hspace{0.5mm}\mch, \Lambda^2\hspace{0.5mm}\mch\subset \mch\otimes\mch$. However, in this case $\mch\cong\mch^*$ (via the identification offered by the bilinear form $\Omega$, namely $v\in\mch\mapsto \Omega(v,-)\in\mch^*$); when we write (e.g.) ${\rm Sym}^2\hspace{0.5mm}\mch\subset\mcl$, we refer technically to its image under this isomorphism.  } the familiar (anti-)symmetric second order tensor products of multilinear algebra,
\begin{equation}
{\rm Sym}^2\hspace{0.5mm}\mch = \frac{ \mch\otimes\mch}{\langle v\otimes w-w\otimes v : v,w\in \mch\rangle}
\end{equation}
and 
\begin{equation}
\Lambda^2\hspace{0.5mm}\mch = \frac{ \mch\otimes\mch}{\langle v\otimes w+w\otimes v : v,w\in \mch\rangle}.
\end{equation}
At first glance, comparing Eqs.~\eqref{eq:gcld} and~\eqref{eq:symp_decompm} appears to suggest different shadow protocols, in contradiction with Ref.~\cite{west2024random}. The resolution comes when we calculate the coefficients $a_\lm^H$ for the irreps of Eq.~\eqref{eq:symp_decompm}; leaving the details for the end of the section, we find\footnote{Here and henceforth $d^0_\lm$ will denote the weight-zero subspace of the irrep $\lm$, i.e. the invariant space of the corresponding maximal torus, our $H$ when measuring in a weight basis.}
\begin{equation}\label{eq:asym}
a_{{\rm Sym}^2\hspace{0.15mm}\mch}=\frac{d_{{\rm Sym}^2\hspace{0.15mm}\mch}^0}{d_{{\rm Sym}^2\hspace{0.15mm}\mch}}=\frac{d/2}{d(d+1)/2}=\frac{1}{d+1}
\end{equation}
and
\begin{equation}\label{eq:aw}
    a_W = \frac{d_{W}^0}{d_W}=\frac{d/2-1}{d(d-1)/2-1}= \frac{1}{d+1},
\end{equation}
resolving the confusion: as in the global Clifford case, the symplectic shadows measurement channel indeed acts trivially on a one-dimensional subspace of $\mcl$, and compresses the remainder by a factor of $d+1$ (c.f. Eq.~\eqref{eq:gcc}). From a representation-theoretic perspective this might be attributed to the ``accidental'' equality of the terms in Eqs.~\eqref{eq:asym} and~\eqref{eq:aw}; from the perspective of Ref.~\cite{west2024random} it is guaranteed by the fact that the state ensembles generated by $\U(d)$ and $\SP(d/2)$ are \textit{indistinguishable}.\\

Let us now turn to the  technical details of the calculation.
We are considering the CS protocol obtained by choosing the ensemble of unitaries to be the (unitary) symplectic group $\SP$  equipped with its uniform measure, and the measurement basis to be the computational basis. As discussed in Appendix~\ref{sec:wc} (Eqs.~\eqref{eq:k2},~\eqref{eq:k3}), the protocol is completely determined by the $k=2,3$ state-moments 
\begin{equation}
\mce_{\SP(d/2)}^{(k)}(\Pi_w\tk)=\int_{U\in {\SP(d/2)}}   d\mu_{\SP(d/2)}(U)\   U^{\dagger\otimes k} \Pi_w^{\otimes k} U^{\otimes k};
\end{equation}
as these have recently been shown to be equal to the corresponding integrals over $\U(d)$~\cite{west2024random}, we know \textit{a priori} that symplectic shadows will simply reproduce the usual unitary shadows\footnote{Note that by the transitivity of the action of $\SP(d/2)$ on $\mch$~\cite{albertini2001notions,oszmaniec2017universal} and the invariance of the Haar measure, the choice of measurement basis cannot affect the protocol.}. The goal of the remainder of this section will be to make use of the representation theory of the symplectic Lie algebra to derive this within our framework. Many more details of this representation theory may be found in (e.g.) Ref.~\cite{fulton1991representation}. \\

\noindent
Now, differentiating Eq.~\eqref{eq:symp_def} yields the condition for belonging to the symplectic Lie algebra\footnote{Well its complexification, anyway},
\begin{equation}
   \mfsp(d/2) = \{ X\in {\rm End\hspace{0.5mm} \mch} : X^{\mathsf{T}} \Omega + \Omega X = 0\}.
\end{equation}
One readily checks that under the choice Eq.~\eqref{eq:omega} of $\Omega$ this translates to the condition 
\begin{equation}
    X=\begin{pmatrix}
        A&B\\C&-A^{\mathsf{T}}
    \end{pmatrix}
\end{equation}
for $A,\ B=B^{\mathsf{T}}$ and $C=C^{\mathsf{T}}$ some $d/2\times d/2$ matrices. 
In this basis a natural choice of Cartan subalgebra is
\begin{equation}\label{eq:symp_csa}
    \mfh = {\rm span}_{\hspace{0.5mm}\mbc} \{\hspace{0.5mm} H_i = E_{i, i}-E_{d/2+i, d/2+i}   \hspace{0.5mm}\} _{1\leq i\leq d/2}.
\end{equation}
We also introduce the basis $\{L_j\}_j$ of $\mfh^*$ dual to the above spanning set of $\mfh$, i.e., satisfying $L_j(H_i)=\delta_{i,j}$.\\ 

Equipped with these definitions, we begin by considering the standard representation $\mch$ of $\mfsp(d/2)$ as the endomorphism algebra acting on $\mch$. One of the nice things about our choice Eq.~\eqref{eq:symp_csa} of Cartan subalgebra is that a basis of degeneracy-free weight vectors is clear; indeed we can simply take the computational basis (i.e., that basis with respect to which $\Omega$ takes the form Eq.~\eqref{eq:omega} that it does). Instead of the usual binary notation, it will be convenient to label the elements of this basis as $e_1,\hspace{0.5mm} e_2,\hspace{0.5mm} \ldots,\ e_d$. Further, from Eq.~\eqref{eq:symp_csa} we can immediately read off the weights; the weight $\lm_i$ of $e_i$ is simply given by
\begin{equation}\label{eq:symp_weights}
    \lm_i=\begin{cases}
        L_i & i\leq d/2\\
        -L_{i-d/2} & i>d/2
    \end{cases}\ \ ,
\end{equation}
so that the total set of weights for the representation $\mch$ is given by $\{\hspace{0.5mm}\pm L_i\hspace{0.5mm}\}_{1\leq i\leq d/2}$. A convenient graphical representation of the weights of the representations encountered in this and the next appendix in the case $d=4$ may be found in Fig.~\ref{fig:weights}.
Now, the fact that every weight appears with its negative implies that the standard representation of $\mfsp(d/2)$ is \textit{self-dual} (because the weights of the dual representation $V^*$ of a representation $V$ are the negative of the weights of $V$~\cite{fulton1991representation}). This is immediately useful: as the action of $\SP$ on $\mch$ is irreducible (moreover, transitive~\cite{albertini2001notions,oszmaniec2017universal}) what we care about for the purposes of $G$-shadows is the representation $\mcl^\DC=\mcl=\endh$, which is therefore related to $\mch$ by
\begin{equation}
\mcl=\mch^*\otimes\mch=\mch\otimes\mch = {\rm Sym}^2\hspace{0.5mm}\mch\hspace{0.7mm} \oplus\hspace{0.5mm}\Lambda^2\hspace{0.5mm}\mch \,,
\end{equation}
where the final equality is simply the usual decomposition from multilinear algebra\footnote{Note that the symplectic form (an element of $(\mch^*)^{\otimes 2}$) induces an isomorphism $\mch\cong\mch^*$; i.e., denoting  in relativistic notation elements of $\mch\ (\mch^*)$ as objects with upper (lower) indices, one can take $v_a=\Omega_{ab}v^b$.}. It turns out that the representation $\Lambda^2\hspace{0.5mm}\mch $ is \textit{not} irreducible, but rather contains a copy of the trivial representation $\mbc$ (corresponding to the trivial action of $\SP$ on $\Omega\in\Lambda^2\hspace{0.5mm}\mch^*\cong\Lambda^2\hspace{0.5mm}\mch$~\cite{fulton1991representation}); denoting by $W$ its complement in $\Lambda^2\hspace{0.5mm}\mch $ the decomposition of $\mcl$ into $\mfsp(d/2)$ irreps is 
\begin{equation}\label{eq:symp_decomp}
\mcl = {\rm Sym}^2\hspace{0.5mm}\mch\hspace{0.7mm} \oplus \hspace{0.7mm}W\hspace{0.7mm}\oplus\hspace{0.7mm} \mbc.
\end{equation}
What we would like to do, then, is calculate the dimensions of ${\rm Sym}^2\hspace{0.5mm}\mch$ and $W$, as well as their weight-zero subspaces. Let's begin with ${\rm Sym}^2\hspace{0.5mm}\mch$. We have a basis of weight vectors given by $e_i\otimes e_j$, $1\leq i\leq j\leq d$, with corresponding weights $\pm L_i\pm L_j$ (where from Eq.~\eqref{eq:symp_weights} the signs are determined by whether or not the index is greater than $d/2$). With a zero weight therefore resulting from pairing each positive weight $L_i$ of $\mch$ with its negative, we obtain a $d/2$-dimensional weight-zero space (see Fig.~\ref{fig:weights}). Combined with ${\rm dim \hspace{0.5mm}}\left({\rm Sym}^2\hspace{0.5mm}\mch\right) = d(d+1)/2$ we conclude 
\begin{equation}
a_{{\rm Sym}^2\hspace{0.15mm}\mch}^{H}=\frac{d/2}{d(d+1)/2}=\frac{1}{d+1}.
\end{equation}
The case of $\Lambda^2\hspace{0.5mm}\mch$ is highly similar; we have a basis of weight vectors given by $e_i\otimes e_j$, $1\leq i< j\leq d$, with corresponding weights $\pm L_i\pm L_j$, and again a $d/2$-dimensional weight-zero space. One of these is of course the trivial representation $\mbc$; $W$ then possesses a $(d/2 -1)$-dimensional weight-zero space, while itself being of dimension ${\rm dim\hspace{0.5mm}}W=d(d-1)/2 -1$. This yields
\begin{equation}
    a_W^{H} = \frac{d/2-1}{d(d-1)/2-1}= \frac{1}{d+1}.
\end{equation}
Combined with (of course) $a_\mbc = 1/1=1$ we find that the action of the symplectic shadows  measurement channel on the non-trivial component of $\mcl$ is simply to divide by $d+1$, agreeing as expected with the case of unitary shadows.

\begin{figure}
\centering
\includegraphics[width=0.9\linewidth]{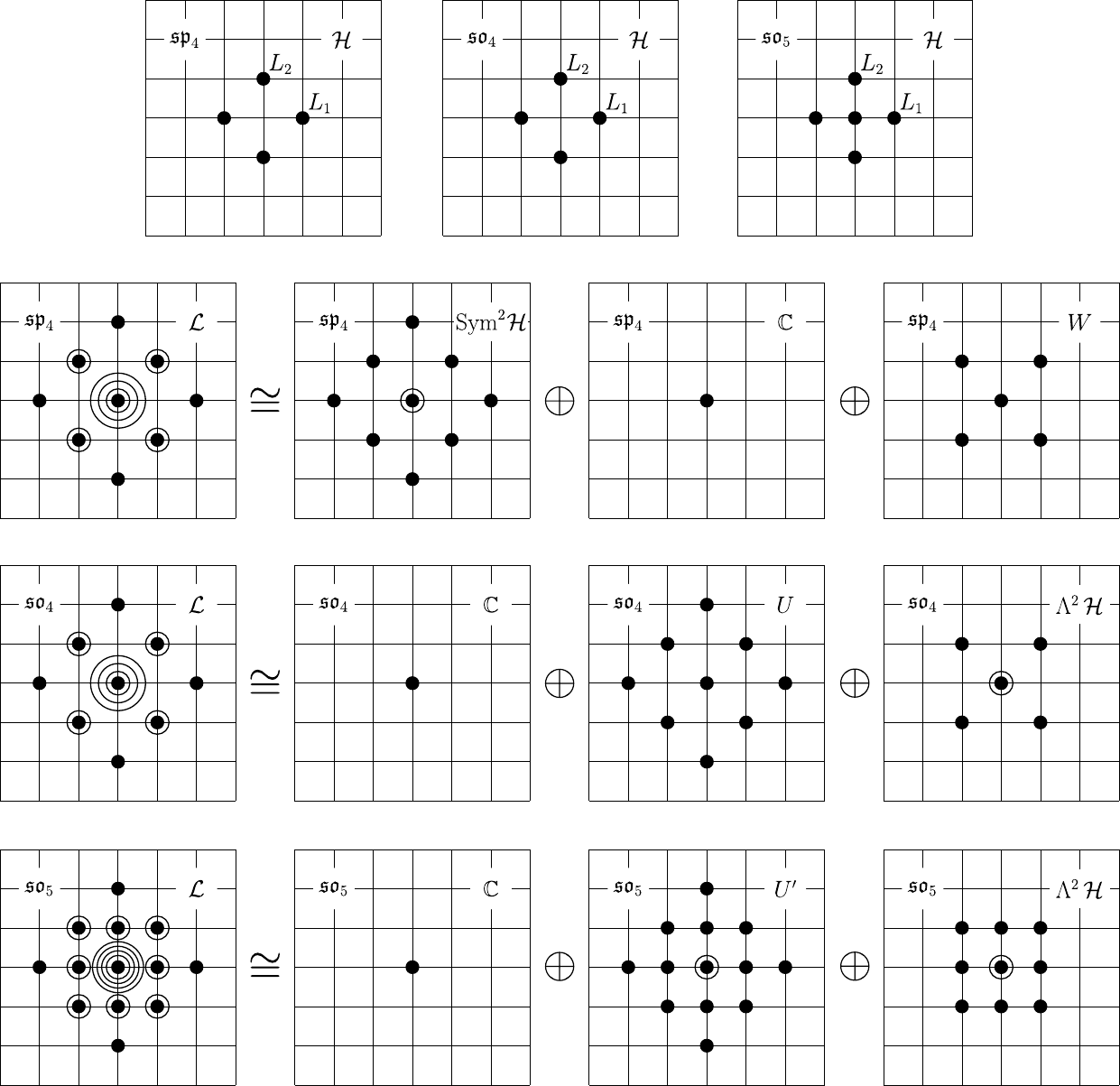}
\caption{In the low-dimensional cases of $d=4,5$ we can visualize the representations that appear in the study of symplectic and orthogonal shadows by means of \textit{weight-diagrams} (in general, for $\mfsp_d$ and $\mfso_d$ the diagrams will live in a ${\rm dim\ \mfh^*}=\floor{d/2}$ dimensional space, so that the extent to which they clarify the situation decreases sharply with $d$). The fact that the geometry (in particular, lengths and angles) implied by the diagrams to exist on $\mfh^*$ behaves sensibly follows from the general theory of representations of Lie algebras~\cite{fulton1991representation}. In the top row we plot the weights of the standard representations $\mch$ of $\mfsp_4$, $\mfso_4$ and $\mfso_5$, and in the following rows the decomposition into irreps of $\mcl^\DC=\mcl$. Degeneracies are denoted by circles, and the centres of the diagrams represent the weight-zero spaces. In the $\mfsp_4$ row we find $a_{{\rm Sym}^2\hspace{0.15mm}\mch}=2/10=1/5=a_W$ (with the equality as expected from Appendix~\ref{sec:symp}); for $\mfso_4$, $a_{\Lambda^2\hspace{0.15mm}\mch}=2/6\neq 1/9=a_U$, and for $\mfso_5$, $a_{\Lambda^2\hspace{0.15mm}\mch}=2/10\neq 2/14=a_{U'}$. One can easily check that these agree with our general formulae. }
    \label{fig:weights}
\end{figure}

\subsection{Orthogonal shadows}\label{sec:mortho}
We now come to the case of uniformly sampling from $G=\O(d)$. Orthogonal CS were recently explored in Ref.~\cite{west2025real}, where it was found that random global orthogonal unitaries coupled with measurement in the computational basis yields the measurement channel
\begin{equation}\label{eq:gomc}
  \mcm(A)  = \frac{\Tr[A]\mathds{1}+A+ A^\mathsf{T} }{d+2} = (\mcd_{d/(d+2)}\circ\mcp_{\rm sym})(A)\,,
\end{equation}
i.e., a depolarizing channel (see Eq.~\eqref{eq:depol}) on the symmetric component $\mcp_{\rm sym}(A) = (A+ A^\mathsf{T})/2$ of its input. In this case (and for the first time in the examples we have presented so far)  the computational basis is \emph{not} a weight basis for the ensemble considered in Ref.~\cite{west2025real}. To see this, recall that the orthogonal group is defined as the operators that preserve a symmetric bilinear form $Q$, 
\begin{equation}\label{eq:ortho_def}
   \O(d) = \{ O\in\U(d) : O^{\mathsf{T}} Q O = Q\}.
\end{equation}
Typically, of course, one simply takes $Q=\id$ (which indeed was done in Ref.~\cite{west2025real}) but let us leave it unspecified for now. At the Lie algebra level, the condition Eq.~\eqref{eq:ortho_def} becomes (via differentiating)
\begin{equation}\label{eq:oalg}
   \mfso(d) = \{ X\in {\rm End\hspace{0.5mm} \mch} : X^{\mathsf{T}} Q + Q X = 0\},
\end{equation}
so that, if $Q=\id$, the Lie algebra is represented by skew-symmetric matrices. This means, however, that there are no (non-zero) diagonal matrices in this instantiation of $\mfso(d)$, and that our previous strategy of taking the Cartan subalgebra to be the diagonal elements of $\mfg$ (immediately implying that the standard basis is a weight basis) fails. So, to employ our framework, we need to work a little harder than in the previous examples\footnote{We will not use it here, but in fact one can show (with $d$ even, say) that the linear combinations $\{\ket{i}\pm\ket{i+d/2}\}_{0\leq i<d/2}$ of the standard basis elements is an NDCSE for the choice $Q=\id$ of bilinear form.}. One approach is to change our choice of $Q$ so that the standard basis \textit{is} a (non-degenerate) weight basis of the resulting copy of $\O$; a $Q$ that does the job is (we assume for the moment that $d$ is even; the (very similar) case of $d$ odd is discussed at the end of this section)
\begin{equation} \label{eq:q_def}
	Q=\begin{pmatrix} 0& \id_{d/2} \\  \id_{d/2} & 0\end{pmatrix},
\end{equation}
where we relegate the details (and the justification of the next  sentence) to the end of this section. The action of $\O(d)$ on $\mch$ is irreducible and self-dual;  the decomposition of the $\O(d)$-diagonal operators into $\O(d)$-irreps is given by (with $U$ the orthogonal complement of the identity in space corresponding to ${\rm Sym}^2\hspace{0.5mm}\mch$ under the isomorphism\footnote{The isomorphism between $\mch$ and $\mch^*$ is facilitated by the bilinear form $Q$ itself; $v\in\mch\mapsto Q(v,-)\in\mch^*$.} $\mcl\cong\mch\otimes\mch^*\cong\mch\otimes\mch$)
\begin{equation}\label{eq:mortho_decomp}
\mcl^\DC= \mcl=\mbc \hspace{0.7mm} \oplus \hspace{0.7mm}U\oplus\hspace{0.7mm} \Lambda^2\hspace{0.5mm}\mch,
\end{equation}
with $a_\mbc=1$,
\begin{equation}
a_{\Lambda^2\hspace{0.15mm}\mch}=\frac{d_{\Lambda^2\hspace{0.15mm}\mch}^0}{d_{\Lambda^2\hspace{0.15mm}\mch}}=\frac{d/2}{d(d-1)/2}=\frac{1}{d-1},
\end{equation}
and
\begin{equation}
    a_U =\frac{d_U^0}{d_U}= \frac{d/2-1}{d(d+1)/2-1}= \frac{d-2}{(d+2)(d-1)}.
\end{equation}
By Theorem~\ref{thm:mc} the measurement channel is 
\begin{equation}
\mcm(A) = \frac{\Tr[A]}{d}\id + \frac{d-2}{(d+2)(d-1)}\mcp_{U} + \frac{1}{d-1}\mcp_{\Lambda^2\hspace{0.15mm}\mch}\,,
\end{equation}
and therefore meaningfully different from the channel in Eq.~\eqref{eq:gomc} (for example, the kernel of one is trivial, and the other not).  From Theorem~\ref{thm:var} we obtain the variance bound (assuming $d\gg 1$ for simplicity)
$\Var[\hat{o}] \lesssim d \|O\|^2_2$.
As usual, one expects to be able to do better by taking into account some properties of the group; indeed, later in this section we will evaluate the relevant integrals via Weingarten calculus, obtaining (for $d\gg 1$) the tighter bound
\begin{equation}
\Var[\hat{o}] \lesssim  \|O\|^2_2,
\end{equation}
the same scaling as for global unitaries.
We note in passing that, keeping the subleading terms in $d$, the bound is very slightly better than the unitary case for operators in ${\Lambda^2\hspace{0.15mm}\mch}$ and very slightly worse for those in $U$. In contrast, the bound of the real shadows protocol of Ref.~\cite{west2025real} is the improved $\Var[\hat{o}] \lesssim (1/2) \|O\|^2_2$; the price paid in that instance is that only the symmetric components of operators are visible, corresponding to a visible space of dimension approximately half that of the full operator space.\\

Next we note that although, as discussed, it does not involve measuring in an weight basis, the real shadows protocol of Ref.~\cite{west2025real} does still employ measurements in an NDCSE. Indeed, consider the subgroup $H=\{{\rm diag}\,(\varepsilon_1,\varepsilon_2,\ldots, \varepsilon_d)\,|\,\varepsilon_i\in\{\pm 1\}\ \forall i\}\subset \O(d)$. Within the irrep of traceless symmetric operators we have a $(d-1)$-dimensional $H$-invariant subspace (the traceless diagonal operators), and within $\Lambda^2\,\mch$ there are no $H$-invariants. By Theorem~\ref{thm:mc} we therefore obtain (with tildes simply to distinguish this instance from the previous calculation):
\begin{align*}
&\mcm (A) =  \widetilde{a}_\mbc \mcp_{\mbc}(A)  + \widetilde{a}_U\mcp_{U}(A) + \widetilde{a}_{\Lambda^2\hspace{0.15mm}\mch}\mcp_{\Lambda^2\hspace{0.15mm}\mch}(A)\\
&= \mcp_{\mbc}(A)+ \frac{d-1}{d(d+1)/2-1}\mcp_{U}(A) + \frac{0}{d(d-1)/2}\mcp_{\Lambda^2\hspace{0.15mm}\mch}(A)\\
&=  \frac{\Tr[A]}{d}\id + \frac{d-1}{d(d+1)/2-1}\left(\frac{A+A^{\mathsf{T}}-2\Tr[A]\id/d}{2}\right) \\
&=\frac{\Tr[A]\mathds{1}+A+ A^\mathsf{T} }{d+2}\,,
\end{align*}
in agreement with Eq.~\eqref{eq:gomc} and Ref.~\cite{west2025real} (after exerting substantially less effort than required in that reference).\\

Let us now come to the technical derivations of the above statements; 
we will essentially repeat the calculation of the symplectic case for the case of the CS protocol resulting from sampling Haar randomly from the orthogonal group and measuring in a weight basis; the clear similarities with the former calculation will reflect the close connection between the representation theory of the orthogonal and symplectic algebras~\cite{fulton1991representation}. 
One readily checks that under the choice Eq.~\eqref{eq:q_def} of $Q$ Eq.~\eqref{eq:oalg} translates to the condition 
\begin{equation}
    X=\begin{pmatrix}
        A&B\\C&-A^{\mathsf{T}}
    \end{pmatrix}
\end{equation}
for $A,\ B=-B^{\mathsf{T}}$ and $C=-C^{\mathsf{T}}$ some $d/2\times d/2$ matrices. 
In this basis a natural choice of Cartan subalgebra is (c.f. Eq.~\eqref{eq:symp_csa}
\begin{equation}\label{eq:ortho_csa}
    \mfh = {\rm span}_{\hspace{0.5mm}\mbc} \{\hspace{0.5mm} H_i = E_{i,i}-E_{d/2+i,d/2+i}   \hspace{0.5mm}\} _{1\leq i\leq d/2}\, ;
\end{equation}
the existence of such a favourable basis of $\mch$ justifies our choice of $Q$; things would not be so nice were $Q=\id$. Completing the setup,
we again introduce the basis $\{L_j\}_j$ of $\mfh^*$ dual to the above spanning set of $\mfh$, i.e., satisfying $L_j(H_i)=\delta_{i,j}$.\\ 

From the equivalence of Eq.~\eqref{eq:ortho_csa} and Eq.~\eqref{eq:symp_csa} we can copy the next part of the analysis from the symplectic case all but verbatim. We have an NDCSE $e_1,\hspace{0.5mm} e_2,\hspace{0.5mm} \ldots,\ e_d$ of $\mch$ given by the standard basis with respect to which  $Q$ takes the form Eq.~\eqref{eq:q_def} that it does. The weights are again (the non-degenerate)
\begin{equation}\label{eq:ortho_weights}
    \lm_i=\begin{cases}
        L_i & i\leq d/2\\
        -L_{i-d/2} & i>d/2
    \end{cases}\ \  ,
\end{equation}
so that the total set of weights for the representation $\mch$ is given by $\{\hspace{0.5mm}\pm L_i\hspace{0.5mm}\}_{1\leq i\leq d/2}$.
As before, the fact that every weight appears with its negative implies that the standard representation of $\mfso(d)$ is self-dual; combined with the irreducibility (but \textit{not}, this time, transitivity) of the action of $\O$ on $\mch$, we have
\begin{equation}
\mcl^\DC=\mcl=\mch^*\otimes\mch=\mch\otimes\mch = {\rm Sym}^2\hspace{0.5mm}\mch\hspace{0.7mm} \oplus\hspace{0.7mm}\Lambda^2\hspace{0.5mm}\mch. 
\end{equation}
This time it is the representation ${\rm Sym}^2\hspace{0.5mm}\mch $ that is not irreducible, but rather contains a copy of the trivial representation $\mbc$ (corresponding to the trivial action of $\O$ on $Q\in{\rm Sym}^2\hspace{0.5mm}\mch^*\cong {\rm Sym}^2\hspace{0.5mm}\mch$; denoting by $U$ its complement in ${\rm Sym}^2\hspace{0.5mm}\mch $ the decomposition of $\mcl$ into $\mfso(d)$ irreps is 
\begin{equation}\label{eq:ortho_decomp}
\mcl = \mbc \hspace{0.7mm} \oplus \hspace{0.7mm}U\oplus\hspace{0.7mm} \Lambda^2\hspace{0.5mm}\mch.
\end{equation}
The next step is of course to calculate the dimensions of $\Lambda^2\hspace{0.5mm}\mch$ and $U$, as well as their weight-zero subspaces. 
Let's begin with $\Lambda^2\hspace{0.5mm}\mch$. We have a basis of weight vectors given by $e_i\otimes e_j$, $1\leq i< j\leq d$, with corresponding weights $\pm L_i\pm L_j$ (where from Eq.~\eqref{eq:ortho_weights} the signs are determined by whether or not the index is greater than $d/2$). With a zero weight therefore resulting from pairing each positive weight $L_i$ of $\mch$ with its negative, we obtain a $d/2$-dimensional weight-zero space. Combined with ${\rm dim \hspace{0.5mm}}\left(\Lambda^2\hspace{0.5mm}\mch\right) = d(d-1)/2$ we conclude 
\begin{equation}\label{eq:apal}
a_{\Lambda^2\hspace{0.15mm}\mch}=\frac{d/2}{d(d-1)/2}=\frac{1}{d-1}.
\end{equation}
The case of ${\rm Sym}^2\hspace{0.5mm}\mch$ is highly similar; we have a basis of weight vectors given by $e_i\otimes e_j$, $1\leq i\leq j\leq d$, with corresponding weights $\pm L_i\pm L_j$, and again a $d/2$-dimensional weight-zero space. One of these is of course the trivial representation $\mbc$; $U$ then possesses a $(d/2 -1)$-dimensional weight-zero space, while itself being of dimension ${\rm dim\hspace{0.5mm}}U=d(d+1)/2 -1$. This yields
\begin{equation}\label{eq:apau}
    a_U = \frac{d/2-1}{d(d+1)/2-1}= \frac{d-2}{(d+2)(d-1)}.
\end{equation}
Combined with the immediate $a_\mbc=1$, Theorem~\ref{thm:mc} now characterizes the measurement channel. \\

Next we turn to the question of bounding the variance of the orthogonal shadow protocol. As usual, by exploiting the specific structure of the group we will be able to outperform the generic bounds of Theorem~\ref{thm:var}. We begin by noting that, for a generic target observable $A$, for the purposes of calculating the variance we can assume without loss of generality that it is traceless~\cite{huang2020predicting}. We then have:
\begin{align}
\mathrm{Var}[\hat{a}]&= \mbe[\hat{a}^2]-\left(\mbe[\hat{a}]\right)^2\\
&\leq \mbe[\hat{a}^2]\\
&=\sum_w\int_{O\in \O(d)}   d\mu_{\O(d)}(O)\  \Tr[\rho O^\dagger \ketbra{w}  O]  \Tr[  \mcm^{-1}\left(A\right)  O^\dagger \ketbra{w} O] ^{ 2}\,.
\end{align}
In principle one can evaluate this integral exactly through a conceptually straightforward but algebraically involved calculation; to simplify things we will work in the limit $d\gg 1$ and retain only leading order terms in $d$, which will be enough to reveal the important aspects of the behaviour. The first result of this simplification is that (recalling Eqs~\eqref{eq:apal} and~\eqref{eq:apau}) the inverse of the measurement channel acts on traceless operators as multiplication by $d$, so that
\begin{align}
\mathrm{Var}[\hat{a}]&\sim d^2\sum_w\int_{O\in \O(d)}   d\mu_{\O(d)}(O)\  \Tr[\rho O^\dagger \ketbra{w}  O]  \Tr[  A  O^\dagger \ketbra{w} O] ^{ 2}\\
&=d^2 \Tr[(\rho\otimes A\otimes A) \sum_w\int_{O\in \O(d)}   d\mu_{\O(d)}(O)\  \left(O^\dagger \ketbra{w}  O\right)^{\otimes 3}]
\end{align}
\begin{figure}
    \centering
    \includegraphics[width=0.75\linewidth]{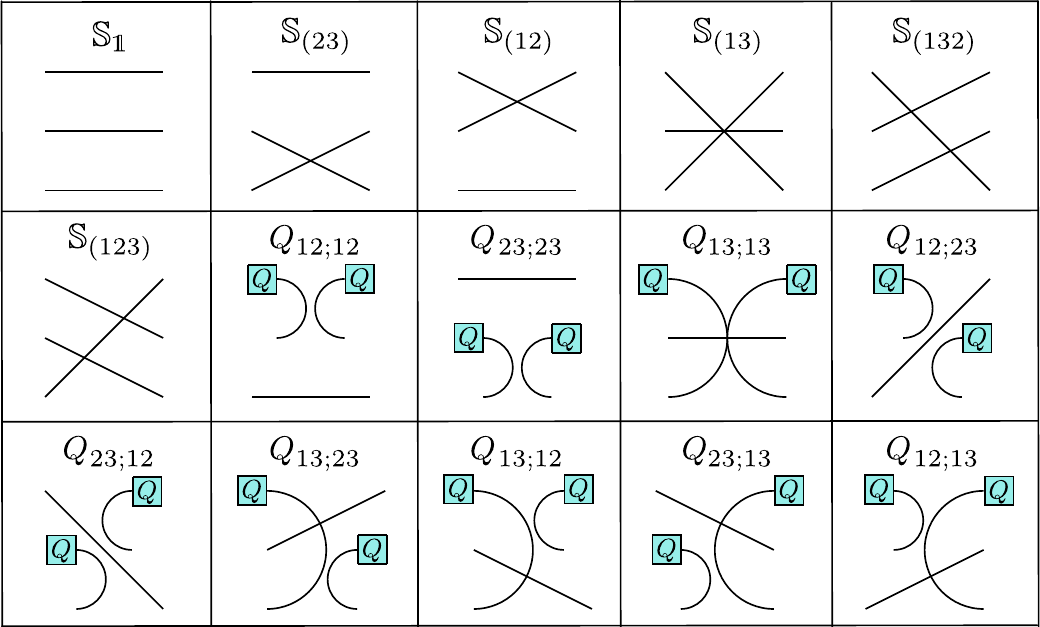}
    \caption{The elements of a spanning set of the third-order commutant of the orthogonal group defined by Eq.~\eqref{eq:ortho_def}. This figure is based on (and in fact slightly generalizes) Fig. 2 of Ref.~\cite{west2025real}, which depicts the set obtained under the choice $Q=\id$. In this appendix we take $Q$ to be defined by Eq.~\eqref{eq:q_def}.  }
    \label{fig:oc3}
\end{figure}

To evaluate the integral, we need a spanning set of the third-order commutant of our copy of $\O(d)$; happily, such a set is known: it is given by the elements of a certain representation of the \textit{Brauer algebra} $\mfb_k(d)$~\cite{collins2006integration,collins2009some,garcia2024architectures}. This bijection has been discussed in some detail in the quantum information literature~\cite{garcia2024architectures,west2024random,west2025real} so we will not belabour it here beyond noting that, as our copy of $\O(d)$ preserves the bilinear form $Q\neq\id$, the explicit representation of $\mfb_k(d)$ is modified slightly compared to that found in Refs.~\cite{garcia2024architectures,west2024random,west2025real}. The resulting elements of the spanning set, 
\begin{equation}\label{eq:oc3}
\{ \mbs_\id, \mbs_{(12)}, \mbs_{(13)}, \mbs_{(23)}, \mbs_{(123)}, \mbs_{(132)}, Q_{12;12}, Q_{23;23}, Q_{13;13}, Q_{12;23}, Q_{23;12}, Q_{13;23}, Q_{13;12}, Q_{23;13}, Q_{12;13}\} \,, 
\end{equation}
are depicted graphically in Fig.~\ref{fig:oc3}.  From the general theory of Weingarten calculus (and mirroring particularly closely the corresponding calculation in Ref.~\cite{west2025real}), we now have
  \begin{equation}
  \int_{O\in \O(d)} d\mu_{\O(d)}(O)\  \left(O^\dagger \ketbra{w}  O\right)^{\otimes 3} = \sum_{\pi\in S_3}a_\pi \mbs_{\pi}+\sum_x b_x Q_x 
    \label{eq:c3span1}
  \end{equation}
  for some coefficients $a_\pi$ and $b_x$ to be determined (suppressing the $d$-dependence of the coefficients). By the permutation symmetry of Eq.~\eqref{eq:c3span1} we immediately have
 \begin{align*}
   &a_{(12)} = a_{(13)} = a_{(23)};\qquad a_{(123)} = a_{(132)};\qquad \\ &
   b_{12;12} = b_{13;13} = b_{23;23};\hspace{4mm} b_{12;13} = b_{12;23} = b_{13;12} = b_{13;23} =b_{23;12} =b_{23;13} 
  \end{align*}
 and so 
  \begin{align}
    \int_{O\in \O(d)} d\mu_{\O(d)}(O)\  \left(O^\dagger \ketbra{w}  O\right)^{\otimes 3}  =&\ a_{\id}\id + a_{(12)}\left(\mbs_{(12)} + \mbs_{(13)} + \mbs_{(23)}\right) + a_{(123)}\left( \mbs_{(123)}+ \mbs_{(132)}\right)\nonumber\\
    &+b_{12;12}\left( Q_{12;12}+ Q_{23;23}+ Q_{13;13}\right)\nonumber\\
    &+ b_{12;13}\left(Q_{12;23}+Q_{23;12}+ Q_{13;23}+ Q_{13;12}+ Q_{23;13}+ Q_{12;13}\right)\,.
    \label{eq:c3span}
  \end{align}
 Following a fairly standard recipe (see e.g., Ref.~\cite{west2025real}) we evaluate the coefficients by multiplying Eq.~\eqref{eq:c3span} by various operators appearing on the right hand side ($\id,\mbs_{(12)},\mbs_{(123)},Q_{12;12}$ and $Q_{12;13}$), and then taking the trace. This yields the system of equations
\begin{align}
    1 &= a_{\id}d^3 + 3a_{\mbs_{12}}d^2 + 2a_{\mbs_{123}}d +3b_{Q_{12;12}}d^2 +6b_{Q_{12;13}}d\\
    1 &= a_{\id}d^2 + a_{\mbs_{12}}\left( d^3+2d\right) + 2a_{\mbs_{123}}d^2 +b_{Q_{12;12}}\left(d^2+2d\right) +b_{Q_{12;13}}\left(2d^2+4d\right)\\
    1 &= a_{\id}d + 3a_{\mbs_{12}} d^2 + a_{\mbs_{123}}\left(d^3+d \right) +3b_{Q_{12;12}}d +3b_{Q_{12;13}}\left(d^2+d\right)\\
    0 &= a_{\id}d^2 + a_{\mbs_{12}}\left( d^2+2d\right) + 2a_{\mbs_{123}}d +b_{Q_{12;12}}\left(d^3+2d\right) +b_{Q_{12;13}}\left(4d^2+2d\right)\\
    0 &= a_{\id}d + a_{\mbs_{12}}\left( d^2+2d\right) + a_{\mbs_{123}}\left( d^2+d\right) +b_{Q_{12;12}}\left(2d^2+d\right) +b_{Q_{12;13}}\left(d^3+2d^2+3d\right)\,.
\end{align}
This system may be easily solved by a computer algebra system; to leading order in $d$ we have
\begin{equation}
a_{\id} = a_{\mbs_{12}} = a_{\mbs_{123}} = d^{-3};\ \ \ \  b_{Q_{12;12}}= b_{Q_{12;13}} = -2d^{-4},
\end{equation}
whence
\begin{align}
\mathrm{Var}[\hat{a}]&\lesssim d^2 \Tr[(\rho\otimes A\otimes A) \sum_w \left(d^{-3}\sum_{\pi\in S_3} \mbs_{\pi}-2d^{-4}\sum_x  Q_x \right)]\\
&= \Tr[(\rho\otimes A\otimes A)  \left(\sum_{\pi\in S_3} \mbs_{\pi}-2d^{-1}\sum_x  Q_x \right)]\label{eq:g27}\,.
\end{align}
The terms in Eq.~\eqref{eq:g27} can be evaluated one by one (perhaps most easily in the graphical notation of Fig.~\ref{fig:oc3}); altogether we have (recalling that $A$ is traceless):
\begin{equation}
\mathrm{Var}[\hat{a}]\sim \Tr[A^2]+2\Tr[\rho A^2]-2d^{-1}\Big(\Tr[A^{\mathsf{T}}QAQ]+2\Tr[Q(A^{\mathsf{T}})^2Q\rho]+\Tr[A^{\mathsf{T}}QAQ\rho]+\Tr[QA^{\mathsf{T}}QA\rho]+2\Tr[AQA^{\mathsf{T}}Q\rho]\Big)\label{eq:ovarsim}
\end{equation}
where Eq.~\eqref{eq:ovarsim} is valid for $d\gg 1$. Typically this will reduce to $\mathrm{Var}[\hat{a}]\sim \Tr[A^2]=\|A\|_2^2$, as by H\"older's inequality we have, for any operator $O$ and density matrix $\rho$, $\Tr[O\rho]\leq\|O\|_\infty\|\rho\|_1=\|O\|_\infty$, which is generally much smaller than $\|O\|_2$~\cite{west2025real}. \\

Without repeating the entire analysis (which would look very similar),
let us briefly remark that one can also work out the details of the orthogonal $G$-shadows protocol in the case where the dimension of the 
Hilbert space $\mch$ of the states is odd (the fact that one must distinguish between odd and even dimensions at all is a reflection of the fact that the representation theory of the orthogonal group behaves differently depending on this parity~\cite{fulton1991representation}). Choosing instead (c.f. Eq.~\eqref{eq:q_def}) the bilinear form 
\begin{equation} 
	Q=\begin{pmatrix} 0& \id_{(d-1)/2} &0 \\  \id_{(d-1)/2} &0& 0\\0&0&1\end{pmatrix},
\end{equation}
one finds that the standard basis is a weight basis. The usual calculations yield
\begin{equation}
    \mcl^\DC=\mcl = \mbc \hspace{0.7mm} \oplus \hspace{0.7mm}U'\oplus\hspace{0.7mm} \Lambda^2\hspace{0.5mm}\mch\,,
\end{equation}
with
\begin{equation}
    a_{U'} = \frac{(d-1)/2}{d(d+1)/2-1}= \frac{1}{d+2}\,,
\end{equation}
and
\begin{equation}
a_{\Lambda^2\hspace{0.15mm}\mch}=\frac{(d-1)/2}{d(d-1)/2}=\frac{1}{d}.
\end{equation}

\subsection{Diagonal SU(2) shadows}\label{sec:su2}
We now come to perhaps one of our main examples, the shadows protocol on $(\mbc^2)^{\otimes n}$ induced by the tensor representation of $\SU(2)$. This will allow us to efficiently learn \textit{any} permutation-invariant operator on a collection of $n$ qubits, an important case with widespread applications throughout both quantum information theory and physics more generally~\cite{schatzki2022theoretical,sauvage2024classical,toth2010permutationally,anschuetz2022efficient}.
In addition to the generic variance bounds of Theorem~\ref{thm:var}, we will find that
\begin{equation}\label{eq:su2nvbound}
\Var_{\SU(2)}[\hat{o}] \leq \frac 43 (n+1)^4 \left\|  O  \right\|_\infty^2.
\end{equation}
Although this bound (Theorem~\ref{thm:su2var}) is often loose in practice, it guarantees -- for classes of operators whose spectral norm does not increase superpolynomially with system size -- a polynomial scaling of the estimator variance, even for operators spread across many isotypics, which cannot be guaranteed by Theorem~\ref{thm:var} alone . Indeed, although Theorem~\ref{thm:var} allows one to make an analogous statement with the spectral norm replaced by the 2-norm, there are natural classes of operators (e.g., Pauli strings) whose 2-norm grows exponentially with $n$. 
This reproduces the polynomial sample-complexity of Ref.~\cite{sauvage2024classical}, while (due to the simplicity of inverting Eq.~\eqref{eq:gb}) avoiding the considerable associated classical post-processing costs.\\ 

\begin{figure*}
\includegraphics[width=\textwidth]{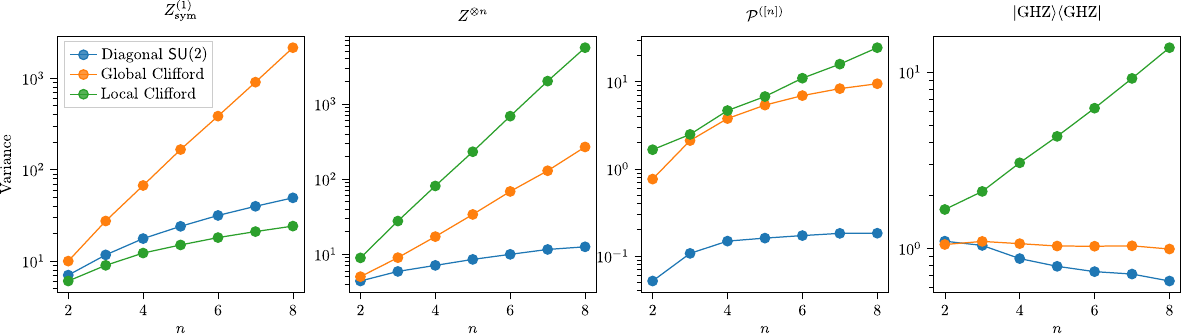}
\caption{Here we numerically investigate the performance of $\SU (2)$-shadows for permutation symmetric $n$-qubit operators, taking $G=\SU(2)$ and acting with the diagonal (tensor) representation. We compare the resulting variance with that of the  local/global Clifford schemes of Ref.~\cite{huang2020predicting}, taking random input states and target observables $O\in\{Z^{(1)}_{\mathrm{sym}},Z^{\otimes n}, \mc{P}^{([n])},\ketbra{\mathrm{GHZ}}\}$,
where $Z_{\mathrm{sym}}^{(1)} = \sum_{i=1}^n Z_i$ and $\mc{P}^{(\lm)}$ is the projector onto the $\lm$ isotypic. In all cases we find that $\SU (2)$-shadows displays at most polynomial growth in the variance, allowing for the efficient estimation of those observables. This is backed up analytically by the guarantees of Theorems~\ref{thm:var} and~\ref{thm:su2var}. }
    \label{fig:su2_channel}
\end{figure*}

\noindent
We will consider the natural  $\SU(2)$-representation $Q$ on $(\mathbb{C}^2)^{\otimes n}$, defined by  
\begin{equation*}
Q(U)\cdot \bigotimes_{i=1}^n \ket{\psi_i}= \bigotimes_{i=1}^n U\ket{\psi_i}.
\end{equation*}
Importantly, this commutes with the natural action $P$ of the permutation group $S_n$ on $(\mathbb{C}^2)^{\otimes n}$, itself given by
\[P(\pi)\cdot \bigotimes_{i=1}^n \ket{\psi_i}= \bigotimes_{i=1}^n \ket{\psi_{\pi^{-1}(i)}},\]
and induces a decomposition
\begin{equation}
(\mathbb{C}^2)^{\otimes n} \cong \bigoplus_{\substack{\lm\vdash n\\\ell(\lm)\leq 2}}Q_\lm \otimes P_\lm 
\end{equation}
block diagonalizing the action of $\SU(2)\times S_n$, the well-known \textit{Schur-Weyl duality}~\cite{bacon2006efficient}. 
Here the direct sum is over the partitions of $n$ of length at most two (i.e., $\lm=[\lm_1,\lm_2]$ with $\lm_1,\lm_2\in\mbz_{\geq 0},\ \lm_1\geq \lm_2,\ \lm_1+\lm_2=n$), and $Q_\lm$ is isomorphic to the spin-$(\lm_1-\lm_2)/2$ ($=:s_\lm$) irrep of $\SU(2)$. We will be concerned solely with the action of $\SU(2)$, expressed in the same basis as
\begin{equation}
(\mathbb{C}^2)^{\otimes n} \cong \bigoplus_{\substack{\lm\vdash n\\\ell(\lm)\leq 2}}Q_\lm \otimes \mbc^{m_\lm} \,,
\end{equation}
with $m_\lm = {\rm dim\ }P_\lm = \frac{n!}{(\lambda_1+1)! \lambda_2!}(\lambda_1 - \lambda_2+1)$ for a given partition $\lm=[\lm_1,\lm_2]$. We denote the elements of this basis (the \textit{Schur basis}) as $\ket{\lm,q_\lm,t_\lm}$. The $t_\lm$ are indexed by the \textit{standard Young tableaux} of $\lm$, and the $q_\lm$ by the \textit{semistandard Young tableaux} of $\lm$~\cite{fulton1991representation}. 
For example, suppose $n=3$. There are two partitions of three of length at most two:  $[3]$ and $[2, 1]$. The corresponding standard Young tableaux are $[[1,2,3]]$ and $\{[[1,2],[3]],\ [[1,3],[2]]\}$ respectively. We therefore get a decomposition
\begin{equation*}
(\mathbb{C}^2)^{\otimes 3} \cong \left(Q_{[3]} \otimes \mathds{1}_{1} \right)\oplus \left(Q_{[2,1]} \otimes \mathds{1}_{2} \right)\cong \mch_{3/2}\oplus\mch_{1/2}\oplus\mch_{1/2}
\end{equation*}
consistent with the usual rules for the addition of angular momenta, $(\mathbf{1/2})^{\otimes 3}=\mathbf{1/2}\otimes(\mathbf{0}\oplus\mathbf{1})=\mathbf{3/2}\oplus\mathbf{1/2}\oplus\mathbf{1/2}$. Importantly for practical applications, the change of basis matrix that relates  the computational basis to the Schur basis can be efficiently implemented on a quantum computer by means of the quantum Schur transform~\cite{bacon2006efficient,bacon2007quantum,kirby2017practical,krovi2019efficient,burchardt2025high}. As every $\SU(2)$-irrep is isomorphic to that of a (half-)integer spin particle, with therefore a non-degenerate spectrum with respect to the generator $J_z$ of $\SU(2)$, we see that by Theorem~\ref{thm:mc}  the Schur basis is centralizing.  \\

As $\mch$ is reducible under the action of $G$, the visible space is a strict subspace of $\mcl$, and, as discussed above, one can only produce unbiased estimates of the expectation values of operators inside of it.  As we show at the end of this appendix, for example, the visible space contains all of the symmetric operators on $n$ qubits, furnishing us with a rich and important set of examples.
The space of $G$-diagonal operators is readily characterised:

\begin{align}
\mcl^\DC &= \bigoplus_{\substack{\lm\vdash n\\\ell(\lm)\leq 2}} \bigoplus_{i=1}^{m_{\lm} } {\rm End}(Q^{\lm})\nonumber\\
&\cong \bigoplus_{\substack{\lm\vdash n\\\ell(\lm)\leq 2}} \bigoplus_{i=1}^{m_{\lm} }  \bigoplus_{s=0}^{2s_\lm} V_s \label{eq:ldirreps}
\,,
\end{align}
where $V_s$ is the spin-$s$ irrep of $\SU(2)$.
So the irreps of $\mcl^\DC$ (Eq.~\eqref{eq:ldirreps}) are   isomorphic to that of integer spin particles. Each of them has a one-dimensional weight-zero subspace, with the $\mcm^{-1}$-eigenvalue of an operator with support only on one such irrep then given by its dimension (i.e., $1/a_\lm^H=d_\lm$). It follows that all of $\mcl^D=\mcl^\mcv$ is visible.\\

\noindent
We now consider in some detail two examples of operators that can be efficiently estimated by $\SU(2)$-tensor shadows.
\subsubsection{A symmetrised Z operator}
\noindent
As a natural first example, we consider the symmetrised $Z$ operator, i.e., for any $1\leq j \leq n$ the
twirl 
\begin{align}
Z_{\rm sym}^{(1)} &:= n\mct_{S_n}[Z_j]=\sum_{i=1}^nZ_i\label{eq:ztwirl}\,,
\end{align}
where $Z_j$ is the Pauli $Z$ operator acting on the $j$th qubit.  This operator occurs frequently, for example, in the study of permutation symmetric quantum systems~\cite{anschuetz2022efficient,schatzki2022theoretical,sauvage2024classical}; 
as it is permutation-symmetric, it lies within the $\SU(2)$-tensor rep visible space (see Lemma~\ref{lem:pivs}). 
In the Schur basis $Z_{\rm sym}^{(1)}$ takes the form~\cite{sauvage2024classical}
\begin{align}
Z_{\rm sym}^{(1)}&=\sum_{\eta,t_{\eta},q_{\eta}} 2q_{\eta} \ketbra{\eta,q_{\eta},t_{\eta}} = \bigoplus_{\substack{\eta\vdash n\\\ell(\lm)\leq 2}}  \sum_{q_{\eta}=-s_{\eta}}^{s_{\eta}} 2q_{\eta} \ketbra{q_{\eta}}{q_{\eta}} \ot \mathds{1}_{m_{\eta}} \label{eq:zsymd}\,.
\end{align}
In order to employ the bounds of Theorem~\ref{thm:var} we need to know the decomposition of $Z_{\rm sym}^{(1)}$ into irreps of $\mcl^\DC$. 
For a given $\eta\vdash n$ we have a basis for the corresponding $G$-diagonal operators in ${\rm End\,}\mch^{\eta,i}$ given by $\{O_{s_{12}m_{12}}^{\eta, i}\}_{s_{12},m_{12}},\ s_{12}=0,1,\ldots, 2{s_{\eta}},\ m_{12}=-s_{\eta},s_{\eta}+1,\ldots,s_{\eta}$, where
(seeing Appendix~\ref{sec:exact} for some explicit examples)
\begin{equation}\label{eq:obasis}
O_{s_{12}m_{12}}^{\eta, i}=\sum_{m=-s_{\eta}}^{m=s_{\eta}}(-1)^{m+s_{\eta}}C^{s_{\eta},s_{\eta},s_{12}}_{m,m_{12}-m,m_{12}}\ketbra{\eta,m}{\eta,m-m_{12}}\ot \ketbra{i}.
\end{equation}
Comparing this with our expression (Eq.~\eqref{eq:zsymd})
for $Z_{\rm sym}$ we see that there is no contribution for  non-zero weight $O_{s_{12}m_{12}}$. In the weight-zero case we have
\begin{align*}
\Tr [(O_{s_{12},0}^{\eta,i})^\dagger Z_{\rm sym}] &= \sum_{m=-s_{\eta}}^{m=s_{\eta}}2m(-1)^{m+s_{\eta}}C^{s_{\eta},s_{\eta},s_{12}}_{m,-m,0}=-\delta_{s_{12},1}\sqrt{2\binom{d_{\eta} + 1}{  3}}\,,
\end{align*}
where the last step may be verified by a computer algebra system. So we have
\begin{align}
Z_{\rm sym}&=-\sum_{\substack{\eta\vdash n\\\ell(\lm)\leq 2}} \sum_{i=1}^{m_\eta} \sqrt{2\binom{d_{\eta} + 1}{3}}O_{1,0}^{\eta,i},\label{eq:zsymd2}
\end{align}
an expression for $Z_{\rm sym}$ as a sum of operators each individually contained in a single irrep isomorphic to the spin-1 irrep of $\SU(2)$. From the second bound of Theorem~\ref{thm:var} we now have (recalling $a_{\mathbf{s}=1}=1/3$)
\begin{equation}
\Var_{\SU(2)}[\hat{z}_{\rm sym}^{(1)}] \leq 9\|Z_{\rm sym}^{(1)}\|_\infty^2 =9n^2.
\end{equation}
Investigating numerically, we find in Fig.~\ref{fig:su2_channel} that this bound is reasonably tight, with the numerical scaling of the variance for random initial states super-linear but sub-quadratic.
In fact, the simplicity of Eq.~\eqref{eq:zsymd2} renders the exact computation of the variance for any input state tractable (although admittedly still somewhat involved). This calculation is carried out in  Appendix~\ref{sec:exact}.\\

\subsubsection{Projecting onto an isotypic}

As a second example we  consider the projector $\Pi^{(\eta)}$ onto the $\eta$ isotypic (corresponding as above to some partition of $n$ of length at most two). 
In this case we have a particularly straightforward decomposition into weight vectors,
\begin{align}
\Pi^{(\eta)} &=\bigoplus_{i=1}^{m_\eta}\id_{Q_\eta}= \sum_{i=1}^{m_\eta}\sqrt{d_{\eta}} O^{(\eta,i)}_{0,0}\label{eq:projdecomp}\,.
\end{align}
In fact this is a decomposition of $\Pi^{(\eta)}$ into weight-zero vectors living in trivial irreps;  the variance bound is therefore (independently of both $d_{\eta}$ and $m_{\eta}$)
\begin{equation}\label{eq:pvbound}
\Var_{\SU(2)}[\hat{\pi}^{(\eta)}]\leq \frac{1}{1} \| \Pi^{(\eta)}\|_\infty^2 =1.
\end{equation}
Note that $\Pi^{(\eta)}$ is positive semi-definite, allowing us to employ the improved bounds of Theorem~\ref{thm:varf} (although in this case there is no benefit over the Theorem~\ref{thm:var}, as $1/a_0^2=1/a_0=1$). In this case too we can go beyond the bounds and evaluate the variance exactly for an arbitrary initial state $\rho$ by means of a Clebsch–Gordan decomposition; the result, derived in Appendix~\ref{sec:exact}, is
\begin{equation*}
\mathrm{Var}_{\SU(2)}[\hat{\pi}^{(\eta)}] = \sqrt{d_{\eta}}\sum_{i}\rho^{\eta,i}_{0,0} -{d_{\eta}}\left[\sum_{i}\rho^{\eta,i}_{0,0}\right]^2 \leq \frac{1}{4}\,,
\end{equation*}
where we have used that for any $x\in\mbr$ we have $x-x^2\leq 1/4$. Happily, this slightly improves upon the general bound Eq.~\eqref{eq:pvbound} given by Theorem~\ref{thm:varf}. This projector is highly non-local and therefore unsuitable for estimation by local Cliffords, and has $\|\Pi^{(\eta)}\|_2^2={\rm rank\ }\Pi^{(\eta)}=d_{\eta}m_{\eta}$, which is generically $\sim\mco(\exp n)$, making measuring it via global Cliffords likewise intractable. This is explored numerically in Fig.~\ref{fig:su2_channel}, where the advantage afforded by $\SU (2)$-shadows is clear. \\

Next we prove that diagonal $\SU (2)$-shadows allows for the efficient estimation of any permutation-invariant operator on $n$ qubits. This is a two-step process: first we must show that the permutation-invariant operators live within the corresponding visible space, so that our estimates are unbiased, and second we must bound the variance of the resulting estimators. We begin with step one:

\begin{lemma}\label{lem:pivs}
The visible space of diagonal $\SU (2)$-shadows contains the permutation-invariant operators.
\end{lemma}
{
\begin{proof}
Consider as before the Schur-Weyl decomposition
\begin{equation}\label{eq:su2hd}
\HC=(\mbb{C}^2)\tn   \cong \bigoplus_{\substack{\lm\vdash n\\\ell(\lm)\leq 2}}Q_\lm \otimes P_\lm \,,
\end{equation}
where the $Q_\lm$ and $P_\lm$ respectively furnish the irreps of $\SU(2)$ and $S_n$ corresponding to the partition $\lm$. Letting $\{\ket{T_\lm}\}_T$ denote a basis for each $S_n$-irrep $P_\lm$, the space of $\SU (2)$-diagonal operators is given as usual by 
\begin{equation}\label{eq:su2diag}
\mc{L}^D \cong \bigoplus_{\lm,T_\lm} {\rm End}(Q_{\lm})\ot \ketbra{T_\lm}\,.
\end{equation}
In this case, the entire space $\mcl^D$ is visible; indeed, we have seen that ${\rm End}(Q_{\lm})\cong \bigoplus_{s=0}^{2s_\lm} V_s$ with each $V_s$ having a non-zero weight-zero subspace, thereby escaping annihilation by the shadows channel. Meanwhile, by Schur's lemma  the permutation-invariant operators are given exactly by the elements of the space 
\begin{equation}
\mcl^{S_n} \cong \bigoplus_\lm A_\lm \ot \id_{P_\lm}\,, 
\end{equation}
for arbitrary operators $A_\lm\in {\rm End}(Q_{\lm})$. But $\id_{P_\lm}$ is of course nothing more than $\sum_{T_\lm}\ketbra{T_\lm}$, so that $\mcl^{S_n}\subset \mcl^D=\mcl^\mcv$. 
\end{proof}

For the sake of clarity, we remark that  $\SU(2)$-invariant operators are \textit{not} necessarily visible to $\SU(2)$-shadows; indeed note that with respect to the decomposition of Eq.~\eqref{eq:su2hd} the space of $\SU(2)$-invariant operators is given by 
$\mcl^{\SU(2)} \cong \bigoplus_\lm   \id_{Q_\lm}\ot B_\lm$,  
for arbitrary operators $B_\lm\in {\rm End}(P_{\lm})$, which need not sit inside the visible space Eq.~\eqref{eq:su2diag}. \\

Next we bound the variance of the diagonal $\SU(2)$-estimators of permutation-invariant operators:}
\sutwovar*

\begin{proof}
{That  the estimators are unbiased is the content of Lemma~\ref{lem:pivs}.} Now, let $O = \bigoplus_\lm A_\lm \ot \id_{m_\lm}$ (with $A_\lm\in{\rm End}(Q_\lm)$) be a permutation-invariant operator. Let us further decompose each $A_\lm$ into its components with respect to the decomposition ${\rm End}(Q_\lm)\cong \bigoplus_{r=0}^{2s_\lm} Q_r$ as $A_\lm = \sum_{r=0}^{2s_\lm} A_{\lm, r}$, so that the measurement channel acts as $\mcm^{-1}(O) = \bigoplus_\lm \left(\sum_{r=0}^{2s_\lm} (2r+1) A_{\lm, r}\right) \ot \id_{m_\lm}$. Now, by Eq.~\eqref{eq:idres} we have
\begin{equation}
    \Var_{\SU(2)}[\hat{o}] \le \|\mcm^{-1}(O) \|_\infty^2 = \max_\lm \left\|\sum_{r=0}^{2s_\lm} (2r+1) A_{\lm, r}\right\|_\infty^2\,;
\end{equation}
let us bound this later quantity. For any $\lm$ we have
\begin{align}
   \left\|\sum_{r=0}^{2s_\lm} (2r+1) A_{\lm, r}\right\|_\infty &\le \sum_{r=0}^{2s_\lm} (2r+1) \left\|A_{\lm, r}\right\|_\infty \\
   &\le \sum_{r=0}^{2s_\lm} (2r+1) \left\|A_{\lm, r}\right\|_2 \\
   &\le \left(\sum_{r=0}^{2s_\lm} (2r+1)^2\right)^{1/2} \left(\sum_{r=0}^{2s_\lm} \left\|A_{\lm, r}\right\|_2^2\right)^{1/2}  \\
   &= \left(\sum_{r=0}^{2s_\lm} (2r+1)^2\right)^{1/2}  \left\|A_{\lm}\right\|_2  \\
   &=\sqrt{\frac{d_\lm (4d_\lm^2-1)}{3} } \left\|A_{\lm}\right\|_2  \\
   &\le\sqrt{\frac{d_\lm (4d_\lm^2-1)}{3} } \sqrt{d_\lm} \left\|A_{\lm}\right\|_\infty \\
   &\le\sqrt{\frac{d_\lm (4d_\lm^2-1)}{3} } \sqrt{d_\lm} \left\|O\right\|_\infty \,,
\end{align}
where we have used the Cauchy–Schwarz inequality, the pairwise orthogonality (with respect to the Hilbert-Schmidt) of the $A_{\lm, r}$, the well-known chain of inequalities $\|A\|_\infty \le\|A\|_2 \le\sqrt d\|A\|_\infty $, and (in the final step) the direct-summand relationship between the $A_\lm$ and $O$. It follows that
\begin{equation}
    \Var_{\SU(2)}[\hat{o}] \le \frac{4d_\lm^4}{3} \left\|O\right\|_\infty ^2 \le \frac{4}{3} (n+1)^4 \left\|O\right\|_\infty ^2 \,.
\end{equation}

\end{proof}
Summarizing, Lemma~\ref{lem:pivs} shows that permutation-invariant operators lie in the visible space, hence the estimators are unbiased. Theorem~\ref{thm:su2var} then follows by combining this visibility statement with the generic bound Eq.~(C28) and the decomposition of \(\MC^{-1}\) in the Schur basis.

\subsection{Symmetric group shadows}\label{sec:snshad}
\noindent
In this appendix we fill in the details from the $S_n$-shadows examples from the main text. We begin with

\begin{lemma}\label{lem:Sn-perm-NDCSE}
Let $G = S_n$ act on $\mch = \mathbb{C}^n$ by the permutation representation $R(\sigma)\ket{i} = \ket{\sigma(i)}$, for $\sigma \in S_n,\ i \in [n]$.
Then there exists an abelian subgroup $H \subseteq S_n$ and a basis $\mcw$ of $\mch$
such that $\mcw$ is an NDCSE in the sense of Definition~\ref{def:cse}.
\end{lemma}

\begin{proof}
It is well-known that
$\mch \cong \mathbf{1} \oplus V$,
where $\mathbf{1}$ is the trivial irrep, spanned by
$\ket{u} = \frac{1}{\sqrt{n}} \sum_{i=1}^n \ket{i}$, and $V= ({\rm span}\,\ket{u})^{\perp}$ is the $(n-1)$-dimensional \textit{standard representation}~\cite{fulton1991representation}. Now, let $c = (1\ 2\ \dots\ n)$ and define $H \coloneqq \langle c \rangle \cong C_n$. Further let $\omega = e^{2\pi i/n}$ and define
\begin{equation}
\ket{w_k} = \frac{1}{\sqrt{n}} \sum_{j=0}^{n-1} \omega^{jk} \ket{j+1},
\qquad k = 0,\dots,n-1,
\end{equation}
so that $c \ket{w_k} = \omega^k \ket{w_k}$, and $\{\ket{w_k}\}_{k=0}^{n-1}$ is an orthonormal eigenbasis for the abelian group $H$.
We note that
$\ket{w_0} = \ket{u}$
spans the trivial subrepresentation $\mathbf{1}$, and
$V = \mathrm{span}\{\ket{w_k} : 1 \le k \le n-1\}$, so that $
\mcw \coloneqq  \{\ket{w_0},\dots,\ket{w_{n-1}}\}$
is a Fourier basis in the sense of Definition~\ref{def:fb}.

Now, as an $H$-representation we have
\begin{equation}
V \downarrow_H \cong \bigoplus_{k=1}^{n-1} \chi_k,
\qquad \chi_k(c) = \omega^k.
\end{equation}
The characters $\chi_k$ are all distinct, so $V \downarrow_H$ is multiplicity-free.
By Definition~\ref{def:cse}, $\mcw$ is a commuting subgroup eigenbasis (CSE) for $(S_n,H)$,
and the multiplicity-freeness implies that it is   an NDCSE.
\end{proof}
\noindent Next we turn to the decomposition of $\mcl^\DC$, finding:
\begin{lemma}
Let $\mch = \mathbb{C}^n$ be the permutation representation of $S_n$ as above,
and let $\mcl = \mathrm{End}(\mch)$ carry the adjoint $S_n$-action
$X \mapsto R(\sigma) X R(\sigma)^{-1}$.
Then the diagonal submodule
\begin{equation}
\mcl^\DC \coloneqq \bigoplus_{\eta} \mathrm{End}(\mch_\eta)
= \mathrm{End}(\mathbf{1}) \oplus \mathrm{End}(V)
\end{equation}
decomposes as
\begin{equation}
\mcl^\DC \cong 2\,V_{(n)} \oplus V_{(n-1,1)} \oplus V_{(n-2,2)} \oplus V_{(n-2,1,1)}
\end{equation}
for $n \ge 4$.
\end{lemma}

\begin{proof}
Since $\mch \cong \mathbf{1} \oplus V$,
we have
$\mcl^\DC = \mathrm{End}(\mathbf{1}) \oplus \mathrm{End}(V)
\cong \mathbf{1} \oplus (V \otimes V^\ast).$
The standard representation $V$  is self-dual (as indeed are all $S_n$ irreps), so $V^\ast \cong V$ and
$V \otimes V \cong \mathrm{Sym}^2 V \oplus \Lambda^2 V$.
For $n \ge 4$ one furthermore has the well-known decompositions
\begin{align}
\mathrm{Sym}^2 V &\cong V_{(n)} \oplus V_{(n-1,1)} \oplus V_{(n-2,2)},\\
\Lambda^2 V &\cong  V_{(n-2,1,1)},
\end{align}
whence
\begin{equation}
\mathrm{End}(V) \cong V_{(n)} \oplus V_{(n-1,1)} \oplus V_{(n-2,2)} \oplus V_{(n-2,1,1)},
\end{equation}
and adding the trivial copy $\mathrm{End}(\mathbf{1}) \cong V_{(n)}$ gives the claim.
\end{proof}
\noindent
We next turn to the calculation of the corresponding $a_\lm$:

\begin{lemma} 
Let $\mcw$ be the NDCSE from Lemma~\ref{lem:Sn-perm-NDCSE}, arising from
$H = \langle c \rangle$.
For each irrep $\lambda$ that appears in $\mcl^\DC$, let $d_\lambda = \dim V^\lambda$
and $d_\lambda^H = \dim\bigl( (V^\lambda)^H \big)$ be the dimension of the
$H$-invariant subspace.
Then:
\begin{align}
a_{(n)}^H &= 1\\
a_{(n-1,1)}^H &= 0\\
a_{(n-2,2)}^H &=
\begin{cases}
\displaystyle \frac{1}{n}, & \text{$n$ odd},\\[4pt]
\displaystyle \frac{n-2}{n(n-3)}, & \text{$n$ even},
\end{cases}\\
a_{(n-2,1,1)}^H &=
\begin{cases}
\displaystyle \frac{1}{n-2}, & \text{$n$ odd},\\[4pt]
\displaystyle \frac{1}{n-1}, & \text{$n$ even}.
\end{cases}
\end{align}
\end{lemma}

\begin{proof}
The statements for $(n)$ and $(n-1,1)$ are immediate:
for any subgroup, the trivial irrep has a one-dimensional fixed space, and
$V_{(n-1,1)} \downarrow_H \cong \bigoplus_{k=1}^{n-1} \chi_k$ with all $\chi_k$ nontrivial.
Now, as a $H$-representation
\begin{equation}
V \downarrow_H \cong \bigoplus_{k=1}^{n-1} \chi_k,
\qquad
\chi_k(c) = \omega^k,
\end{equation}
whence
\begin{equation}
\mathrm{Sym}^2 V \downarrow_H \cong \bigoplus_{k \le \ell} \chi_{k+\ell},
\qquad
\Lambda^2 V \downarrow_H \cong \bigoplus_{k < \ell} \chi_{k+\ell},
\end{equation}
where the indices are understood mod $n$;
a $H$-invariant vector in either space corresponds to a pair $(k,\ell)$ with
$1 \le k,\ell \le n-1$ and $k+\ell \equiv 0 \pmod{n}$.

If $n$ is odd, then $2k \equiv 0 \pmod{n}$ has no nonzero solution, so all
valid solutions have $k \ne \ell$, and the unordered pairs are $\{1,n-1\},\{2,n-2\},\dots,\left\{\frac{n-1}{2},\frac{n+1}{2}\right\},$
giving $\frac{n-1}{2}$ such pairs. Hence
\begin{equation}
\dim\bigl((\mathrm{Sym}^2 V)^H\big)
=
\dim\bigl((\Lambda^2 V)^H\big)
=
\frac{n-1}{2}.
\end{equation}
Using the decompositions
\begin{equation}
\mathrm{Sym}^2 V \cong V_{(n)} \oplus V_{(n-1,1)} \oplus V_{(n-2,2)},
\qquad
\Lambda^2 V \cong   V_{(n-2,1,1)},
\end{equation}
with $d_{(n)}^H = 1$ and $d_{(n-1,1)}^H = 0$, we obtain
\begin{equation}
\frac{n-1}{2}
=
1 + d_{(n-2,2)}^H,
\qquad
\frac{n-1}{2}
=
d_{(n-2,1,1)}^H.
\end{equation}
Alternately, if $n$ is even, say $n=2m$, then $k+\ell \equiv 0 \pmod{n}$ has one solution
with $k=\ell=m$ and $m-1$ solutions with $k \ne \ell$, so there are $m$ unordered
pairs in the symmetric square and $m-1$ in the exterior square. Thus
\begin{equation}
\dim\bigl((\mathrm{Sym}^2 V)^H\big) = m = \frac{n}{2},
\qquad
\dim\bigl((\Lambda^2 V)^H\big) = m-1 = \frac{n}{2} - 1.
\end{equation}
Using the same $S_n$-decompositions and the known values of $d_{(n)}^H$ and
$d_{(n-1,1)}^H$ gives
\begin{equation}
\frac{n}{2} = 1 + d_{(n-2,2)}^H,
\qquad
\frac{n}{2} - 1 = d_{(n-2,1,1)}^H.
\end{equation}
which yields the claimed values for $n$ even. Elementary applications of the \textit{hook length formula}~\cite{fulton1991representation} yield
\begin{equation}
    d_{(n)} = 1,\qquad d_{(n-1,1)} = n-1,\qquad d_{(n-2,2)} = \frac{n(n-3)}{2},\qquad d_{(n-2,1,1)} = \frac{(n-1)(n-2)}{2},
\end{equation}
whence the claimed formulas for the $a_\lambda^H$ follow by dividing the $d_\lambda^H$ by the $d_\lambda$.
\end{proof} 

\noindent
We now come to the proof of
\snndcse*

\begin{proof}
By Definition~\ref{def:cse}, an NDCSE is determined by an abelian subgroup
$H \subseteq S_n$ such that, for each $S_n$-irrep $V^\eta$, the restriction
$V^\eta \downarrow_H$ is multiplicity-free. In particular, if $V^\eta$ admits
an NDCSE via some abelian $H$, then
\begin{equation}
V^\eta \downarrow_H \cong \bigoplus_{\chi \in \widehat{H}} m_\chi \chi,
\qquad m_\chi \in \{0,1\},
\end{equation}
so
\begin{equation}\label{eq:dim-le-H}
\dim( V^\eta) = \sum_{\chi} m_\chi \le |\widehat{H}| = |H|.
\end{equation}
Thus a necessary condition for $V^\eta$ to admit an NDCSE is that
$\dim V^\eta \le |H|$ for some abelian subgroup $H \subseteq S_n$. The key technical result that we need is that, for $n\ge 3$, every abelian subgroup $H \subseteq S_n$ satisfies
$ |H| \le 3^{n/3}$~\cite{dixon1971maximal}. On the other hand, we can find irreps whose dimensions are much larger than this; indeed, assume $n= 2m$ is even and let $V^\lambda$ be the irreducible
representation of $S_n$ corresponding to the two-row partition $\lambda = (m,m) \vdash 2m$. 
By the hook length formula~\cite{fulton1991representation} and Stirling's approximation one can see that
\begin{equation}
\dim (V^{(m,m)}) = \frac{(2m)!}{m!(m+1)!}\sim \frac{2^n}{\mathrm{poly}(n)},
\end{equation}
which evidently dominates $3^{n/3}\approx1.44^n$
for  sufficiently large $n$. We conclude that for such an $n$, $V^{(m,m)}$ does not admit an NDCSE.
\end{proof}

\noindent
Next we come to the proof of
\lemaba*
\begin{proof}
We will prove this lemma by producing an explicit example of an $S_n$ irrep and a  non-degenerate commutative subalgebra eigenbases (in the form of a Gelfand-Tsetlin basis) for which we can calculate the measurement channel exactly, and find that it is not central. \\

Recall that the Gelfand--Tsetlin basis of an irrep $V^\lm$ of $S_n$ is a non-degenerate simultaneous basis of the algebra of \textit{Jucys--Murphy elements},
$J_k \coloneqq \sum_{j<k} (jk) \in\mbc[S_n], \ 2\leq k\leq n$ (where $(jk)$ denotes the permutation which exactly swaps $j$ and $k$). As is well-known, the elements of the Gelfand--Tsetlin basis are indexed by standard Young tableau of shape $\lm$. For example, let us take the partition $\lm=[3,1,1]$ and considered the corresponding irrep $V^\lm$ of $S_5$. A routine calculation with the computer algebra system \texttt{Sage}~\cite{sage} yields 
\begin{equation}\label{eq:s5decomp}
    V^{[3,1,1]} \ot \left(V^{[3,1,1]}\right)^* \cong V^{[5]}\oplus V^{[4,1]}\oplus 2\,V^{[3,2]}\oplus V^{[3,1,1]}
\oplus 2\,V^{[2,2,1]}\oplus V^{[2,1,1,1]}\oplus V^{[1,1,1,1,1]}\,.
\end{equation}
Meanwhile, the Gelfand--Tsetlin basis vectors are given by
\begin{equation}
        T_1  = \begin{ytableau}
        1&2&3\\4\\5
    \end{ytableau}\,,\ T_2  = \begin{ytableau}
        1&2&4\\3\\5
    \end{ytableau}\,,\ T_3  = \begin{ytableau}
        1&2&5 \\3\\4
    \end{ytableau}\,,\ T_4  = \begin{ytableau}
        1&3&4\\2\\5
    \end{ytableau}\,,\ T_5  = \begin{ytableau}
        1&3&5\\2\\4
    \end{ytableau}\,,\ T_6  = \begin{ytableau}
        1&4&5\\2\\3
    \end{ytableau}\,.
\end{equation}
The transpositions $s_i=(i,i+1)$ have a particularly nice form in the Gelfand--Tsetlin basis. Indeed, if $i$ and $i+1$ are in the same row of $T$ then $s_i\ket T=\ket T$, and if they are in the same column then $s_i\ket T=-\ket T$. Otherwise, they map $\ket T$ into the span of $\{\ket T,\ket {s_i T}\}$ in a manner represented by the matrix 
\begin{equation}\label{eq:trans}
    V^\lm(s_i)\rvert_{{\rm span}\, \{\ket T,\ket {s_i T}\} } = \begin{pmatrix}
        1/r & \sqrt{1-1/r^2}\\
        \sqrt{1-1/r^2} & -1/r
    \end{pmatrix} ,
\end{equation}
where $r=c_T(i+1)-c_T(i)$, with $c_T(k)={\rm column}(k)-{\rm row}(k)$. As the adjacent transpositions generate the symmetric group, this allows one to mechanically obtain all of the representing matrices. Now, as the measurement channel is given by
\begin{align}
\mcm(X)&= \frac1{|S_5|}\sum_{g\in S_5}V^{[3,1,1]}(g)^{-1}\,
\mca_{\mathrm{GT}}\!\bigl(V^{[3,1,1]}(g)XV^{[3,1,1]}(g)^{-1}\big)\,V^{[3,1,1]}(g)\\
&= \frac1{120}\sum_{g\in S_5}\sum_{i=1}^{d_{V^{[3,1,1]}}} \sbraket{T_i\vert  V^{[3,1,1]}(g)XV^{[3,1,1]}(g)^{-1}} {T_i}   V^{[3,1,1]}(g)^{-1}\ketbra{T_i}
V^{[3,1,1]}(g)\,,
\end{align}
and we can by means of Eq.~\eqref{eq:trans} obtain the explicit matrix representations of each of the elements of $S_5$ with respect to the basis $\{\ket{T_i}\}_i$, we can perform the above finite sum on a computer and obtain the  exact eigenvalues of the channel. Upon doing this, we find that the channel acts in particular on the isotypics with degeneracies (recall Eq.~\eqref{eq:s5decomp}) as
\begin{equation}
    \mcm\big|_{2V^{(3,2)}} \cong 
\begin{pmatrix}
\frac{11}{90} & 0\\[1mm]
0 & \frac{43}{120}
\end{pmatrix}\otimes \id_{5},
\qquad
\mcm\big|_{2V^{(2,2,1)}} \cong 
\begin{pmatrix}
0 & 0\\[1mm]
0 & \frac{1}{24}
\end{pmatrix}\otimes \id_{5}\,;
\end{equation}
that is, not as a scalar (on, indeed, either of them). This concludes the proof.
\end{proof}

\subsection{Particle preserving free fermions}\label{sec:ppff}

We now consider the subgroup of the matchgate symmetry that preserves particle number. As in Section~\ref{sec:ff}, the underlying Hilbert space is the same $n$-qubit space $\mch \cong (\mbc^2)^{\otimes n}$, which we
identify  with the fermionic Fock space of $n$ modes via the Jordan--Wigner transform.
We begin by recalling the Majorana operators (Eq.~\eqref{eq:majoranas}), $\gamma_{2j-1}=Z_1\cdots Z_{j-1}X_j,$ $\gamma_{2j}=Z_1\cdots Z_{j-1}Y_j$, where $1\le j\le n.$
From these we form the usual creation and annihilation operators
\begin{equation}\label{eq:ppff-aa}
a_j=\frac{1}{2}\big(\gamma_{2j-1}+i\gamma_{2j}\big)
=
Z_1\cdots Z_{j-1}\left(\frac{X_j+iY_j}{2}\right),
\qquad
a_j^\dagger=\frac{1}{2}\big(\gamma_{2j-1}-i\gamma_{2j}\big)
=
Z_1\cdots Z_{j-1}\left(\frac{X_j-iY_j}{2}\right).
\end{equation}
The particle-number--preserving Gaussian unitaries are those generated by quadratic Hamiltonians containing only $a^\dagger a$ terms, $H=\sum_{i,j=1}^n h_{ij}\,a_i^\dagger a_j.  $
Equivalently, they form the image of $\U(n)$ under fermionic second quantization. We denote this represented group by $G_{\mathrm{pp}}$. 
A convenient Hermitian generating set for the Lie algebra is 
\begin{align}
H^{(x)}_{j,j+1}&=\frac{1}{2}\big(X_jX_{j+1}+Y_jY_{j+1}\big),\\
H^{(y)}_{j,j+1}&=\frac{1}{2}\big(X_jY_{j+1}-Y_jX_{j+1}\big),\\
N_j&=\frac{1-Z_j}{2},
\end{align}
so that on a connected chain one may think of $G_{\mathrm{pp}}$ as the subgroup of matchgates generated by nearest-neighbor $XX+YY$, $XY-YX$, and $Z$ terms. 

We now turn to the search for an appropriate basis for shadows. 
Let $\ket{\Omega}$ denote the vacuum state, characterized by $a_j\ket{\Omega}=0$ for all $j$. The \emph{Fock} (occupation-number) basis of $\mch$ has elements
\begin{equation}
\ket{n_1,\dots,n_n}
:=
(a_1^\dagger)^{n_1}\cdots (a_n^\dagger)^{n_n}\ket{\Omega},
\qquad
n_j\in\{0,1\};
\end{equation}
equivalently, writing $S\subseteq[n]$ for the occupied modes, $\ket{S}:=\prod_{j\in S} a_j^\dagger\ket{\Omega}$. 
Now, the particle-number operator $N:=\sum_{j=1}^n a_j^\dagger a_j$
decomposes $\mch$ into $r$-particle sectors, $\mch=\bigoplus_{r=0}^n \mch_r,
$ where $
\mch_r:=\mathrm{span}\{\ket{S}:\ |S|=r\}$. 
As a $\U(n)$-module, $\mch_r$ is naturally isomorphic to the $r$\textsuperscript{th} exterior power $\Lambda^r\mbc^n$ via $
e_{i_1}\wedge\cdots\wedge e_{i_r}
\longmapsto
a_{i_1}^\dagger\cdots a_{i_r}^\dagger\ket{\Omega}
\quad (i_1<\cdots<i_r)$ (this identification is well-defined by the canonical anticommutation relations of the creation operators).
Hence
\begin{equation}\label{eq:ppff-H-dec}
\mch \cong \bigoplus_{r=0}^n \Lambda^r \mbc^n
\end{equation}
as a multiplicity-free decomposition into irreducible $\U(n)$-modules.
Now, let $T=\U(1)^n\subset \U(n)$ be the natural diagonal torus. If $t=\mathrm{diag}(z_1,\dots,z_n)\in T,$ then
$t\cdot \ket{S}=\big(\prod_{j\in S} z_j\big)\ket{S}$, so that the Fock basis $\mcf$ consists of joint $T$-eigenvectors. Since different subsets $S$ give different weights, this basis is non-degenerate, and is therefore an NDCSE.

Let us turn next to the decomposition of ${\rm End}\,(\Lambda^r V), $ where to simplify the notation a little we write $V=\mbc^n$.
First, we have a natural isomorphism
\begin{equation}
{\rm End}\,(\Lambda^r V)\cong \Lambda^r V\otimes \big(\Lambda^r V\big)^*  \cong
\Lambda^r V \otimes \Lambda^{n-r}V \otimes \det(V)^{-1}.
\end{equation}
This is helpful as the tensor product $\Lambda^r V \otimes \Lambda^{n-r}V$ is known to be multiplicity-free. Indeed, by Pieri's rule, its irreducible summands are indexed by the Young diagrams $(2^j,1^{n-2j}),\, j=0,1,\dots,m_r,\ m_r:=\min(r,n-r)$, each appearing with multiplicity one.
After twisting by $\det(V)^{-1}$, one obtains irreducibles $X_j$ 
of highest weight $(1^j,0^{\,n-2j},-1^j).$
 Thus
\begin{equation}\label{eq:ppff-end-r-dec}
{\rm End}\,(\Lambda^r V)\cong \bigoplus_{j=0}^{m_r} X_j,
\end{equation}
again multiplicity-free. 

To determine the measurement channel, we need the dimension of $X_j$ and the dimension of its zero-weight subspace.
First, the hook-content formula gives
\begin{equation}\label{eq:ppff-dim}
d_j:=\dim X_j
=
\binom{n}{j}^2-\binom{n}{j-1}^2,
\end{equation}
with the convention $\binom{n}{-1}=0$.
Next, the zero-weight multiplicity can be read off from Schur functions. The coefficient of the monomial $x_1x_2\cdots x_n$ in the Schur polynomial $s_{(2^j,1^{n-2j})}$ is simply the Kostka number $K_{(2^j,1^{n-2j}),(1^n)},$ which counts semistandard tableaux of shape $(2^j,1^{n-2j})$ and content $(1^n)$. Since each symbol $1,\dots,n$ appears exactly once, such tableaux are automatically standard. Hence the zero-weight multiplicity is the number of standard Young tableaux of shape $(2^j,1^{n-2j})$, which by the hook-length formula equals
\begin{equation}\label{eq:ppff-zerowt}
d_j^0:=\dim (X_j)^T
=
\binom{n}{j}-\binom{n}{j-1}.
\end{equation}
Combining~\eqref{eq:ppff-dim} and~\eqref{eq:ppff-zerowt} with Theorem~\ref{thm:mc}, we find that the measurement channel acts on every copy of $X_j$ by the scalar
\begin{equation}\label{eq:ppff-aj}
a_j^H
=
\frac{d_j^0}{d_j}
=
\frac{1}{\binom{n}{j}+\binom{n}{j-1}}.
\end{equation}
Since $d_j^0>0$ for every $j$, all of the particle preserving operators are  visible. For example, consider  one-body observables.
The one-body sector $\mathrm{span}\{a_i^\dagger a_j\}_{i,j} \cong V\otimes \overline V \cong \mbc \oplus X_1$ The trivial summand is spanned by the total number operator $N,$
which therefore has channel eigenvalue $1$. Meanwhile, the traceless one-body observables transform in the $X_1$ summand and have channel eigenvalue $a_1=1/(n+1).$ 

Finally, let us briefly discuss the variance bounds.
If an observable $O$ is supported on a single copy of $X_j$, then Theorem~\ref{thm:var} and the single-isotypic estimate give
\begin{equation}
{\rm Var}[\hat o]\le \frac{\|O\|_2^2}{a_j}
=
\bigg(\binom{n}{j}+\binom{n}{j-1}\bigg)\|O\|_2^2,
\end{equation}
and
\begin{equation}
{\rm Var}[\hat o]\le \frac{\|O\|_\infty^2}{a_j^2}
=
\bigg(\binom{n}{j}+\binom{n}{j-1}\bigg)^2\|O\|_\infty^2.
\end{equation}
In the full free-fermion protocol of Sec.~\ref{sec:ff}, on the other hand, the degree-$2j$ Majorana sector has channel eigenvalue
\begin{equation}
a_{2j}^{{\rm FGU}(2n)}=\frac{\binom{n}{j}}{\binom{2n}{2j}}.
\end{equation}
For number-conserving observables, the present $G_{\mathrm{pp}}$-protocol is therefore always at least as good, since
\begin{equation}
a_{j;\,{G_{\mathrm{pp}}}}
=
\frac{1}{\binom{n}{j}+\binom{n}{j-1}}
\;\ge\;
\frac{\binom{n}{j}}{\binom{2n}{2j}}
=
a_{2j;\,{\rm FGU}(2n)},
\end{equation}
as can be readily verified by direct calculation.
Thus the particle-preserving protocol gives smaller variance bounds on number-conserving observables, at the price of making particle-number--changing observables invisible.

\begin{figure}[t]
    \centering
    \includegraphics[width=0.95\linewidth]{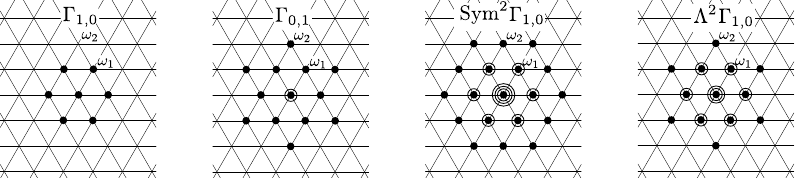}
    \caption{Weight-diagrams of $\mfg_2$ representations that appear in Appendix~\ref{sec:eshad}. The Cartan subalgebra of $\mfg_2$ being two-dimensional allows us to easily draw the diagrams and reason geometrically about highest-weights. For example, we see immediately that $\Gamma_{2,0}\subset{\rm Sym}^2\hspace{0.5mm}\Gamma_{1,0}$ and $\Gamma_{0,1}\subset\Lambda^2\hspace{0.5mm}\Gamma_{1,0}$; removing the weights of $\Gamma_{0,1}$ from the diagram of $\Lambda^2\hspace{0.5mm}\Gamma_{1,0}$ subsequently reveals that the complementary factor is $\Gamma_{1,0}$. In all cases, the dimensions of both the representations and their weight-zero subspaces can be read directly from the diagrams. }
    \label{fig:g2}
\end{figure}

\begin{table}
\centering
\begin{tabular}{c c c}
\hline
\ \ Simple type\ \  & Minuscule highest weights & Corresponding irreducible module(s) \\
\hline
$\mfsl_{n+1}$ &
$\omega_1,\omega_2,\ldots,\omega_n$ &
$\Lambda^k \mathbb{C}^{n+1}$ for $1\leq k\leq n$ \\

$\mfso_{2n+1}$ &
$\omega_n$ &
spin module \\

$\mfsp_{2n}$ &
$\omega_1$ &
defining module $\mathbb{C}^{2n}$ \\

$\mfso_{2n}$ &
$\omega_1,\omega_{n-1},\omega_n$ &
vector module $\mathbb{C}^{2n}$ and the two spinor modules \\

$\mfe_6$ &
$\omega_1,\omega_6$ &
the two dual $27$-dimensional modules \\

$\mfe_7$ &
$\omega_7$ &
the $56$-dimensional module \\

$\mfe_8$ &
none &
none \\

$\mff_4$ &
none &
none \\

$\mfg_2$ &
none &
none \\
\hline
\end{tabular}
\caption{Classification of the nontrivial minuscule irreducible representations of simple Lie algebras~\cite{li2018inner}. }
\label{tab:minuscule}
\end{table}

\subsection{Minuscule Lie group modules}\label{sec:eshad}
As our final (set of) examples we briefly make some general remarks on the shadows protocols induced by a Lie group $G$ acting on a \textit{minuscule} irreducible weight module, with measurements in a weight basis. We recall that a module is said to be minuscule if its weights form a single orbit under the action of the Weyl group~\cite{fulton1991representation}. As the highest weight space of an irrep is one-dimensional, the transitivity of this action immediately implies that all of the other weight spaces appearing in the module are also one-dimensional. This, combined with the assumed irreducibility of the module (so that any basis is a FB) implies that a weight basis of the module is an NDCSE. As evidenced by Table~\ref{tab:minuscule}, this immediately gives us quite a few examples of representations with centralizing bases. In the case of the classical groups we are simply rediscovering the centralizing nature of the weight bases that we have seen above, but the cases of the irreps of the exceptional Lie groups $E_6$ and $E_7$ mentioned in the table are new. \\

We emphasise, however, that while an irrep being minuscule is a sufficient condition for us to obtain a centralizing basis, it can also be possible to find centralizing basis when working with a non-minuscule irrep. To see this, let us turn to 
the exceptional Lie group $G_2$. The simplest of the  exceptional Lie groups, $G_2$ is the automorphism group of the octonions~\cite{fulton1991representation} and occurs in physics in such varied settings as string theory~\cite{atiyah2001m} and the symmetries of rolling balls~\cite{baez2014g}. For our purposes the first important fact is that it admits an embedding $G_2\hookrightarrow\SO(7)$ onto a compact subgroup of $\SO(7)$, with the uniform measure over it constituting an ensemble of unitaries one can use for a shadows protocol targeting seven-dimensional quantum systems. Although we will not reproduce the construction here, we note that explicit matrices realizing this isomorphism are known~\cite{arenas2005constructing}. 
To understand the relevant representations we pass, as usual, to the Lie algebra.\\

Happily, $\mfg_2$ is of rank two (whence, indeed, the subscript of its name) so we can readily visualize the representations by means of weight-diagrams (see Fig.~\ref{fig:g2}) and characterize the irreps by a pair of non-negative integers $a$ and $b$, using the notation $\Gamma_{a,b}$ to denote the irrep of highest weight $a\omega_1+b\omega_2$. Omitting the details (which can be found in, for example, Ref.~\cite{fulton1991representation}) the fundamental representation $\mch$ (acting on our $7-$dimensional system) is irreducible of highest weight $\omega_1$, i.e., in our notation $\mch=\Gamma_{1,0}$. We will take measurements in the weight basis. Now, from the diagram (Fig.~\ref{fig:g2}) we see  that $\mch $ is \textit{not} minuscule (as witnessed by the presence of a zero weight), but that it is degeneracy-free and self-dual; combined with its irreducibility we have
\begin{equation}
\mcl^\DC=\mcl=\mch\otimes\mch^*=\Gamma_{1,0}\otimes\Gamma_{1,0}={\rm Sym}^2\hspace{0.5mm}\Gamma_{1,0}\hspace{0.7mm} \oplus\hspace{0.7mm}\Lambda^2\hspace{0.5mm}\Gamma_{1,0}\,.
\end{equation}
Next up, from Fig.~\ref{fig:g2} we see that ${\rm Sym}^2\hspace{0.5mm}\Gamma_{1,0}$ contains a copy of $\Gamma_{2,0}$; though it is not immediate, it turns out that, up to a copy of the trivial representation $\mbc$, this is it: ${\rm Sym}^2\hspace{0.5mm}\Gamma_{1,0}=\Gamma_{2,0}\oplus\mbc$~\cite{fulton1991representation}. Finally, we can understand the case of $\Lambda^2\hspace{0.5mm}\Gamma_{1,0}$  from Fig.~\ref{fig:g2}: it certainly contains a copy of $\Gamma_{0,1}$, and examining the remaining weights reveals that the other factor is $\Gamma_{1,0}$. In totality, then, we have
\begin{equation}
\mcl^\DC=\Gamma_{2,0}\hspace{0.7mm} \oplus\hspace{0.7mm}\Gamma_{1,0}\hspace{0.7mm} \oplus\hspace{0.7mm}\Gamma_{0,1}\hspace{0.7mm} \oplus\hspace{0.7mm}\mbc\,.
\end{equation}
As in the previous appendices we are at this point in business, and from Fig.~\ref{fig:g2} can simply read off
\begin{equation}
a_{\Gamma_{2,0}}=\frac{3}{27}=\frac{1}{9},\ a_{\Gamma_{1,0}}=\frac{1}{7},\ a_{\Gamma_{0,1}}=\frac{2}{14}=\frac{1}{7},\ a_\mbc=1;
\end{equation}
by Theorem~\ref{thm:mc} we immediately understand the $G_2$-shadows measurement channel. 

\section{Exact variance calculations}\label{sec:exact}
When the decomposition of the inverse of the measurement channel on the target observable, $\mcm^{-1}(O)$, with respect to a basis of weight vectors of the action of $G$ on ${\rm End\ }\mch$ is known, it can be tractable to analytically obtain the variance of the shadows estimator,
\begin{equation}\label{eq:vexact}
\mathrm{Var}[\hat{o}]=\sum_{\eta,i,\alpha}\sdbra{\rho\otimes \mcm^{-1}(O)\otimes \mcm^{-1}(O)}\mce^{(3)}\sdket{\Pi_{\eta,i,\alpha}^{\otimes 3}}-\left(\sum_{\eta,i,\alpha}\sdbra{\rho\otimes \mcm^{-1}(O)}\mce^{(2)}\sdket{\Pi_{\eta,i,\alpha}^{\otimes 2}}\right)^2,
\end{equation}
obviating the need to appeal to the (often loose) variance bounds derived above. We will give several examples of this in the case $G=\SU(2)$ (acting via the tensor rep), the permutation symmetric action on $n$ qubits.\\

As $\mce^{(2)}$ and $ \mce^{(3)}$ project respectively onto the trivial $G$-irreps of $\mcl^{\otimes 2}$ and $ \mcl^{\otimes 3}$,
the key step  is to determine the overlap of those subspaces with the other operators present in Eq.~\eqref{eq:vexact}. Our strategy for doing this is to iteratively construct weight vectors of $\mcl^{\otimes k}$ from $ \mcl^{\otimes (k-1)}$ via Clebsch–Gordan expansions, where we recall the definition
\begin{equation}
C^{s_1,s_2,s_{12}}_{m_1,m_2,m_{12}}=\braket{s_1,m_1;s_2,m_2}{s_1,s_2;s_{12},m_{12}}
\end{equation}
of the Clebsch–Gordan coefficients, relating the decompositions of the tensor basis $\ket{s_1,m_1}\otimes \ket{s_2,m_2}$ of $\HC^{s_1}\otimes\HC^{s_2} \cong \bigoplus_{s_{12}} \HC^{(2)}_{s_{12}}$ to the ``coupled'' basis $\ket{s_{12},m_{12}}$ of each irrep  $\HC^{(2)}_{s_{12}}$  . Beginning with $\Pi_{\eta,i,\alpha}$, from Eq.~\eqref{eq:obasis} we have:
\begin{equation}
\Pi_{\eta,i,\alpha}=\sum_{\mu}(-1)^{\alpha+s_{\eta}}C^{s_{\eta},s_{\eta},\mu}_{\alpha,-\alpha,0}O^{\eta,i}_{\mu,0}.
\end{equation}
Taking the tensor square of this expression and introducing the ``coupled operator basis elements'' defined implicitly by 
\begin{equation}\label{eq:c2}
O^{{\eta},i}_{\mu,\nu}\otimes O^{{\eta},i}_{\mu',\nu'}=\sum_{\mu_{12}}C^{\mu,\mu',\mu_{12}}_{\nu,\nu',\nu+\nu'}O^{(2),{\eta},i}_{[\mu,\mu'];\mu_{12},\nu+\nu'},
\end{equation}
we have
\begin{align}
\Pi_{\eta,i,\alpha}^{\otimes 2}&=\sum_{\mu,\mu'}C^{s_{\eta},s_{\eta},\mu}_{\alpha,-\alpha,0}C^{s_{\eta},s_{\eta},\mu'}_{\alpha,-\alpha,0}O^{{\eta},i}_{\mu,0}\otimes O^{{\eta},i}_{\mu',0}=\sum_{\mu,\mu',\mu_{12}}C^{s_{\eta},s_{\eta},\mu}_{\alpha,-\alpha,0}C^{s_{\eta},s_{\eta},\mu'}_{\alpha,-\alpha,0}C^{\mu,\mu',\mu_{12}}_{0,0,0}O^{(2),{\eta},i}_{[\mu,\mu'];\mu_{12},0}
\end{align}
Tensoring our expressions for $\Pi_{\eta,i,\alpha}$ and $\Pi_{\eta,i,\alpha}^{\otimes 2}$ then similarly yields
\begin{equation}\label{eq:pi3}
\Pi_{\eta,i,\alpha}^{\otimes 3}=\sum_{\mu,\mu',\mu'',\mu_{12},\mu_{123}}(-1)^{\alpha+s_{\eta}}C^{s_{\eta},s_{\eta},\mu}_{\alpha,-\alpha,0}C^{s_{\eta},s_{\eta},\mu'}_{\alpha,-\alpha,0}C^{s_{\eta},s_{\eta},\mu''}_{\alpha,-\alpha,0}C^{\mu,\mu',\mu_{12}}_{0,0,0}C^{\mu''\mu_{12},\mu_{123}}_{0,0,0}O^{(3),{\eta},i}_{[\mu''[\mu,\mu']_{\mu_{12}}];\mu_{123},0}
\end{equation}
with the third order coupled operator basis elements defined by the relation
\begin{equation}\label{eq:c3}
O^{{\eta},i}_{\mu'',\nu''}\otimes O^{(2),{\eta},i}_{[\mu,\mu'];\mu_{12},\nu_{12}}=\sum_{\mu_{123}}C^{\mu''\mu_{12},\mu_{123}}_{\nu'',\nu_{12},\nu''+\nu_{12}}O^{(3),{\eta},i}_{[\mu''[\mu,\mu']_{\mu_{12}}];\mu_{123},\nu''+\nu_{12}}\,.
\end{equation}
Let us also write
\begin{equation}
\rho:=\sum_{{\eta},i,s'',m''}\rho^{{\eta},i}_{s'',m''} O^{{\eta},i}_{s'',m''} + \rho_{\rm invisible},
\end{equation}
where the ``invisible'' component of $\rho$ (i.e., the component that is not $G$-diagonal) will not affect the  result of the shadow tomography, and is henceforth ignored. 
Before moving to the calculation of the variance for some choices of target observable $O$ it is perhaps helpful to have in mind a simple example of an explicit realization of the operators $O^{{\eta},i}_{\mu,\nu}$ in the Schur basis. Taking for example $n=2$, $\mch$ decomposes into irreps as 
$(\mathbf{1/2})\otimes(\mathbf{1/2})=\mathbf{0}\oplus\mathbf{1}$, and the corresponding $G$-diagonal operators decompose as 
\begin{align}
{\rm End}(\mathbf{0})\oplus {\rm End}(\mathbf{1})&=(\mathbf{0})\oplus(\mathbf{0}\oplus\mathbf{1}\oplus\mathbf{2})\\
&={\rm span}\{ (O^{0,0}_{0,0}),\ ((O^{1,0}_{0,0}),\ (O^{1,0}_{1,0},O^{1,0}_{1,-1},O^{1,0}_{1,1}),\ (O^{1,0}_{2,-2},O^{1,0}_{2,-1},O^{1,0}_{2,0},O^{1,0}_{2,1},O^{1,0}_{2,2})) \}\,,
\end{align}
where in the Schur basis we have
\begin{equation*}
O^{0,0}_{0,0} = \begin{pmatrix}
    0 & 0 & 0 & 0\\
    0 & 0 & 0 & 0\\
    0 & 0 & 0 & 0\\
    0 & 0 & 0 & 1
\end{pmatrix};\quad
O^{1,0}_{0,0} = \begin{pmatrix}
    \frac{1}{\sqrt{3}} & 0 & 0 & 0\\
    0 & \frac{1}{\sqrt{3}} & 0 & 0\\
    0 & 0 & \frac{1}{\sqrt{3}} & 0\\
    0 & 0 & 0 & 0
\end{pmatrix};\quad
O^{1,0}_{1,-1} = \begin{pmatrix}
    0 & -\frac{1}{\sqrt{2}} & 0 & 0\\
    0 & 0 & -\frac{1}{\sqrt{2}} & 0\\
    0 & 0 & 0 & 0\\
    0 & 0 & 0 & 0
\end{pmatrix};\quad
O^{1,0}_{1,0} = \begin{pmatrix}
    -\frac{1}{\sqrt{2}} & 0 & 0 & 0\\
    0 & 0 & 0 & 0\\
    0 & 0 & \frac{1}{\sqrt{2}} & 0\\
    0 & 0 & 0 & 0
\end{pmatrix}
\end{equation*}

\begin{equation*}
O^{1,0}_{1,1} = \begin{pmatrix}
    0 & 0 & 0 & 0\\
    \frac{1}{\sqrt{2}} & 0 & 0 & 0\\
    0 & \frac{1}{\sqrt{2}} & 0 & 0\\
    0 & 0 & 0 & 0
\end{pmatrix};\quad  
O^{1,0}_{2,-2} = \begin{pmatrix}
    0 & 0 & 1 & 0\\
    0 & 0 & 0 & 0\\
    0 & 0 & 0 & 0\\
    0 & 0 & 0 & 0
\end{pmatrix};\quad  
O^{1,0}_{2,-1} = \begin{pmatrix}
    0 & \frac{1}{\sqrt{2}} & 0 & 0\\
    0 & 0 & -\frac{1}{\sqrt{2}} & 0\\
    0 & 0 & 0 & 0\\
    0 & 0 & 0 & 0
\end{pmatrix};\quad  
O^{1,0}_{2,0} = \begin{pmatrix}
    \frac{1}{\sqrt{6}} & 0 & 0 & 0\\
    0 & -\frac{2}{\sqrt{6}} & 0 & 0\\
    0 & 0 & \frac{1}{\sqrt{6}} & 0\\
    0 & 0 & 0 & 0
\end{pmatrix}
\end{equation*}
\begin{equation*}
O^{1,0}_{2,1} = \begin{pmatrix}
    0 & 0 & 0 & 0\\
    -\frac{1}{\sqrt{2}} & 0 & 0 & 0\\
    0 & \frac{1}{\sqrt{2}} & 0 & 0\\
    0 & 0 & 0 & 0
\end{pmatrix};\quad  
O^{1,0}_{2,2} = \begin{pmatrix}
    0 & 0 & 0 & 0\\
    0& 0 & 0 & 0\\
    1 & 0 & 0 & 0\\
    0 & 0 & 0 & 0
\end{pmatrix}.
\end{equation*}
Here we have labelled $\eta$ by the spin of the corresponding $\mch$ irrep, and indexed the multiplicities $i$ from 0.
Note that in this basis $O^{{\eta},i}_{\mu,\nu}$ is diagonal if and only if  $\nu=0$. The visible space has dimension $3^2+1^2=10<16$; one could extend the above list to a basis of ${\rm End}((\mbc^2)^{\otimes 2})$ by (for example) adding the six elementary matrices given in the Schur basis by $\{\ketbra{0,0,0}{1,0,i}, \ketbra{1,0,i}{0,0,0}\}_{i=-1,0,1}$, which are not $G$-diagonal, mapping as they do between $G$-irreps of $\mch$. \\

\noindent
We now detail two examples of  explicit variance calculations.

\subsection{A symmetrized \textit{Z} operator}
\noindent
We begin by recalling our expression Eq.~\eqref{eq:zsymd2} for $Z_{\rm sym}=\sum_i Z_i$, restated here for convenience:
\begin{equation}\label{eq:zsymd3}
{Z}_{\rm sym}=-\sum_{{\eta},i}\sqrt{2\binom{d_{\eta} + 1}{3}} O^{{\eta},i}_{1,0}=:-\sum_{{\eta},i}z_{\eta} O^{{\eta},i}_{1,0}\,.
\end{equation}
As ${Z}_{\rm sym}$ is therefore entirely contained within ``spin-1'' irreps of $\mcl^\DC$ we see $\mcm^{-1}({Z}_{\rm sym})=a_{{\rm spin-}1}^{-1}{Z}_{\rm sym}=3{Z}_{\rm sym}$. Beginning with the first term on the right hand side of Eq.~\eqref{eq:vexact} we have:
\begin{align}
\mbe[\hat{z}_{\rm sym}^2]  &= \sum_{\eta,i,\alpha}\sdbra{\rho\otimes \mcm^{-1}({Z}_{\rm sym})\otimes \mcm^{-1}({Z}_{\rm sym})}\mce^{(3)}\sdket{\Pi_{\eta,i,\alpha}^{\otimes 3}}\\
&={9}\sum_{\eta,i,\alpha}\sdbra{\rho\otimes {Z}_{\rm sym}\otimes {Z}_{\rm sym}}\mce^{(3)}\sdket{\Pi_{\eta,i,\alpha}^{\otimes 3}}\\
&={9}\sum_{\eta,i,\alpha,s'',m''}z_{\eta}^2 \rho^{{\eta},i}_{s'',m''}\sdbra{O^{{\eta},i}_{s'',m''}\otimes O^{{\eta},i}_{1,0}\otimes O^{{\eta},i}_{1,0}}\mce^{(3)}\sdket{\Pi_{\eta,i,\alpha}^{\otimes 3}}\label{eq:zins}\\
&={9}\sum_{\eta,i,\alpha,s,s'',m''}z_{\eta}^2 \rho^{{\eta},i}_{s'',m''}C^{1,1,s}_{0,0,0}\sdbra{O^{{\eta},i}_{s'',m''}\otimes O^{(2),{\eta},i}_{[1,1];s,0}}\mce^{(3)}\sdket{\Pi_{\eta,i,\alpha}^{\otimes 3}}\\
&={9}\sum_{\eta,i,\alpha,s}z_{\eta}^2 \rho^{{\eta},i}_{s,0}C^{1,1,s}_{0,0,0}C^{s,s,0}_{0,0,0}\sdbra{O^{(3),{\eta},i}_{[s[1,1]_s];0,0}}\mce^{(3)}\sdket{\Pi_{\eta,i,\alpha}^{\otimes 3}}\label{eq:lo}\,,
\end{align}
where (for example) $O^{(3),{\eta},i}_{[s[1,1]];0,0}\in {\rm End\ }\mch^{\otimes 3}$ is an element of the ``coupled'' basis with total and $z$ spin quantum numbers $s_{123}=m_{123}=0$ (see Eqs.~\eqref{eq:c2} and~\eqref{eq:c3}). In Eq.~\eqref{eq:zins} we inserted our expression Eq.~\eqref{eq:zsymd3} for $Z_{\rm sym}$. To obtain the last line we used (recalling Eq.~\eqref{eq:c3}) the following calculation:
\begin{align*}
\sdbra{O^{{\eta},i}_{s'',m''}\otimes O^{(2),{\eta},i}_{[1,1];s,0}}\mce^{(3)}&= \sum_{s'''}C^{s,s'',s'''}_{m'',0,m''}\sdbra{O^{(3),{\eta},i}_{[s''[1,1]_s];s''',m''}}\mce^{(3)}\\
&= \sum_{s'''}\delta_{s''',0}C^{s,s'',s'''}_{m'',0,m''}\sdbra{O^{(3),{\eta},i}_{[s''[1,1]_s];s''',m''}}\mce^{(3)}\\
&= C^{s,s'',0}_{m'',0,m''}\sdbra{O^{(3),{\eta},i}_{[s''[1,1]_s];0,m''}}\mce^{(3)}\\
&= \delta_{s,s''}\delta_{m'',0} C^{s,s,0}_{0,0,0}\sdbra{O^{(3),{\eta},i}_{[s[1,1]_s];0,0}}\mce^{(3)}\,.
\end{align*}
Where we have noticed that the projector $\mce^{(3)}$ kills operators living in non-trivial irreps (i.e., forcing $s'''=0$). In turn, this requires $s=s''$ and $m''=0$.
Vectorizing and inserting our expression Eq.~\eqref{eq:pi3} for $\Pi_{\eta,i,\alpha}^{\otimes 3}$ into Eq.~\eqref{eq:lo} then gives
\begin{align*}
\mbe[\hat{z}_{\rm sym}^2]&={9}\sum_{\eta,i,\alpha,s}z_{\eta}^2 \rho^{{\eta},i}_{s,0}C^{1,1,s}_{0,0,0}C^{s,s,0}_{0,0,0}\sdbra{O^{(3),{\eta},i}_{[s[1,1]_s];0,0}}\mce^{(3)}|\\
&\hspace{20mm}\sdket{\sum_{\mu,\mu',\mu'',\mu_{12},\mu_{123}}(-1)^{\alpha+s_{\eta}}C^{s_{\eta},s_{\eta},\mu}_{\alpha,-\alpha,0}C^{s_{\eta},s_{\eta},\mu'}_{\alpha,-\alpha,0}C^{s_{\eta},s_{\eta},\mu''}_{\alpha,-\alpha,0}C^{\mu,\mu',\mu_{12}}_{0,0,0}C^{\mu''\mu_{12},\mu_{123}}_{0,0,0}O^{(3),{\eta},i}_{[\mu''[\mu,\mu']_{\mu_{12}}];\mu_{123},0}}\\
&={9}\sum_{\eta,i,\alpha,s}\sum_{\mu,\mu',\mu'',\mu_{12},\mu_{123}}(-1)^{\alpha+s_{\eta}}C^{s_{\eta},s_{\eta},\mu}_{\alpha,-\alpha,0}C^{s_{\eta},s_{\eta},\mu'}_{\alpha,-\alpha,0}C^{s_{\eta},s_{\eta},\mu''}_{\alpha,-\alpha,0}C^{\mu,\mu',\mu_{12}}_{0,0,0}C^{\mu''\mu_{12},\mu_{123}}_{0,0,0}z_{\eta}^2 \rho^{{\eta},i}_{s,0}C^{1,1,s}_{0,0,0}C^{s,s,0}_{0,0,0}\ \ \times\\
&\hspace{40mm}\sdbra{O^{(3),{\eta},i}_{[s[1,1]_s];0,0}}\mce^{(3)}\sdket{O^{(3),{\eta},i}_{[\mu''[\mu,\mu']_{\mu_{12}}];\mu_{123},0}}\,.
\end{align*}
Now, $\sdbra{O^{(3),{\eta},i}_{[s[1,1]_s];0,0}}\mce^{(3)}\sdket{O^{(3),{\eta},i}_{[\mu''[\mu,\mu']_{\mu_{12}}];\mu_{123},0}}=\delta_{s,\mu''}\delta_{1,\mu}\delta_{1,\mu'}\delta_{0,\mu_{123}}\delta_{s,\mu_{12}}$, so the sum simplifies quite significantly:

\begin{align*}
\mbe[\hat{z}_{\rm sym}^2]&={9}\sum_{\eta,i,\alpha,s}z_{\eta}^2 \rho^{{\eta},i}_{s,0}C^{1,1,s}_{0,0,0}C^{s,s,0}_{0,0,0}(-1)^{\alpha+s_{\eta}}C^{s_{\eta},s_{\eta},1}_{\alpha,-\alpha,0}C^{s_{\eta},s_{\eta},1}_{\alpha,-\alpha,0}C^{s_{\eta},s_{\eta},s}_{\alpha,-\alpha,0}C^{1,1,s}_{0,0,0}C^{s,s,0}_{0,0,0}\\
&={9}\sum_{\eta,i,\alpha,s}z_{\eta}^2 \rho^{{\eta},i}_{s,0}(-1)^{\alpha+s_{\eta}}C^{s_{\eta},s_{\eta},s}_{\alpha,-\alpha,0}\left(C^{1,1,s}_{0,0,0}C^{s,s,0}_{0,0,0}C^{s_{\eta},s_{\eta},1}_{\alpha,-\alpha,0}\right)^2\\
&={18}\sum_{\eta,i,\alpha,s}\binom{d_{\eta} + 1}{3}\rho^{{\eta},i}_{s,0}(-1)^{\alpha+s_{\eta}}C^{s_{\eta},s_{\eta},s}_{\alpha,-\alpha,0}\left(C^{1,1,s}_{0,0,0}C^{s,s,0}_{0,0,0}C^{s_{\eta},s_{\eta},1}_{\alpha,-\alpha,0}\right)^2\,.
\end{align*}
We now turn to the second term on the right hand side of Eq.~\eqref{eq:vexact}, where we have a similar calculation: 
\begin{align}
\mbe[\hat{z}_{\rm sym}]&=\sum_{\eta,i,\alpha}\sdbra{\rho\otimes  \mcm^{-1}({Z}_{\rm sym})}\mce^{(2)}\sdket{\Pi_{\eta,i,\alpha}^{\otimes 2}}\\
&={3}\sum_{\eta,i,\alpha}\sdbra{\rho\otimes {Z}_{\rm sym}}\mce^{(2)}\sdket{\Pi_{\eta,i,\alpha}^{\otimes 2}}\\
&={3}\sum_{\eta,i,\alpha,s'',m''}\rho^{{\eta},i}_{s'',m''}z_{\eta} \sdbra{O^{{\eta},i}_{s'',m''}\otimes O^{{\eta},i}_{1,0}}\mce^{(2)}\sdket{\sum_{\mu,\mu',\mu_{12}}C^{s_{\eta},s_{\eta},\mu}_{\alpha,-\alpha,0}C^{s_{\eta},s_{\eta},\mu'}_{\alpha,-\alpha,0}C^{\mu,\mu',\mu_{12}}_{0,0,0}O^{(2),{\eta},i}_{\mu,\mu';\mu_{12},0}}\\
&={3}\sum_{\eta,i,\alpha,s'',\mu,\mu'}\rho^{{\eta},i}_{s'',0}z_{\eta} C^{s_{\eta},s_{\eta},\mu}_{\alpha,-\alpha,0}C^{s_{\eta},s_{\eta},\mu'}_{\alpha,-\alpha,0}C^{\mu,\mu',0}_{0,0,0}\sdbra{C^{s'',1,0}_{0,0,0}O^{(2),{\eta},i}_{s'',1;0,0}}\mce^{(2)}\sdket{O^{(2),{\eta},i}_{\mu,\mu';0,0}}\\
&={3}\sum_{\eta,i,\alpha}\rho^{{\eta},i}_{1,0}z_{\eta} C^{s_{\eta},s_{\eta},1}_{\alpha,-\alpha,0}C^{s_{\eta},s_{\eta},1}_{\alpha,-\alpha,0}C^{1,1,0}_{0,0,0}C^{1,1,0}_{0,0,0}\\
&={3}\sum_{{\eta},i}\rho^{{\eta},i}_{1,0}\sqrt{2\binom{d_{\eta} + 1}{3}} \left(C^{1,1,0}_{0,0,0}\right)^2\sum_{\alpha}\left(C^{s_{\eta},s_{\eta},1}_{\alpha,-\alpha,0}\right)^2\\
&=\sum_{{\eta},i}\rho^{{\eta},i}_{1,0}\sqrt{2\binom{d_{\eta} + 1}{3}}\,,
\end{align}
where we have used $C^{1,1,0}_{0,0,0}=-1/\sqrt{3}$ and $\sum_{\alpha}\left(C^{s_{\eta},s_{\eta},1}_{\alpha,-\alpha,0}\right)^2=1$. 
The variance is therefore given by
\begin{equation}
\mathrm{Var}[\hat{z}_{\rm sym}] = \left[{18}\sum_{\eta,i,\alpha,s}\binom{d_{\eta} + 1}{3}\rho^{{\eta},i}_{s,0}(-1)^{\alpha+s_{\eta}}C^{s_{\eta},s_{\eta},s}_{\alpha,-\alpha,0}\left(C^{1,1,s}_{0,0,0}C^{s,s,0}_{0,0,0}C^{s_{\eta},s_{\eta},1}_{\alpha,-\alpha,0}\right)^2\right]-\left[\sum_{{\eta},i}\rho^{{\eta},i}_{1,0}\sqrt{2\binom{d_{\eta} + 1}{3}}\right]^2
\end{equation}
when $s_{\eta}\in\mathbb{Z}$ this further simplifies to 
\begin{align}
\mathrm{Var}[\hat{z}_{\rm sym}] &= \sum_{{\eta},i}\binom{d_{\eta} + 1}{3}\left[\frac{6\rho^{{\eta},i}_{0,0}}{\sqrt{d_{\eta}}}+\frac{12\rho^{{\eta},i}_{2,0}(4s_{\eta}(1+s_{\eta})-3)}{5\sqrt{5s_{\eta}(s_{\eta}+1)(2s_{\eta}-1)(2s_{\eta}+1)(2s_{\eta}+3)}}\right]-\left[\sum_{{\eta},i}\rho^{{\eta},i}_{1,0}\sqrt{2\binom{d_{\eta} + 1}{3}}\right]^2
\end{align}
(note $C^{1,1,s}_{0,0,0}=0$ for $s\neq 0, 2$).

\subsection{A projector onto an isotypic}
One can in principle approach this calculation in a conceptually identical manner to the previous one; in this instance, however, we do not require the full Clebsch-Gordan machinery. Indeed, let us note that the single-shot estimator corresponding to a projector $\Pi^{(\eta)}$ is given by
\begin{equation}
    \hat{\pi}_{w,U}^{(\eta)} = \Tr[\mcm^{-1}((U\tn)\ad\Pi_w U\tn)\Pi^{(\eta)}] = \Tr[(U\tn)\ad\Pi_w U\tn\mcm^{-1}(\Pi^{(\eta)})]= \Tr[(U\tn)\ad\Pi_w U\tn \Pi^{(\eta)}]= \Tr[\Pi_w  \Pi^{(\eta)}] 
\end{equation}
where we have used that $[\Pi^{(\eta)},U\tn]=0$. But, writing $\ket w=\ket{\eta',i',\a'}$, this final expression is simply $\delta_{\eta,\eta'}$. It follows that $\hat{\pi}^{(\eta)}$ is just an indicator variable for whether or not measuring $\rho$ in the Schur basis yields (in the irrep label) the value $\eta$, which  happens with probability $\Tr [\rho\Pi^{(\eta)}]$. As (like any random variable that takes only the values zero and one) we have $\expect[\hat{\pi}^{(\eta)}] = \expect [(\hat{\pi}^{(\eta)})^2] $, the variance is therefore simply given by
\begin{equation}
\mathrm{Var}[\hat{\pi}^{(\eta)}] = \Tr [\rho\Pi^{(\eta)}] \left(1-\Tr [\rho\Pi^{(\eta)}]\right) \leq \frac{1}{4}\,,
\end{equation}
where we have used  $x(1-x)\leq 1/4$ for any $x\in\mbr$. So one improves by a factor  of four upon the generic bound of Theorem~\ref{thm:var}.

\section{Tight frame formalism}\label{sec:tf}
We now describe what one might think of as a generalization of the \textit{tight frame} formalism introduced in Ref.~\cite{zhao2021fermionic}, although we will not feel the need to explicitly invoke such language, nor the machinery of tight frames in general.  We begin with the following result:
\begin{restatable}{lemma}{lemorb}\label{lem:orbit}
(A minor generalization of Theorem 11 of Ref.~\cite{zhao2021fermionic}).
Suppose $\mch$ possesses an NDCSE $\mcw$ with respect to a subgroup $H$ of $G$, and a Hilbert–Schmidt orthonormal basis $\{e_\sigma\}_\sigma$ of the irrep $({\rm End\ }\mch^{\eta,i})^{\lm,j}$ whose elements are either diagonal or entirely off-diagonal with respect to $\mcw$. Denote by $\lm$ the corresponding adjoint representation, i.e., 
\begin{equation}
    U_g e_\sigma U_g^\dagger = \sum_\tau \lm(g)_{\sigma,\tau}e_\tau\,.
\end{equation}
Then with $D\subset \{e_\sg\}_\sg$ denoting the set of diagonal basis elements we have:
\begin{equation}
\expect_{g\in G} \sum_{\sg\in D} |\lm(g)_{\sg,\tau}|^2=a_\lm^H,
\end{equation}
where, as before, $a_\lm^H=d^H_\lm/d_\lm$.
\end{restatable}
\begin{proof}
Consider the map $P_D:\lm\to\lm,\ P_D = \sum_{\sg\in D} \ketbra{e_\sigma} $. As $\lm $ is irreducible, by Schur's lemma we have
\begin{equation}
    \expect_{g\in G} [\lm(g)\ad P_D\lm(g)]  = \frac{\Tr[P_D]}{d_\lm}\id = \frac{d_\lm^H}{d_\lm}\id = a_\lm^H \id.
\end{equation}
Taking the $(\tau,\tau)$-matrix element then exactly gives
\begin{equation}
    a_\lm^H = \sbraket{e_\tau| \expect_{g\in G} [\lm(g)\ad P_D\lm(g)]}{e_\tau} = \expect_{g\in G} \sum_{\sg\in D}\sbraket{e_\tau|  \lm(g)\ad \ketbra{e_\sg}\lm(g)}{e_\tau} = \expect_{g\in G} \sum_{\sg\in D} |\lm(g)_{\sg,\tau}|^2.
\end{equation}
\end{proof}

With the above lemma we are ready to prove Theorem~\ref{thm:varf}, concerning special cases where we can improve over the general bounds of Theorem~\ref{thm:var}:
\thmvarf*

\begin{proof}
\textit{(i)}. In fact, when the first condition holds the result follows  quickly, with no need to invoke the result of Lemma~\ref{lem:orbit}.  Suppose first that $O$ is positive semidefinite. Then we have:
\begin{align}
\mathrm{Var}[\hat{o}]&\leq \mbe[\hat{o}^2]\\
&=\sum_w\expect_{U\sim G} \  \Tr[\rho U^\dagger \ketbra{w}  U]  \Tr[  \mcm^{-1}\left(O\right)  U^\dagger \ketbra{w} U] ^{ 2}\\
&=\sum_w\expect_{U\sim G} \  \Tr[\rho U^\dagger \ketbra{w}  U]  \Tr[ \frac{O}{a_\lm^H}  U^\dagger \ketbra{w} U] ^{ 2}\\
&\leq \frac{\| O\|_\infty}{(a_\lm^H)^2} \sum_w\expect_{U\sim G} \  \Tr[\rho U^\dagger \ketbra{w}  U]  \Tr[  O U^\dagger \ketbra{w} U]\label{eq:po} \\
&= \frac{\| O\|_\infty}{(a_\lm^H)^2} \Tr[\mcm(O)\rho] \\
&= \frac{\| O\|_\infty}{a_\lm^H} \Tr[O\rho] \\
&\leq \frac{\| O\|_\infty^2}{a_\lm^H}\,,
\end{align}
where we have used the definition   of the measurement channel and that $\Tr[O\rho]\leq\|O\|_\infty\|\rho\|_1=\|O\|_\infty$ (H\"older's inequality). Importantly, the bound of Eq.~\eqref{eq:po} is possible as, by the positive semidefiniteness of both $\rho$ and $O$,   
\begin{equation*}
 \Tr[\rho U^\dagger \ketbra{w}  U]  \Tr[  O U^\dagger \ketbra{w} U] \geq 0 \qquad \forall U\in G.
\end{equation*}
An entirely analogous argument applies in the negative semidefinite case:
\begin{align}
\mathrm{Var}[\hat{o}]&\leq \mbe[\hat{o}^2]\\
&=\sum_w\expect_{U\sim G} \  \Tr[\rho U^\dagger \ketbra{w}  U]  \Tr[  \mcm^{-1}\left(O\right)  U^\dagger \ketbra{w} U] ^{ 2}\\
&=\sum_w\expect_{U\sim G} \  \Tr[\rho U^\dagger \ketbra{w}  U]  \Tr[ \frac{O}{a_\lm^H}  U^\dagger \ketbra{w} U] ^{ 2}\\
&\leq -\frac{\| O\|_\infty}{(a_\lm^H)^2} \sum_w\expect_{U\sim G} \  \Tr[\rho U^\dagger \ketbra{w}  U]  \Tr[  O U^\dagger \ketbra{w} U] \\
&= -\frac{\| O\|_\infty}{(a_\lm^H)^2} \Tr[\mcm(O)\rho] \\
&= \frac{\| O\|_\infty}{a_\lm^H} \Tr[-O\rho] \\
&\leq \frac{\| O\|_\infty^2}{a_\lm^H}\,,
\end{align}
where one has to be a little careful with the minus signs; e.g., note that
\begin{equation*}
0\leq   \Tr[ {O}  U^\dagger \ketbra{w} U] ^{ 2} \leq -{\| O\|_\infty}  \Tr[  O U^\dagger \ketbra{w} U]
\end{equation*}
by the assumed negative-semidefiniteness of $O$.

\textit{(ii)}. For an element $O=O_\a$ of the presumed basis we have
\begin{align}
\mathrm{Var}[\hat{o}]&\leq \mbe_{U_g\sim G}\left[\sum_w \Tr[\rho U_g^\dagger \ketbra{w}  U_g]  \Tr[  \mcm^{-1}\left(O_\alpha\right)  U_g^\dagger \ketbra{w} U_g] ^{ 2}\right]\\
&=\frac{1}{(a_\lm^H)^2} \mbe_{U_g\sim G}\left[\sum_w \Tr[\rho U_g^\dagger \ketbra{w}  U_g]  \Tr[  O_\alpha  U_g^\dagger \ketbra{w} U_g] ^{ 2}\right]\\
&=\frac{1}{(a_\lm^H)^2} \mbe_{U_g\sim G}\left[\sum_w  \Tr[\rho U_g^\dagger \ketbra{w}  U_g] \bra{w} \sum_\beta \lm_{\beta,\alpha}(g) O_{\beta}   \ket{w}  ^{ 2}\right]\\
&=\frac{1}{(a_\lm^H)^2} \mbe_{g\sim G}\left[\sum_w \Tr[\rho U_g^\dagger \ketbra{w}  U_g]  \bra{w} \sum_{{\rm diagonal\ }O_\beta} \lm_{\beta,\alpha}(g) O_{\beta}   \ket{w}  ^{ 2}\right]\label{eq:wza} \\
&=\frac{1}{(a_\lm^H)^2} \mbe_{g\sim G}\left[\sum_w \Tr[\rho U_g^\dagger \ketbra{w}  U_g] \sum_{{\rm diagonal\ }O_\beta} |\lm_{\beta,\alpha}(g)|^2 \bra{w}  O_{\beta}   \ket{w}  ^{ 2}\right] \label{eq:outsum} \\
&\leq\frac{1}{(a_\lm^H)^2} \mbe_{g\sim G}\left[\sum_w \Tr[\rho U_g^\dagger \ketbra{w}  U_g] \sum_{{\rm diagonal\ }O_\beta} |\lm_{\beta,\alpha}(g)|^2  \|O\|_\infty^2 \right]\label{eq:unibound}\\
&=\frac{\|O\|_\infty^2}{(a_\lm^H)^2}  \sum_{{\rm diagonal\ }O_\beta} \mbe_{g\sim G}\left[|\lm_{\beta,\alpha}(g)|^2\right]  \\
&=\frac{\|O\|_\infty^2}{(a_\lm^H)^2}  \mbe_{g\sim G}\left[ \sum_{{\rm diagonal\ }O_\beta} |\lm_{\beta,\alpha}(g)|^2\right]  \\
&=\frac{\|O\|_\infty^2}{a_\lm^H}\,.
\end{align}
Here we have used that, by assumption, $G$ is contained within the normalizer of the set $\{O_\beta\}_\beta$, the adjoint action of $G$ simply permutes the elements of this set amongst themselves; accordingly, for a given $\alpha$, the matrix elements $\lm_{\beta,\alpha}$ can be non-zero only for a single value of $\beta$. This is employed in Eq.~\eqref{eq:outsum} to pull the sum over $\beta$ outside of the square (c.f. the proof of Lemma 6 in Ref.~\cite{zhao2021fermionic}). In Eq.~\eqref{eq:unibound} we have used the assumed uniformity of the spectral norms of the elements of the basis; finally we have recalled the result of
Lemma~\ref{lem:orbit}.
\end{proof}

\end{document}